\documentclass[a4paper,10pt]{article}
\usepackage{url}
\pdfoutput=1

\usepackage{geometry}
\usepackage{times}
\usepackage{amssymb} 
\usepackage{cite}
\usepackage[cmex10]{amsmath}
\usepackage{amsthm}
\usepackage{dsfont}
\usepackage{array}
\usepackage{fixltx2e}
\usepackage{color}
\usepackage{tikz}
\usepackage{epstopdf}
\usepackage[font={footnotesize,up}]{caption}
\usepackage[margin=2ex,font=scriptsize]{subcaption}
\usepackage{wrapfig}
\usepackage{float}

\DefineNamedColor{named}{Purple}{cmyk}{0.45,0.86,0,0}
\definecolor{green}{rgb}{0.0, 0.42, 0.24}

 \usepackage{benstyle}
\newcommand*{\rb}{\mathrm{b}}
\newcommand*{\ri}{\mathrm{i}}
\newcommand*{\rf}{\mathrm{f}}
\newcommand*{\rp}{\mathrm{p}}

\newcommand*{\rw}{\mathrm{w}}

\begin{document}

\title{Analyzing the structure of multidimensional compressed sensing problems through coherence}
\author{A. Jones\\ \small
University of Cambridge \and B. Adcock\\ \small Simon Fraser Univesity
\and
A. Hansen\\ \small
University of
Cambridge 
}

\date{}

\maketitle

\begin{abstract}
Recently it has been established that asymptotic incoherence can be used to facilitate subsampling, in order to optimize reconstruction quality, in a variety of continuous compressed sensing problems, and the coherence structure of certain one-dimensional Fourier sampling problems was determined.   This paper extends the analysis of asymptotic incoherence to cover multidimensional reconstruction problems.  It is shown that Fourier sampling and separable wavelet sparsity in any dimension can yield the same optimal asymptotic incoherence as in one dimensional case. Moreover in two dimensions the coherence structure is compatible with many standard two dimensional sampling schemes that are currently in use. However, in higher dimensional problems with poor wavelet smoothness we demonstrate that there are considerable restrictions on how one can subsample from the Fourier basis with optimal incoherence. This can be remedied by using a sufficiently smooth generating wavelet. It is also shown that using tensor bases will always provide suboptimal decay marred by problems associated with dimensionality. The impact of asymptotic incoherence on the ability to subsample is demonstrated with some simple two dimensional numerical experiments.
\end{abstract}

\section{Introduction}

Exploiting additional structure has always been central to the success of compressed sensing,  ever since it was introduced by Cand\`es, Romberg \& Tao \cite{CandesRombergTao} and Donoho \cite{donohoCS}. Sparsity and incoherence has allowed us to recover signals and images from uniformly subsampled measurements. Recently \cite{AHPRBreaking} the notions of asymptotic sparsity in levels and asymptotic incoherence were introduced to provide enough flexibility to recover signals in a larger variety of inverse problems using subsampling in levels. The key is that optimal subsampling strategies will depend both on the signal structure (asymptotic sparsity) and the asymptotic incoherence structure. 

There is a wide variety of problems that lack incoherence, a fact that has been widely recognized \cite{AHPRBreaking, discrete, VanderEtAlSpreadSpectrum, VanderEtAlVariable, ChauffertGradientwaveform, ChauffertVDS, BoyerBlockStructured, PoonFrames, Siemens, Gitta_Fourier, PoonTV, WardFourierAndPolys}, however, they instead posses asymptotic incoherence. Examples include Magnetic Resonance Imaging (MRI) \cite{Unser,Lustig3}, X-ray Computed Tomography \cite{Stanford_CT, quinto2006xrayradon}, Electron Tomography \cite{lawrence2012et,leary2013etcs}, Fluorescence microscopy \cite{Candes_PNAS, Roman} and Surface scattering \cite{JonesTamtoglHAS}, to name a few. This phenomena often originates from the inverse problems being based upon integral transforms, for example, reconstructing a function $f$ from pointwise evaluations of its Fourier transform. In compressed sensing, such a transform is combined with an appropriate sparsifying transformation associated to a basis or frame, giving rise to an infinite measurement matrix $U$. The `coherence' of $U \in \bbC^{\bbN \times \bbN} $ or $U' \in C^{N \times N}$ is defined by
\bes{
\mu(U) = \sup_{i,j \in \bbN} | U_{ij} |^2,  \qquad \mu(U') = \max_{i,j=1,\ldots,N} | U'_{ij} |^2.
}

Small coherence is refered to as `incoherence'. Asymptotic incoherence is the phenomena of when
\be{ \label{asympinco}
\mu(P^\perp_N U), \mu(U P^\perp_N) \to 0, \qquad N \to \infty,
}
where $P^\perp_N$ denotes the projection onto the indices $N+1,N+2,...$. As a general rule, the faster asymptotic incoherence decays the more we are able to subsample (see (\ref{conditions31_levels})). The study of more precise notions of coherence has also been considered for the one and two dimensional discrete Fourier sampling, separable Haar sparsity problems in \cite{discrete}. This paper focuses on studying the structure of (\ref{asympinco}) in continuous multidimensional inverse problems and the impact this has on the ability to effectively subsample.

In previous work \cite{onedimpaper}, the structure of incoherence was analyzed as a general problem and theoretical limits on how fast it can decay over all such inverse problems were established. Furthermore, the notion of optimal decay was introduced, which describes the fastest asymptotic incoherence decay possible for a given inverse problem. The notion of an optimal ordering was also introduced, which acted as a set of instructions on how to actually attain this optimal incoherence decay rate by ordering the sampling basis. Optimal decay rates and optimal orderings were determined for the one-dimensional Fourier-wavelet and Fourier-polynomial cases and the former was found to attain the theoretically optimal incoherence decay rate of $N^{-1}$. By `optimal' here we mean in the sense of over all inverse problems that has $U$ an isometry. Furthermore, it is the fastest decay as a power of $N$. This paper extends the basic findings in \cite{onedimpaper} to general $d$-dimensional problems.

\begin{figure}[t]
\begin{center}
\begin{subfigure}[t]{0.43\textwidth}
\begin{center}
\includegraphics[width=\textwidth]{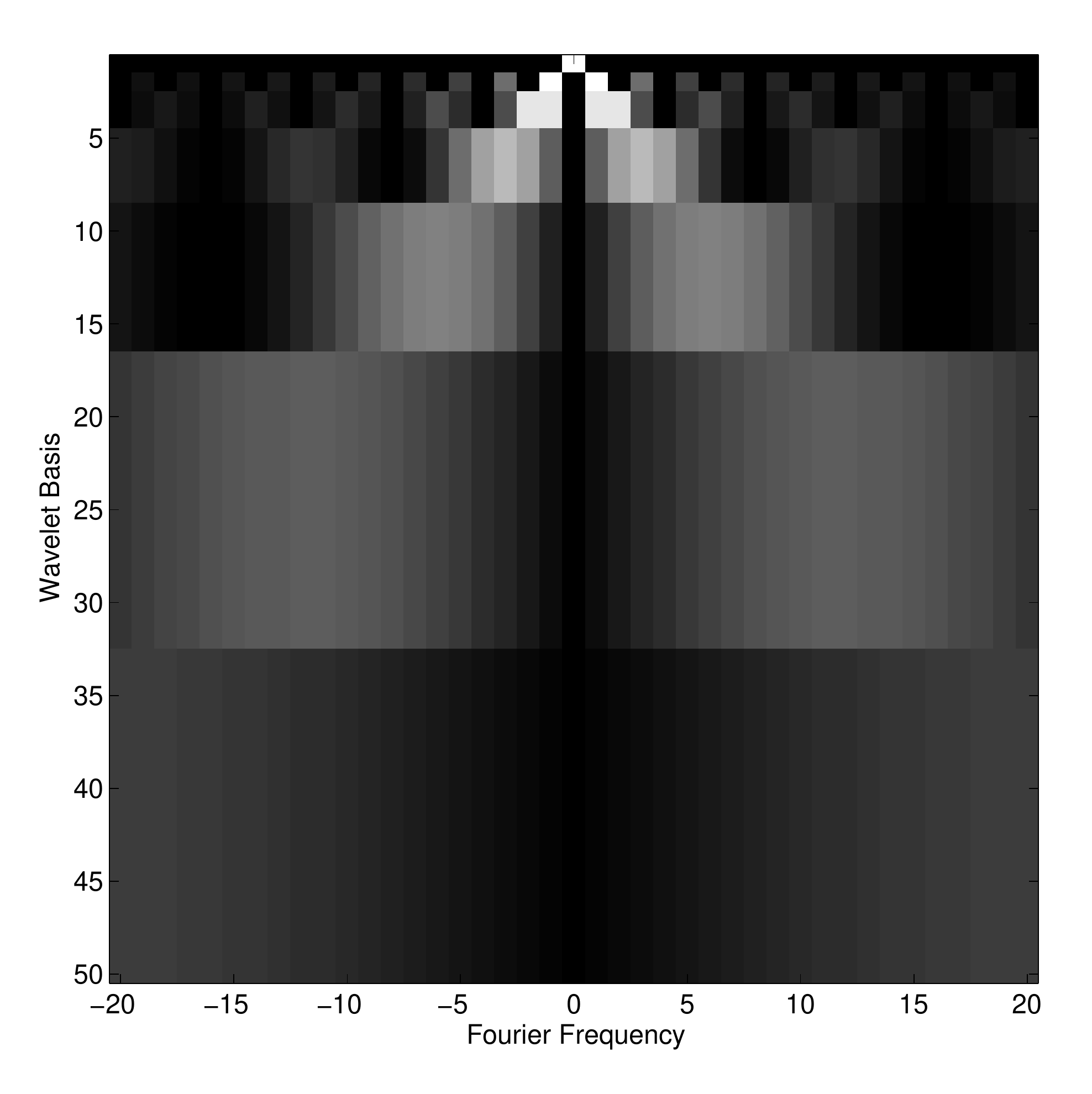}
 \caption{\footnotesize Coherence Matrix for the 1D case}
\end{center}
\end{subfigure}
\begin{subfigure}[t]{0.43\textwidth}
\begin{center}
 \includegraphics[width=\textwidth]{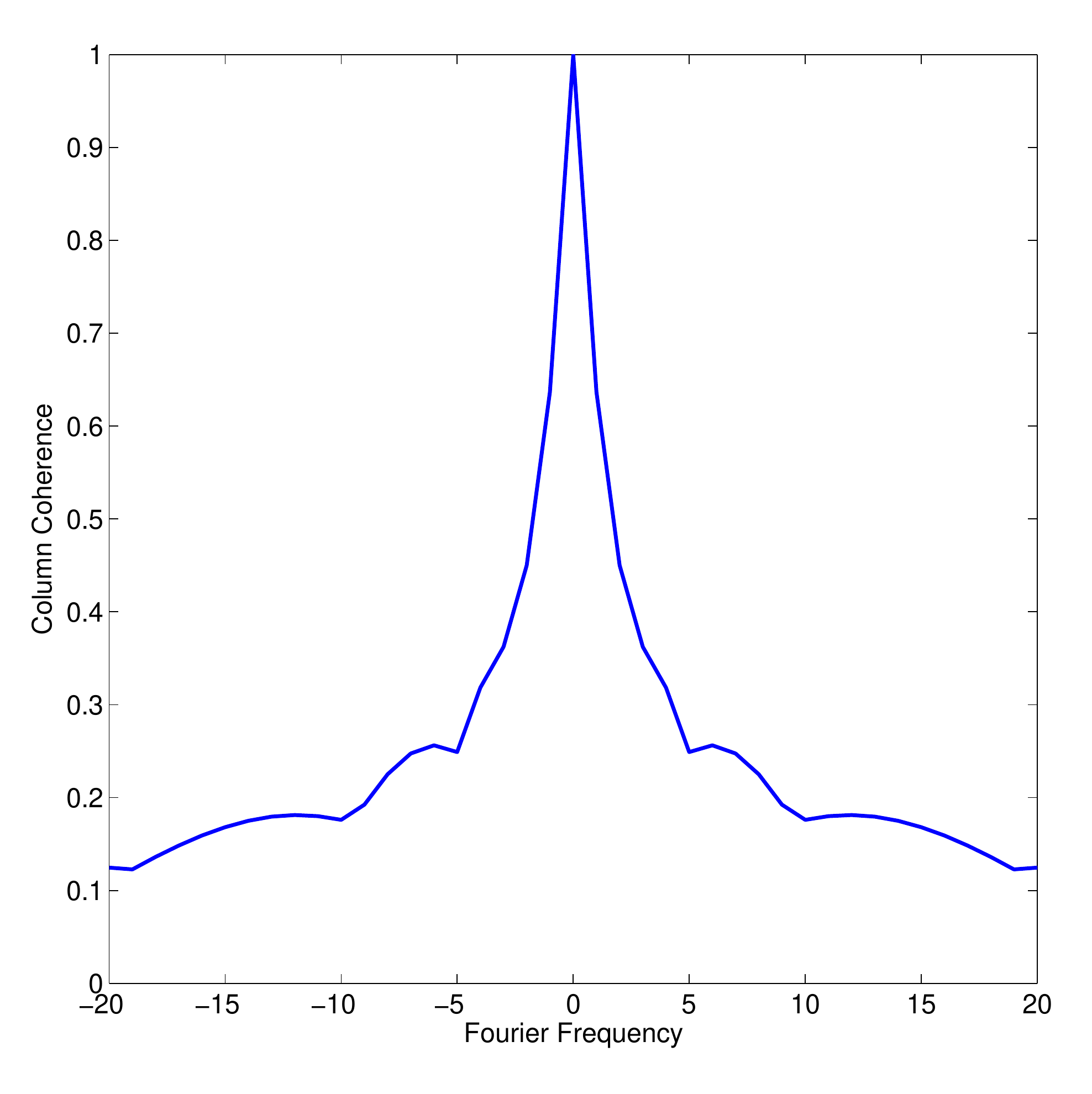} 
  \caption{\footnotesize 1D Column Coherences}
  \end{center}
\end{subfigure}\\
\begin{subfigure}[t]{0.43\textwidth}
\begin{center}
\includegraphics[width=\textwidth]{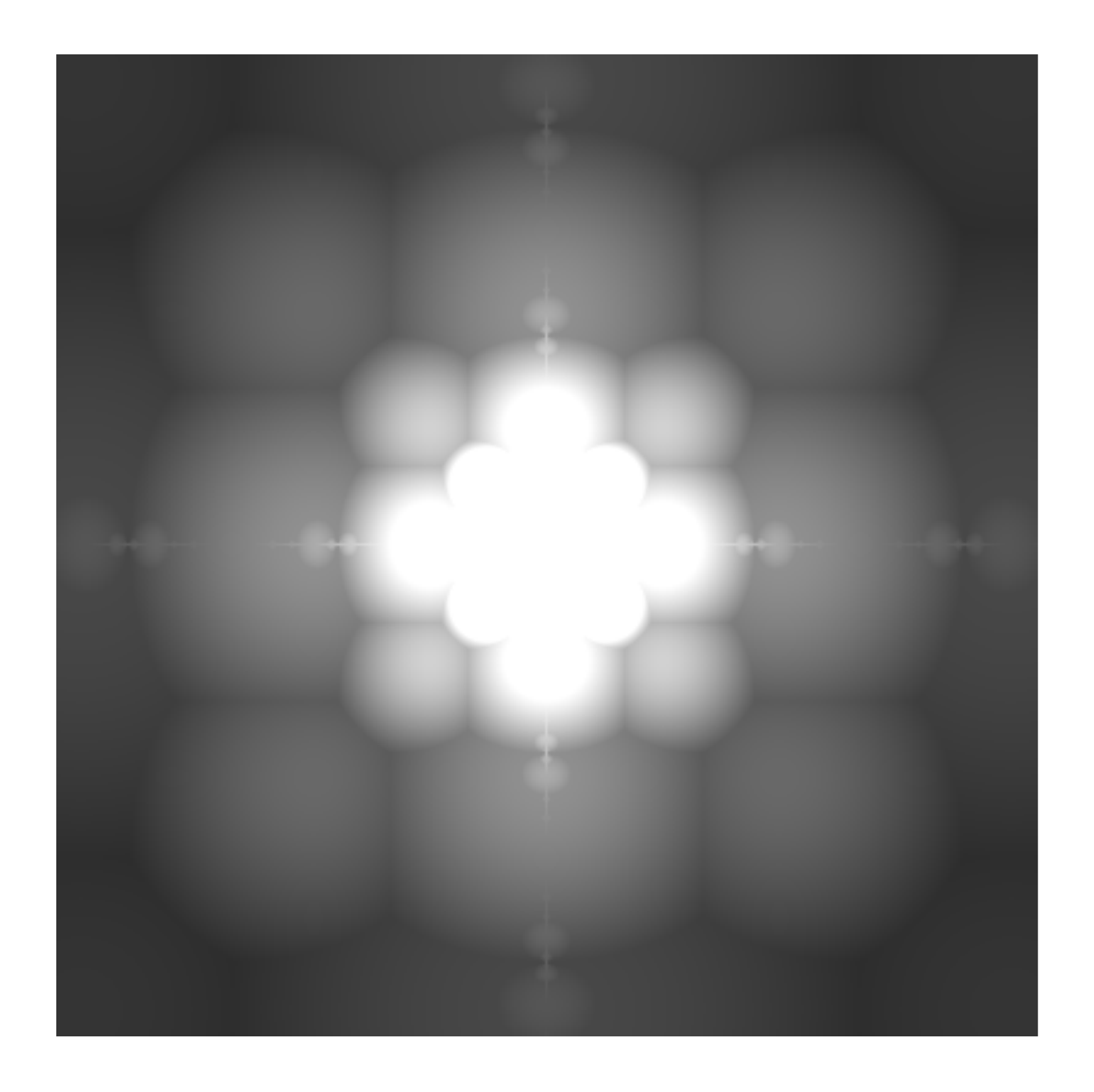} 
  \caption{\footnotesize 2D analogue of (b)}
  \end{center}
\end{subfigure}
\begin{subfigure}[t]{0.43\textwidth}
\begin{center}
\includegraphics[width=\textwidth]{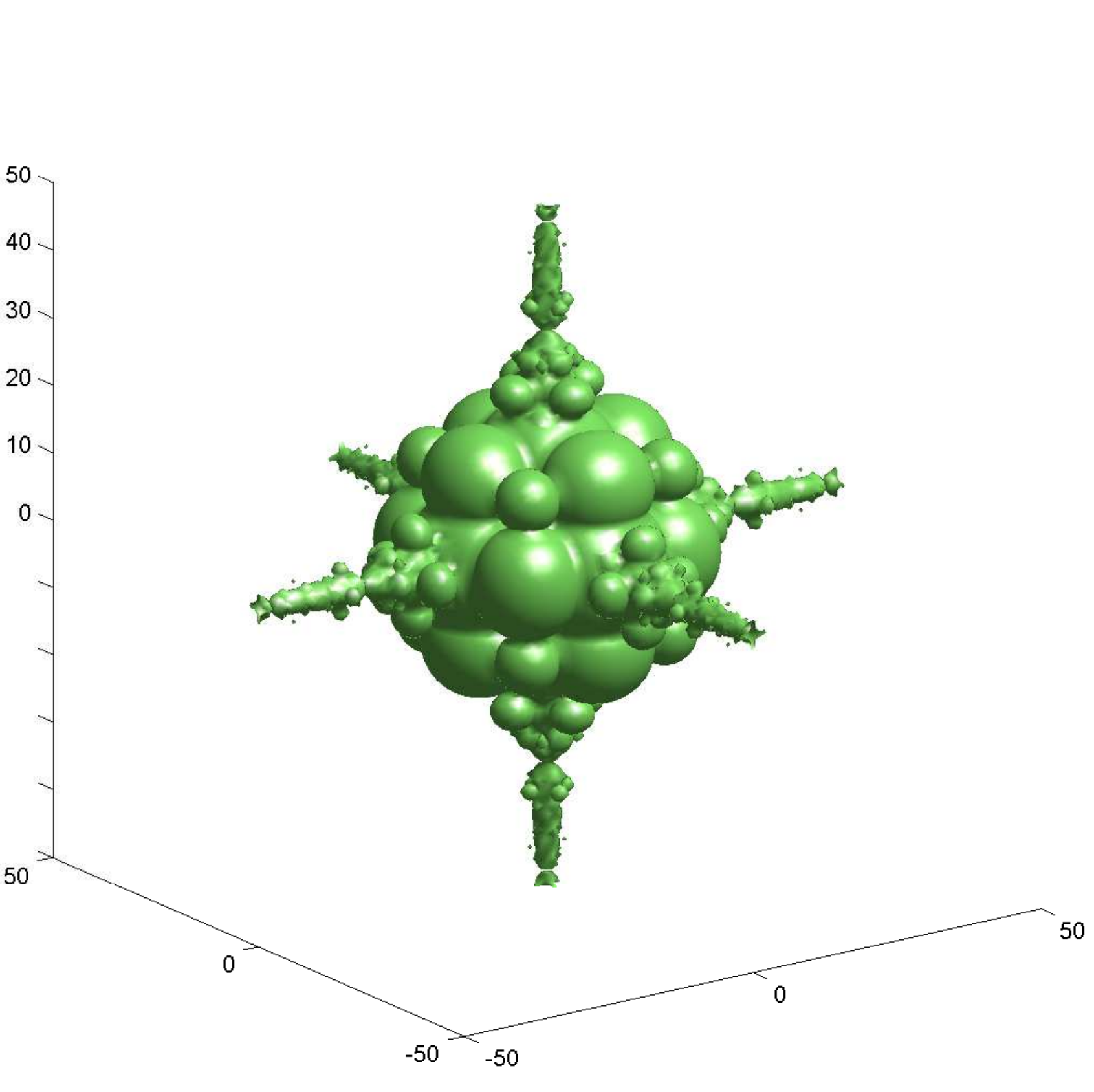} 
\caption{\footnotesize Isosurface of 3D Case}
\end{center}
\end{subfigure}
\end{center}
\caption{Fourier - Separable Haar Cases: Incoherence Structures in Different Dimensions. In (b), the coherences are calculated by taking the maxima over the columns in (a), demonstrating decay that scales with frequency. In 2D this decay roughly matches that of the norm of the frequency as seen in (c). However in 3D there are hyperbolic spikes around the coordinate axes that lead to poor incoherence decay (see (d)) when using sampling patterns with rotational invariance or linear scaling. In the black and white plots, white indicates larger absolute value.}
\label{haarimages}
\end{figure}

The optimal orderings in these one dimensional cases matched the leveled schemes that were already used for subsampling. For example, when sampling from the 1D Fourier basis, the sampling levels are usually ordered according to increasing frequency. In multiple dimensions there is no such consensus, instead many different sampling patterns are used, especially when it comes to 2D sampling patterns where radial lines \cite{radial}, spirals \cite{spiralmed} or other k-space trajectories are used. There are also a variety of other sampling techniques used in even higher dimensional (3-10D) problems, such as in the field of NMR spectroscopy \cite{highdimnmr}. If one desires to exploit asymptotic incoherence to its fullest it must be understood whether the coherence structure is consistent with the sampling pattern that one intends to use.

This paper determines optimal orderings for the case of Fourier sampling and (separable) wavelet sparsity in any dimension. It is shown that the optimal decay is always that of the one-dimensional case, and moreover in two dimensions the optimal orderings are compatible with the structure of the 2D sampling patterns mentioned above. However, in higher dimensions problems with poor wavelet smoothness, such as the three dimensional separable Haar case, the class of optimal orderings\footnote{Technically we mean \emph{strongly} optimal here (see Definition \ref{strongoptimality}).} are no longer rotationally invariant (as in Figure \ref{haarimages}), hindering the ability to subsample with traditional sampling schemes. It is also shown that using a pair of tensor bases in general leads to a best possible incoherence decay that is always anisotropic and suboptimal.

We should mention here that for many inverse problems in higher dimensions, using separable wavelets as a reconstruction basis fairs poorly against other bases such as shearlets \cite{shearlet} and curvelets \cite{candes2004new} for approximating images with curve-like features. However, it is not our goal to focus on a particular reconstruction basis in this paper, instead we wish to demonstrate how the incoherence structure can vary for different bases and the impact this has on its application in compressed sensing problems, for good or for worse.

\subsection{Setup \& Key Concepts : Incoherence, Sparsity \& Orderings}
Throughout this paper we shall work in an infinite dimensional separable Hilbert space $ \mathcal{H}$, typically $\mathcal{H}=L^2(\bbR^d)$, with two closed infinite dimensional subspaces $V_1, V_2$ spanned by orthonormal bases $B_1,B_2$ respectively,
\[ V_1 = \overline{ \text{Span} \{ f \in B_1 \}}, \qquad   V_2 = \overline{ \text{Span} \{ f \in B_2 \}    }.\]
We call $(B_1,B_2)$ a `basis pair'. If we are to form the change of basis matrix $U=(U_{i,j})_{i,j \in \bbN}$ we must list the two bases, which leads to following definitions:
\begin{definition}[Ordering]
Let $S$ be a set. Say that a function $\rho: \mathbb{N} \to S$ is an `ordering' of $S$ if it is bijective.
\end{definition}
\begin{definition}[Change of Basis Matrix]
For a basis pair $(B_1,B_2)$, with corresponding orderings $\rho:\mathbb{N} \to B_1$ and $\tau:\mathbb{N} \to B_2$, form a matrix $U$ by the equation
\begin{equation}\label{U}
 U_{m,n} :=  \langle \tau(n) , \rho(m) \rangle.
\end{equation}
Whenever a matrix $U$ is formed in this way we write `$U:=[(B_1,\rho),(B_2,\tau)]$'.
\end{definition}

Standard compressed sensing theory says that if $x \in \mathbb{C}^N$ is $s$-sparse, i.e.\ $x$ has at most $s$ nonzero components, then, with probability exceeding $1-\epsilon$, $x$ is the unique minimiser to the problem  
\bes{
\min_{\eta \in \bbC^N} \| \eta \|_{l^1} \quad \mbox{subject to} \quad P_{\Omega} U \eta 
= P_{\Omega} Ux,
}
where $P_{\Omega}$ is the projection onto $\mathrm{span}\{e_j:j\in \Omega\}$, $\{e_j\}$ is the canonical basis, $\Omega$ is chosen uniformly at random with $|\Omega| = m$ and
\be{
\label{m_est_Candes_Plan}
m \ge  C \cdot \mu(U) \cdot N \cdot s \cdot \log (\epsilon^{-1}) \cdot \log (N),
}
for some universal constant $C>0$ (see  \cite{Candes_Plan} and 
\cite{BAACHGSCS}). In \cite{AHPRBreaking} a new theory of compressed sensing was introduced based on the following three key concepts: 

\defn{[Sparsity in Levels]
\label{d:Asy_Sparse}
Let $x$ be an element of either $\bbC^N$ or $l^2(\bbN)$.  For $r \in \bbN$ let $\mathbf{M} = (M_1,\ldots,M_r) \in \bbN^r$ with $1 \leq M_1 < \ldots < M_r$ and $\mathbf{s} = (s_1,\ldots,s_r) \in \bbN^r$, with $s_k \leq M_k - M_{k-1}$, $k=1,\ldots,r$, where $M_0 = 0$.  We say that $x$ is $(\mathbf{s},\mathbf{M})$-sparse if, for each $k=1,\ldots,r$,
\bes{
\Delta_k : = \mathrm{supp}(x) \cap \{ M_{k-1}+1,\ldots,M_{k} \},
}
satisfies $| \Delta_k | \leq s_k$.  We denote the set of $(\mathbf{s},\mathbf{M})$-sparse vectors by $\Sigma_{\mathbf{s},\mathbf{M}}$.
}

\defn{[Multi-level sampling scheme]
\label{multi_level_dfn}
Let $r \in \bbN$, $\mathbf{N} = (N_1,\ldots,N_r) \in \bbN^r$ with $1 \leq N_1 < \ldots < N_r$, $\mathbf{m} = (m_1,\ldots,m_r) \in \bbN^r$, with $m_k \leq N_k-N_{k-1}$, $k=1,\ldots,r$, and suppose that
\bes{
\Omega_k \subseteq \{ N_{k-1}+1,\ldots,N_{k} \},\quad | \Omega_k | = m_k,\quad k=1,\ldots,r,
}
are chosen uniformly at random, where $N_0 = 0$.  We refer to the set
\bes{
\Omega = \Omega_{\mathbf{N},\mathbf{m}} := \Omega_1 \cup \ldots \cup \Omega_r
}
as an $(\mathbf{N},\mathbf{m})$-multilevel sampling scheme.
}

\defn{[Local coherence]\label{loc_coherence}
Let $U$ be an isometry of either $\bbC^{N}$ or $l^2(\bbN)$.
If $\mathbf{N} = (N_1,\ldots,N_r) \in \bbN^r$ and $\mathbf{M} = (M_1,\ldots,M_r) \in \bbN^r$ with $1 \leq N_1 < \ldots N_r $ and $1 \leq M_1 < \ldots < M_r $ the $(k,l)^{\rth}$ local coherence of $U$ with respect to $\mathbf{N}$ and $\mathbf{M}$ is given by
\ea{ \label{localincoherence}
\mu_{\mathbf{N},\mathbf{M}}(k,l) &= \sqrt{\mu(P^{N_{k-1}}_{N_{k}}UP^{M_{l-1}}_{M_{l}}) \cdot  \mu(P^{N_{k-1}}_{N_{k}}U)},\qquad k,l=1,\ldots,r,
}
where $N_0 = M_0 = 0$ and $P^{a}_{b}$ denotes the projection matrix corresponding to indices $\{a+1,\hdots, b\}$.  
}

The paper \cite{AHPRBreaking} provided the following estimate (with $C>0$ a universal constant) regarding the local number of measurements $m_k$ in the $k^{\rth}$ level in order to obtain a good reconstruction with probability $\ge 1-\epsilon$:
\be{
\label{conditions31_levels}
\frac{m_k}{N_k-N_{k-1}} \ge C \cdot \log(\epsilon^{-1}) \cdot \left(
\sum_{l=1}^r \mu_{\mathbf{N},\mathbf{M}}(k,l) \cdot s_l\right) \cdot \log\left(N\right),\quad k=1,\ldots,r.
}
In particular, the sampling strategy (i.e.\ the parameters $\mathbf{N}$ and $\mathbf{m}$) is now determined through the local sparsities and coherences. Since the local coherence (\ref{localincoherence}) is rather difficult to analyze in its current form, we bound it above by the following:
\ea{ \label{local2asymp}
\mu_{\mathbf{N},\mathbf{M}}(k,l) &= \sqrt{\mu(P^{N_{k-1}}_{N_{k}}UP^{M_{l-1}}_{M_{l}}) \cdot  \mu(P^{N_{k-1}}_{N_{k}}U)}
\\  \label{local2asymp2} & \le \sqrt{\min(\mu(P^{N_{k-1}}_{N_{k}}U), \mu(U P^{M_{l-1}}_{M_{l}})) \cdot  \mu(P^{N_{k-1}}_{N_{k}}U)}
\\ \label{local2asymp3} & \le \sqrt{\min(\mu(P^\perp_{N_{k-1}}U), \mu(U P^\perp_{M_{l-1}})) \cdot  \mu(P^\perp_{N_{k-1}}U)}
}
It is arguably (\ref{local2asymp3}) rather than (\ref{local2asymp2}) that is the roughest bound here, however we shall see that this becomes effectively an equality in what follows. The crucial improvement of (\ref{local2asymp3}) over (\ref{local2asymp}) is that it is completely in terms of the asymptotic incoherences $\mu(P_N^\perp U), \mu(U P_N^\perp)$, which depend only on the orderings of $B_1,B_2$ respectively, rather than both of them. Furthermore, we can treat the two problems of maximizing the decay of  
$\mu(P_N^\perp U), \mu(U P_N^\perp)$ separately and then combine the two resulting orderings together at the end.

Next we describe how one determines the fastest decay of $\mu(P_N^\perp U)$. In \cite{onedimpaper} this was done via the notion of optimality up to constants:

\begin{definition}[Optimal Orderings] \label{fasterdecay}
Let $\rho_1, \rho_2 : \mathbb{N} \to B_1$ be any two orderings of a basis $B_1$ and $\tau$ any ordering of a basis $B_2$. Let $U_1:=[(B_1,\rho_1) , (B_2,\tau)], \ U_2:=[(B_1,\rho_2) , (B_2,\tau)]$ as in (\ref{U}). Also let $Q_N:=P_{N-1}^\perp$. If there is a constant $C>0$ such that
\[ \mu(Q_NU_1) \le C \cdot \mu(Q_NU_2), \qquad \forall N \in \mathbb{N}, \]
then we write $\rho_1 \prec \rho_2$ and say that `$\rho_1$ has a faster decay rate than $\rho_2$ for the basis pair $(B_1, B_2)$'. $\rho_1$ is said to be an `optimal ordering of $(B_1,B_2)$' if $\rho_1 \prec \rho_2$ for all other orderings $\rho_2$ of $B_1$. The relation $\prec$, defined on the set of orderings of $B_1$, is independent of the ordering $\tau$ since the values of $\mu(Q_N U_1), \mu(Q_N U_2)$ are invariant under permuation of the columns of $U_1, U_2$.
 \end{definition}

It was shown in \cite{onedimpaper} that optimal orderings always exist. Optimal orderings are used to give us the optimal decay rate:
 \begin{definition}[Optimal Decay Rate] Let $f,g : \mathbb{N} \to \mathbb{R}_{>0}$ be decreasing functions.
We write $f \lesssim g$ to mean there is a constant $C>0$ such that
\[ f(N)  \le  C \cdot g(N), \qquad \forall N \in \mathbb{N}. \]
If both $f \lesssim g$ and $g \lesssim f$ holds, we write `$f \approx g$'.
Now suppose that $\rho: \mathbb{N} \to B_1$ is an optimal ordering for the basis pair     $(B_1,B_2)$ and we let $U=[(B_1,\rho),(B_2,\tau)]$ be a corresponding incoherence matrix (with some ordering $\tau$ of $B_2$). Then any decreasing function $f: \mathbb{N} \to \mathbb{R}_{>0}$ which satisfies $f \approx g$, where $g$ is defined by $g(N) = \mu(Q_N U)$, $\forall N \in \bbN$, is said to `represent the optimal decay rate' of the basis pair $(B_1,B_2)$. 
\end{definition}

Notice that the optimal decay rate is unique up to the equivalence relation $\approx$ defined on the set of decreasing functions $f: \mathbb{N} \to \mathbb{R}_{>0}$.

We also have a stronger notion of optimality, which gives us finer details on the exact decay:

\begin{definition} [Strong Optimality] \label{strongoptimality}
Let $U=[(B_1,\rho),(B_2,\tau)]$ and $\pi_N$ denote the projection onto the single index $N$. If $f$ represents the optimal decay rate of the basis pair $(B_1,B_2)$ then $\rho$ is said to be `strongly optimal' if the function $g(N):= \mu(\pi_N U)$ satisfies $f \approx g$. 
\end{definition}

Estimates in terms of the row incoherence $\mu(\pi_N U)$ have used before in \cite{discrete}, where it was called the `local coherence'. If $\rho$ is a strongly optimal ordering, $U=[(B_1,\rho),(B_2,\tau)]$ and $f$ represents the optimal decay of $(B_1,B_2)$ then 
\[
\mu(Q_N U) \le C_1 \cdot  f(N) \le C_2 \cdot \mu(\pi_N U) \le C_2 \cdot \mu(P_{N-1}^{N-1+M} U), \qquad N, M \in \bbN,
\]
for some constants $C_1(\rho), C_2(\rho)>0$, which can then be used to show the $\le$ in (\ref{local2asymp3}) can be replaced by $\approx$.

We shall introduce the Fourier basis here as it is used in all of the examples discussed in this paper:

\begin{definition}[Fourier Basis] \label{fourier}

If we define
\[ \chi_k(x) = \sqrt{\epsilon} \exp(2 \pi \ri \epsilon k x)\cdot \ \mathds{1}_{[(- 2 \epsilon)^{-1},(2 \epsilon)^{-1}]} (x), \qquad k \in \mathbb{Z}, \]
then the $(\chi_k)_{k \in \bbZ}$ form a basis\footnote{The little $\rf$ here stands for `Fourier'.} $B_\rf(\epsilon)$ of $L^2([-(2 \epsilon)^{-1},(2 \epsilon)^{-1}])$ . We can form a $d$-dimensional basis of $L^2([-(2 \epsilon)^{-1},(2 \epsilon)^{-1}]^d)$ by taking tensor products (see Section \ref{tensors})
\[ \chi_{k} := \bigotimes_{j=1}^d \chi_{k_j} ,
 \qquad k \in \mathbb{Z}^d, \]
and setting $B^d_\rf(\epsilon)= \{ \chi_{k} \ : \ k \in \mathbb{Z}^d \} $.
It shall be convenient to identify $B^d_\rf(\epsilon)$ with $\mathbb{Z}^d$ using the function 
\be{ \label{multidimlambda}
\lambda_d:B_\rf^d \to \mathbb{Z}^d, \quad  \lambda_d(\chi_k):=(\lambda(\chi_{k_1}), ..., \lambda(\chi_{k_d}))= (k_1,...,k_d)=k.
} 

\end{definition}

\section{Main Results}

It turns out that the task of determining the asymptotic incoherence for general $d$-dimensional cases is substantially more difficult and subtle than the $1$-dimensional problems. However, we are able to present sharp results on the decay as well as optimal orderings of the bases. 
The main results can be broken down into two groups: one for tensor cases in general and one for the Fourier-Separable wavelet case. In what follows $ d \in \bbN$ denotes dimension.

\subsubsection{Fourier to Tensor Wavelets}

\begin{theorem} \label{tensormainwavelet}
Let $B^d_\rw$ be a tensor wavelet basis. The optimal decay rate of both $(B^d_\rf, B^d_\rw)$ and $(B^d_\rw,B^d_\rf)$ is represented by $f(N)= \log^{d-1}(N) \cdot N^{-1}$.
\end{theorem}
This theorem is a user friendly and easy-to-read restatement of Theorem \ref{TensorResultsWavelet}. The latter theorem contains the more subtle and technical statements of the results. 

\subsubsection{Fourier to Legendre polynomials}
\begin{theorem} \label{tensormainpoly}
Let $B^d_\rp$ be a (tensor) Legendre polynomial basis. The optimal decay rate of both $(B^d_\rf, B^d_\rp)$ and $(B^d_\rp,B^d_\rf)$ is represented by $f(N)= \big( \log^{(d-1)}(N) \cdot N^{-1} \big)^{2/3}$. 
\end{theorem}
This theorem is a restatement of Theorem \ref{TensorResultsPoly} for the purpose of an easy-to-read exposition.  
The additional logarithmic factors in the tensor cases here demonstrates the typical problems associated with dimensionality. In general the optimal orderings for all tensor problems are constructed using the hyperbolic cross on the original one-dimensional optimal orderings.

\subsubsection{Fourier to Separable Wavelets}
The definition of a separable wavelet basis $B_\text{sep}^d$ is provided in Section \ref{separable}. The main results on these cases are summarized below:
\begin{theorem} \label{separablesummary}
Consider the Fourier basis $B^d_\rf$ and the wavelets basis $B^d_\text{sep}$. Then the following is true.
\begin{itemize}
\item[(i)]
The optimal decay rate of $(B^d_\text{sep},B^d_\rf)$ is represented by $f(N)= N^{-1}$. The optimal decay rate of $(B^d_\text{sep},B^d_\rf)$ is obtained by using an ordering $\tau$ that is consistent with the wavelet levels (see definitions \ref{consistent_ordering} and \ref{leveled}) and this ordering is strongly optimal.
\item[(ii)]
In 2D ($d=2$) the optimal decay rate of $(B^d_\rf,B^d_\text{sep})$ is represented by $f(N)= N^{-1}$. This optimal decay rate is obtained by using an ordering $\rho$ of $B^d_\rf$ that satisfies, for some constants $C_1, C_2>0$ and some norm $\| \cdot \|$ on $\bbR^d$,
\be{ \label{linearrough}
\max(\| \lambda_d(\rho(N)) \|,1)  \approx N^{1/d}, \qquad  N \in \mathbb{N}.
}
In fact $\rho$ is strongly optimal in 2D if and only if (\ref{linearrough}) holds.
\item[(iii)]
In higher dimensions ($d \ge 3$) the optimal decay rate of $(B^d_\rf,B^d_\text{sep})$ is still represented by $f(N)= N^{-1}$. However the optimal ordering used to obtain this decay rate is dependent on the wavelet used to generate the basis $B^d_\text{sep}$. 
\end{itemize}
\end{theorem}
Part (i) is the subject of Section \ref{separablewaveletordering} and is proven in Corollary \ref{leveledresults}. Part (ii), tackled in Section \ref{linearproof}, is the same as Corollary \ref{twodimresults}. Part (iii), covered in Section \ref{semihypsection}, is proven in Theorem \ref{semihyperbolicthm}.

An ordering satisfying (\ref{linearrough}) is called a `linear ordering'. The class of linear orderings are rotation invariant and  compatible with sampling schemes based on linearly scaling a fixed shape from the origin (see Section \ref{linearsection}).

 Optimal orderings in the case of high dimensions and poor wavelet smoothness can be found by 
interpolating between the case of (\ref{linearrough}) and the hyperbolic cross, which generates semi-hyperbolic orderings (see Definition \ref{semihyperbolic}). If the wavelet is sufficiently smooth relative to the dimension then linear orderings are optimal. It is also shown that if a linear ordering is optimal then the wavelet used must have some degree of smoothness proportional to the dimension; in 3D it is $C^0$, 5D it is $C^1$, 7D it is $C^2$, etc. (see Section \ref{orderingsandsmoothness}).

The differences between the two incoherence structures of the Fourier-Tensor wavelet and Fourier-Separable wavelet cases are tested in 2D in Section \ref{numericalsection}.

\subsection{Outline for the Remainder}

Some key tools that we use to find optimal orderings are given in Section \ref{orderings}. Those familiar with \cite{onedimpaper} can skip the majority of this section except for the concept of characterization. We then cover the general tensor case and introduce hyperbolic orderings in Section \ref{tensors} and prove Theorem \ref{tensormainwavelet} and Theorem \ref{tensormainpoly}. In Section \ref{separable} we discuss the separable cases, first covering how to optimally order the wavelet basis before quickly moving on to the central problem of finding optimal orderings of the Fourier basis. Linear orderings are introduced first, then we justify the need for semihyperbolic orderings. Finally we move onto some simple compressed sensing experiments, one demonstrating the benefits of multilevel subsampling and one showing the impact of differing incoherence structures between the 2D tensor and separable cases.

\section{Tools for Finding Optimal Orderings \& Theoretical Limits on Optimal Decay} \label{orderings}

The first tool is perhaps the most important, as it is a very easy way to identify a strongly optimal ordering:

\begin{lemma} \label{StrongOptimalEquivalence}
\textbf{1):} Let $(B_1,B_2)$ be a basis pair and $\tau$ any ordering of $B_2$. Furthermore, let $ B_1$ have an ordering $\rho_1 : \mathbb{N} \to B_1$, and define $U_1:=[(B_1,\rho_1),(B_2,\tau)]$.  Suppose that that there is a decreasing function $f_1: \mathbb{N} \to \mathbb{R}_{>0}$ such that
\[ f_1(N)  \le  \mu(\pi_N U_1), \qquad \forall N \in \mathbb{N}. \]
Then if $\rho_2: \mathbb{N} \to B_1$ is an ordering, $U_2=[(B_1,\rho_2),(B_2,\tau)]$  and $f_2: \mathbb{N} \to \mathbb{R}_{>0}$ is a function with
\[ \mu( Q_N U_2)  \le   f_2(N), \qquad \forall N \in \mathbb{N}, \]
then $f_1(N) \le f_2(N)$ for every $N \in \mathbb{N}$.

\textbf{2):} Let $\rho$ be an ordering of $B_1$ with $U:=[(B_1,\rho),(B_2,\tau)]$ and $f: \mathbb{N} \to \mathbb{R}_{\ge 0}$ be a decreasing function  with $f(N) \to 0$ as $N \to \infty$. If, for some constants $C_1, C_2 > 0$, we have
\begin{equation} \label{rowordercond}
 C_1 f(N)  \le  \mu( \pi_N U)  \le  C_2 f(N), \qquad \forall N \in \mathbb{N}, 
 \end{equation}
 then $\rho$ is a strongly optimal ordering and $f$ is a representative of the optimal decay rate.
\end{lemma}
\begin{proof}
See Lemma 2.11 in \cite{onedimpaper}.
\end{proof}

\begin{definition}[Best ordering]
Let $(B_1,B_2)$ be a basis pair. Then any ordering $\rho: \mathbb{N} \to B_1$ is said to be a `best ordering' if for any ordering $\tau$ of $B_2$ and $U=[(B_1,\rho),(B_2,\tau)]$ we have that the function $g(N):= \mu(\pi_N U)$ is decreasing.
\end{definition}
Notice that any best ordering is also a strongly optimal ordering. We shall need the notion of a best ordering briefly to prove Lemma \ref{characterisationlemma}.
\begin{lemma} \label{bestexistence} 
Suppose that we have a basis pair $(B_1, B_2)$ with two orderings $\rho: \bbN \to B_1$, $\tau: \bbN \to B_2$ of $B_1, B_2$ respectively. If $U=[(B_1,\rho),(B_2,\tau)]$ satisfies
\[ \mu(\pi_N U) \to 0 \quad \text{as} \quad N \to \infty, \]
then a best ordering exists.
\end{lemma}
\begin{proof}
See Lemma 2.10 in \cite{onedimpaper}.
\end{proof}

Throughout this paper we would like to define an ordering according to a particular property of the basis but this property may not be enough to specify a unique ordering. To deal with this issue we introduce the notion of consistency:

\begin{definition}[Consistent ordering]\label{consistent_ordering}
Let $F: S \to \mathbb{R}$ where $S$ is a set. We say that an ordering $\rho: \mathbb{N} \to S$ is `consistent with F' if
\[ F(f)  <  F(g)  \quad \Rightarrow  \quad \rho^{-1}(f)  <  \rho^{-1}(g), \qquad \forall f,g \in S. \]
\end{definition}

The notion of consistency becomes important if we want to convert bounds on the coherence into optimal orderings:

\begin{definition} \label{Characterisation}
\begin{itemize}
\item[1.)]
 Suppose $F:S \to \bbR_{> 0} $ satisfies $| \{ x \in S :  1/F(x) \le K \}| < \infty$  for all $K>0$, $\sigma: \bbN \to S$ is consistent with $1/F$ and $F(\sigma(N)) \to 0$ as $N \to \infty$. Then any  decreasing function $f: \bbN \to \bbR_{>0}$  such that $f \approx F \circ \sigma$ is said to `represent the fastest decay of $F$'.
\item[2.)]
Suppose $(B_1,B_2)$ is a basis pair and $\iota: S \to B_1$ a bijection. If there exists a function $F:S \to \bbR_{>0}$ and a constant $C_1>0$ such that 
\be{ \label{dominate}
\sup_{g \in B_2} | \langle \iota(s) , g \rangle |^2 \le C_1 \cdot F(s), \quad \forall s \in S,
}
then $F$ is said to `dominate the optimal decay of $(B_1,B_2)$'. If the inequality is reversed we say $F$ is `dominated by the optimal decay of $(B_1,B_2)$'. Furthermore, if there is a constant $C_2>0$ such that
\be{ \label{characterise}
C_2 \cdot F(s) \le \sup_{g \in B_2} | \langle \iota(s) , g \rangle |^2 \le C_1 \cdot F(s), \quad \forall s \in S,
}
then $F$ is said to `characterize the optimal decay of $(B_1,B_2)$'.
\end{itemize}
\end{definition}

\begin{lemma} \label{characterisationlemma}
1):  \ \ Suppose $f$ is a representative of the optimal decay rate for the basis pair $(B_1,B_2)$, $\iota: S \to B_1$ is a bijection, $F: S \to \bbR$ dominates the optimal decay of $(B_1,B_2)$, $\sigma : \bbN \to S$ is consistent with $1/F$ and $U=[(B_1, \iota \circ \sigma), (B_2, \tau)]$ . Then if $g$ represents the fastest decay of $F$ then $f, \mu(\pi_{\cdot}U) \lesssim g$. 

2):  \ \ If $F$ is instead is dominated by the optimal decay of $(B_1,B_2)$ then $f, \mu(\pi_{\cdot}U) \gtrsim g$.

3):  \ \ If $F$ now characterizes the optimal decay of $(B_1,B_2)$ then $f, \mu(\pi_{\cdot}U) \approx g$ and therefore $\rho$ is a strongly optimal ordering for the basis pair $(B_1,B_2)$ if and only if $F(\iota^{-1} \circ \rho(\cdot)) \approx g $.
\end{lemma}
\begin{proof} 1.) \ \ We may assume, without loss of generality, that $g(N) \to 0$ as $N \to \infty$ else there is nothing to prove as $f(N), \mu(\pi_N U)$ are bounded functions of $N$. Therefore, a best ordering exists by Lemma \ref{bestexistence}. (\ref{dominate}) becomes (for $C_1'>0$ a constant),
\[
\mu(\pi_N U) \le C_1 \cdot F(\sigma(N)) \le C_1' \cdot g(N), \quad \forall N \in \bbN.
\]
Since $g$ is decreasing we have $\mu(Q_N U) \le C_1'\cdot g(N)$ and therefore we can apply part 1) of Lemma \ref{StrongOptimalEquivalence} to $f_1=f$  and $f_2=g$ (using a best ordering as $\rho_1$ and $\rho_2 = \iota \circ \sigma$) to deduce that $f \lesssim g$.

2.) \ \ (\ref{dominate}) reversed becomes
\[
\mu(\pi_N U) \ge C \cdot F(\sigma(N)) \ge C_1' \cdot g(N), \quad \forall N \in \bbN.
\]
Therefore we can apply part 1) of Lemma \ref{StrongOptimalEquivalence} to $f_1=g$ and $f_2=f$ (using $\rho_1 = \iota \circ \sigma$, $\rho_2$ an optimal ordering) to deduce that $f \lesssim g$.

3.) \ \ Notice that if 
$F$ characterizes the optimal decay of $(B_1,B_2)$ then (\ref{characterise}) becomes
\[
C_2' \cdot g(N) \le C_2 \cdot F(\sigma(N)) \le \mu(\pi_N U)  \le C_1 \cdot F(\sigma(N)) \le C'_1 \cdot g(N), \quad \forall N \in \bbN,
\]
and we can then apply part 2) of Lemma  \ref{StrongOptimalEquivalence} to show $f, \mu(\pi_{\cdot}U) \approx g$. If we let $U':=[(B_1, \rho), (B_2, \tau)]$ then (\ref{characterise}) becomes
\[
C_2 \cdot F(\iota^{-1} \circ \rho(N)) \le \mu(\pi_N U')  \le C_1 \cdot F(\iota^{-1} \circ \rho(N)) , \quad \forall N \in \bbN,
\]
and the result follows from Definition \ref{strongoptimality}.
\end{proof}

Before moving on, we recall from \cite{onedimpaper} some results on the fastest optimal decay rate for a basis pair:

\begin{theorem} \label{isometrydecaylowerbound}
Let $U \in \cB(l^2(\bbN))$ be an isometry.  Then
$
\sum_{N} \mu(Q_N U)
$
diverges. 
\end{theorem}
\begin{proof}
See Theorem 2.14 in \cite{onedimpaper}.
\end{proof}
\begin{corollary}
Let $U \in \cB(l^2(\bbN))$ be any isometry.  Then there does not exist an $\epsilon > 0$ such that
$$
\mu(Q_N U)  = \mathcal{O}(N^{-1-\epsilon}) , \qquad \ N \to \infty.
$$
\end{corollary}

It turns out that Theorem \ref{isometrydecaylowerbound} cannot be improved without imposing additional conditions on $U$:

\begin{lemma} \label{incoherencecounter}
Let $f,g: \bbN \to \bbR$ be any two strictly positive decreasing functions and suppose that $\sum_N f(N)$ diverges. Then there exists $U \in \cB(l^2(\bbN))$ an isometry with 
\be{ \label{strongerthanoptimal}
\mu(Q_N U) \le f(N), \quad \mu(U Q_N) \le g(N), \qquad N \in \bbN .
}
\end{lemma}
\begin{proof}
See Lemma 2.16 in \cite{onedimpaper}.
\end{proof}

If we restrict our decay function to be a power law, i.e. $f(N):= CN^{- \alpha}$ for some constants $\alpha, C >0$ then the largest possible value of $\alpha>0$ such that (\ref{strongerthanoptimal}) holds for an isometry $U$ is $\alpha=1$. This gives us a notion of the fastest optimal decay rate as a power of $N$ over all pairs of bases where the span of $B_2$ lies in the span of $B_1$.

\section{One-dimensional Bases and Incoherence Results} \label{onedim}

Before we begin our review of the one-dimensional cases by quickly going the one-dimensional bases and orderings that we shall be working with to construct multi-dimensional bases and orderings in Section \ref{tensors}.

\subsection{Fourier Basis} 
We recall the one-dimensional Fourier basis $B_\rf(\epsilon)=(\chi_k)_{k \in \bbZ}$ from Definition \ref{fourier}.

\begin{definition}[Standard ordering]\label{standard_ordering}
We define $F_\rf:B_\rf \to \mathbb{N} \cup \{0\}$  by $F_\rf(\chi_k)=|k|$ and say that an ordering $\rho: \mathbb{N} \to B_\rf$ is a `standard ordering' if it is consistent with $F_\rf$ (recall Definition \ref{consistent_ordering}). 
\end{definition}

\subsection{Standard Wavelet Basis} \label{waveletbasis}

Take a Daubechies wavelet $\psi$ and corresponding scaling function $\phi$ in $L^2(\mathbb{R})$ with \[\text{Supp} (\phi) = \text{Supp} (\psi) = [-p+1,p]. \]  
We write
\[
\begin{aligned}
 \phi_{j,k}(x) =2^{j/2} \phi(2^j x-k) , \qquad \psi_{j,k} (x)  = 2^{j/2} \psi(2^j x - k), 
 \\ V_j := \overline{\text{Span} \{ \phi_{j,k}: k \in \bbZ \}}, \quad W_j := \overline{\text{Span} \{ \psi_{j,k}: k \in \bbZ \}}.
 \end{aligned}
 \]
With the above notation, $(V_j)_{j \in \mathbb{Z}}$ is the multiresolution analysis for $\phi$, and therefore
\[V_j \subset V_{j+1} , \qquad V_{j+1} = V_j \oplus W_j, \qquad L^2(\bbR) = \overline{\bigcup_{j \in \bbZ} V_j},  \]
where $W_j$ here is the orthogonal complement of $V_j$ in $V_{j+1}$. For a fixed $J \in \bbN$ we define the set\footnote{`$\rw$' here stands for `wavelet'.}
\begin{align} \label{waveletbasisdefine}
B_\rw := \left\{ \begin{array}{cc}   &  \mathrm{Supp}(\phi_{J,k}) \cap (-1,1) \neq \emptyset , \\  \phi_{J,k} , \ \psi_{j,k} :  &  \mathrm{Supp}(\psi_{j,k}) \cap (-1,1)  \neq \emptyset, \\ &  j \in \mathbb{N}, j \ge J , \ k \in \mathbb{Z}  \end{array} \right \} , 
\end{align}
Let $\rho$ be an ordering of $B_\rw$. Notice that since $L^2(\mathbb{R})= \overline{ V_J \oplus \bigoplus^{\infty}_{j=J} W_j}$ for all 
$f \in L^2(\mathbb{R})$ with $\mathrm{supp}(f) \subseteq [-1,1]$ we have
\[ f  =  \sum_{n=1}^\infty c_{n} \rho(n) \quad \text{for some} \quad (c_n)_{n \in \bbN} \in \ell^2(\bbN) .\] 

\begin{definition} [Leveled ordering (standard wavelets)]\label{leveled}
Define $F_\rw:B_\rw \to \mathbb{R}$ by 
\[ F_\rw( f) \ = \ \begin{cases} 
\ j, \  & \mbox{if }  f \in W_j\\
\ -1,  \ & \mbox{if }f \in V_J 
\end{cases} ,  \]
and say that any ordering $\tau: \mathbb{N} \to B_\rw$ is a `leveled ordering' if it is consistent with $F_\rw$.
\end{definition}
 Notice that $F_\rw(\psi_{j,k})=j$. We use the name ``leveled'' here since requiring an ordering to be leveled means that you can order however you like within the individual wavelet levels themselves, as long as you correctly order the sequence of wavelet levels according to scale.
 
Suppose that $U=[(B_\rf(\epsilon), \rho), (B_\rw,\tau)]$ for orderings $\rho, \tau$. If we require $U$ to be an isometry we must impose the constraint $(2\epsilon)^{-1} \ge 1+2^{-J+1}(p-1)$ otherwise the elements in $B_\rw$ do not lie in the span of $B_\rf(\epsilon)$. For convenience we rewrite this as $\epsilon \in I_{J,p}$ where 
\[I_{J,p}:=(0,(2+2^{-J+2}(p-1))^{-1}]. \]

\subsection{Boundary Wavelet Basis} \label{BoundaryWavelets}

We now look at an alternative way of decomposing a function $f \in L^2([-1,1])$ in terms of a wavelet basis, which involves using boundary wavelets  \cite[Section 7.5.3]{dDwav}. The basis functions all have support contained within $[-1,1]$, while still spanning $L^2[-1,1]$. Furthermore, the new multiresolution analysis retains the ability to reconstruct polynomials of order up to $p-1$ from the corresponding original multiresolution analysis. We shall not go into great detail here but we will outline the construction;  we take, along with a Daubechies wavelet $\psi$ and corresponding scaling function $ \phi$ with $ \mathrm{Supp} (\psi) = \mathrm{Supp} (\phi)=[-p+1,p]$, boundary scaling functions and wavelets (using the same notation as in \cite{dDwav} except that we use $[-1,1]$ instead of $[0,1]$ as our reconstruction interval)

\[ \phi^{\text{left}}_n, \ \phi^{\text{right}}_n, \ \psi^{\text{left}}_n , \ \psi^{\text{right}}_n , \qquad  n =0,\cdots,p-1 .\]
Like in the standard wavelet case we shift and scale these functions,
\[ \phi^{\text{left}}_{j,n}(x) = 2^{j/2} \phi^{\text{left}}_{n}(2^j (x+1)), 
\qquad \phi^{\text{right}}_{j,n}(x)= 2^{j/2} \phi^{\text{right}}_{n}(2^j (x-1)). \]
We are then able to construct nested spaces ,
$ (V^{\text{int}}_j)_{j \ge J}$, $ (W^{\text{int}}_j)_{j \ge J}$ for a fixed base level $J \ge \lceil \log_2 (p) \rceil $, such that \mbox{$L^2([-1,1])=\overline{ \bigoplus^{\infty}_{j=J} V^{\text{int}}_j}$}, $V^{\text{int}}_{j+1}=V^{\text{int}}_j \oplus W^{\text{int}}_j$ with $W^\text{int}_j$ the orthogonal complement of $V^\text{int}_j$ in $V^\text{int}_{j+1}$ by defining
\begin{equation*}
 V^{\text{int}}_j = \overline{ \text{Span} 
\left \{ \begin{aligned}
 \phi^{\text{left}}_{j,n} & , \phi^{\text{right}}_{j,n}  \\ & \phi_{j,k} 
 \end{aligned} 
 :  
\begin{aligned}
 & n =0 , \cdots , p-1 \ \\  & k \in \mathbb{Z} \ s.t. \ \mathrm{Supp}(  \phi_{j,k} ) \subset (-1,1)
 \end{aligned} 
 \right \} } ,
 \end{equation*}
 
 \begin{equation*} 
 W^{\text{int}}_j = \overline{ \text{Span} 
\left \{ \begin{aligned}
 \psi^{\text{left}}_{j,n} & , \psi^{\text{right}}_{j,n}  \\ & \psi_{j,k} 
 \end{aligned} 
 : 
\begin{aligned}
 & n =0 , \cdots , p-1 \ \\  & k \in \mathbb{Z} \ s.t. \ \mathrm{Supp}(  \psi_{j,k} ) \subset (-1,1)
 \end{aligned} 
 \right \} } .
 \end{equation*}
 
We then take the spanning elements of $V^{ \text{int}}_J$ and the spanning elements of $W^{\text{int}}_j$ for every $j \ge J$ to form the basis $B_{\rb \rw}$ ($\rb \rw$ for 'boundary wavelets').
\begin{definition}[Leveled ordering (boundary wavelets)]
Define $F_w: B_{\rb \rw} \to \mathbb{R}$ by the formula
\[ F_{\rb \rw}( f) \ = \ \begin{cases} 
\ j, \  & \mbox{if }  f \in W^{\text{int}}_j\\
\ -1,  \ & \mbox{if }f \in V^{\text{int}}_J 
\end{cases}.  \]
Then we say that an ordering $\tau: \mathbb{N} \to B_{\rb \rw}$ of this basis is a `leveled ordering' if it is consistent with $F_{\rb \rw}$.
\end{definition}

 \subsection{Legendre Polynomial Basis} \label{polynomialbasis}
If $(p_n)_{n \in \mathbb{N}}$ denotes the standard Legendre polynomials on $[-1,1]$ (so $p_n(1)=1$ and $p_1(x)=1$ for $x \in [-1,1]$) then the $L^2$-normalised Legendre polynomials are defined by $\tilde{p}_n=\sqrt{n-1/2} \cdot p_n$ and we write $B_\rp := (\tilde{p}_n )_{n=1}^\infty$  (the $\rp$ here stands for ``polynomial'' ). $B_\rp$ is already ordered; call this the \emph{natural ordering} . 

\subsection{Incoherence Results for One-dimensional Bases} \label{1DResults}

Next we recall the one-dimensional incoherence results proved in \cite{onedimpaper}, which shall be used to prove the corresponding multi-dimensional tensor results in Section \ref{tensors}:

\begin{theorem} \label{FourierWaveletResults}
Let $\rho$ be a standard ordering of $B_\rf(\epsilon)$ with $\epsilon \in I_{J,p}$, $\tau$ a leveled ordering of $B_\rw$ and $U=[(B_\rf(\epsilon),\rho),(B_\rw,\tau)]$. Then we have, for some constants $C_1, C_2>0$ the decay
\be{ \label{FourierWaveletOptimalBounds}  
\frac{C_1}{N} \le \mu(\pi_N U), \ \mu(U \pi_N) \le \frac{C_2}{N}, \qquad \forall N \in \mathbb{N} ,
}
The same conclusions also hold if the basis $B_\rw$ is replaced by $B_{\rb \rw}$ and the condition $\epsilon \in I_{J,p}$ by $\epsilon \in (0,1/2]$.
\end{theorem}

\begin{theorem} \label{FourierPolynomialResults}
Let $\rho$ be a standard ordering of $B_\rf(\epsilon)$ with $\epsilon \in (0,0.45], $ $\tau$ a natural ordering of $B_\rp$ and $U=[(B_\rf(\epsilon),\rho),(B_\rp,\tau)]$. Then we have, for some constants $C_1, C_2>0$ the decay
\be{ \label{FourierPolynomialOptimalBounds}  
\frac{C_1}{N^{2/3}} \le \mu(\pi_N U), \ \mu(U \pi_N), \le \frac{C_2}{N^{2/3}}, \qquad \forall N \in \mathbb{N} .
}
\end{theorem}

\section{Multidimensional Tensor Cases: Proof of Theorem \ref{tensormainwavelet} and Theorem \ref{tensormainpoly} } \label{tensors} 

In this section we prove Theorem \ref{tensormainwavelet} and Theorem \ref{tensormainpoly}. In fact, we state and prove their slightly more involved generalisations: Theorems \ref{TensorResultsWavelet} and \ref{TensorResultsPoly}. We also provide examples of hyperbolic orderings.

\subsection{General Estimates}
\begin{definition}[Tensor basis]
Suppose that $B$ is an orthonormal basis of some space $T \le L^2 (\mathbb{R})$ (i.e. $T$ is a subspace $L^2 (\mathbb{R})$) and we already have an ordering $\rho: \mathbb{N} \to B$. Define $\rho^d: \mathbb{N}^d \to \bigotimes_{j=1}^d T  \le L^2 (\mathbb{R}^d)$ by the formula ($m \in \mathbb{N}^d$)
\[ \rho^d(m)(x):= \Big( \bigotimes_{j=1}^d \rho(m_j) \Big) (x) = \prod_{j=1}^d \rho(m_j)(x_j).   \]
This gives a basis of $\bigotimes_{j=1}^d T  \le L^2 (\mathbb{R}^d)$ because of the formula
\begin{equation} \label{prodsplit}
 \langle \rho^d (m),\rho^d(n) \rangle_{L^2(\mathbb{R}^d)} = \prod_{j=1}^d \langle \rho(m_j), \rho(n_j) \rangle_{L^2(\mathbb{R})}.
\end{equation}
We call $B^d:=(\rho^d(m))_{m \in \mathbb{N}^d}$ a `tensor basis'. The function $\rho^d$ is said to be the `d-dimensional indexing induced by $\rho$'. Notice that $\rho^d$ is not an ordering unless $d=1$.
\end{definition}

Now suppose that we have two one-dimensional bases $B_1$, $B_2$ with corresponding optimal orderings $\rho_1, \rho_2$. Let $\rho^d_1, \rho^d_2$ be the d-dimensional indexings induced by $\rho_1,\rho_2$ of the bases $B^d_1,B^d_2$. What are optimal orderings of the basis pair $(B^d_1,B^d_2)$ and what is the resulting optimal decay rate? Some insight is given by the following Lemma:
\begin{lemma} \label{generaltensor}
Let $(B_1,B_2)$ be a pair of bases with corresponding tensor bases $B_1^d, B_2^d$. 
Let $\rho_1$ be a strongly optimal ordering of $B_1$ and $\rho_1^d$ be the $d$-dimensional indexing induced by $\rho_1$. Finally, for some ordering $\tau$ of $B_2$, let $U=[(B_1, \rho_1), (B_2, \tau)]$ . Then if $f$ represents the optimal decay rate corresponding to the basis pair $(B_1,B_2)$ we have, for some constants $C_1, C_2>0$,
\be{ \label{generaltensorequation}
\prod_{i=1}^d C_1^d \cdot f(n_i) \le \sup_{g \in B^d_2} | \langle \rho_1^d (n) , g \rangle |^2 = \prod_{i=1}^d \mu(\pi_{n_i} U) \le \prod_{i=1}^d C_2^d \cdot f(n_i), \quad n \in \bbN^d.
}
Consequently, if we let $\iota := \rho_1^d$ then $F(n):=\prod_{i=1}^d f(n_i)$ characterizes the optimal decay of $(B_1,B_2)$. 

\end{lemma}
\begin{proof}
Let $\tau^d$ denote the $d$-dimensional indexing induced by $\tau$. Then by breaking the down the tensor product into terms and using the bijectivity of $\tau^d$ we have 
\[
\begin{aligned}
\sup_{g \in B^d_2} | \langle \rho_1^d (n) , g \rangle |^2 & = \sup_{m \in \bbN^d} | \langle \rho_1^d (n) , \tau^d(m) \rangle |^2 = \sup_{m \in \bbN^d} \prod_{i=1}^d | \langle \rho_1 (n_i) , \tau(m_i) \rangle |^2
\\ & = \prod_{i=1}^d \sup_{m \in \bbN} | \langle \rho_1 (n_i) , \tau(m) \rangle |^2 = \prod_{i=1}^d \mu(\pi_{n_i} U).
\end{aligned}
\]
Therefore (\ref{generaltensorequation}) follows from applying the definition of a strongly optimal ordering to each term in the product.
\end{proof}

Lemma \ref{generaltensor} says that if we have a strongly optimal ordering for the basis pair $(B_1,B_2)$ then we can use Lemma \ref{characterisationlemma} to find all strongly optimal orderings for the corresponding tensor basis pair $(B_1^d,B_2^d)$. In particular, we have

\begin{corollary} \label{generaltensorcorollary}
We use the framework of the previous Lemma. Let $\sigma:\bbN \to \bbN^d$ be consistent with $1/F$. Then an ordering $\rho$ is strongly optimal for the basis pair $(B^d_1,B^d_2)$ if and only if there are constants $C_1,C_2>0$ such that
\[
C_1 F(\sigma(N)) \le F((\rho_1^d)^{-1} \circ \rho(N))  \le C_2 F(\sigma(N)), \quad N \in \bbN.
\]
\end{corollary}

Suppose that we have a strongly optimal ordering $\rho_1$ of $B_1$ such that the optimal decay rate is a power of $N$, namely that $f(n)=n^{-\alpha}$ for some $\alpha>0$, which is the case for the one dimensional examples we covered in Section \ref{onedim}. The above Lemma tells us that to find the optimal decay rate we should take an ordering $\sigma : \bbN \to \bbN^d$ that is consistent with $1/F(n):= \prod_{i=1}^d 1/f(n_i)= \prod_{i=1}^d n_i^{\alpha}$ which is equivalent to being consistent with $1/F^{1/\alpha}(n)=\prod_{i=1}^d n_i$. This motivates the following:

\begin{definition}[Corresponding to the hyperbolic cross]
Define $F_H: \mathbb{N}^d \to \mathbb{R}$ by $ F_{H}(n)= \prod_{i=1}^d n_i$.
Then we say a bijective function $\sigma: \mathbb{N} \to \mathbb{N}^d$ `corresponds to the hyperbolic cross' if it is consistent with $F_H$.
\end{definition}
The name `hyperbolic cross' originates from its use in approximation theory \cite{crossorig,hypcross}. We now claim that if $\sigma$ corresponds to the hyperbolic cross and $d \ge 2$, then 
\be{ \label{hyperbolicdecayrate}
\prod_{i=1}^d \sigma(N)_i \sim \frac{(d-1)! N}{\log^{d-1}(N+1)}  \quad \text{as } \quad N \to \infty.
}
Next we proceed to prove this claim.

\begin{definition}
For $d \in \mathbb{N}$ let $f_d(x) = x \log^{d-1} x$ be defined on $[1,\infty)$. We define
 $g_d$ as the inverse function of $f_d$ on $[1,\infty)$, and so $g_d: [ 0, \infty) \to [1,\infty)$. Furthermore, we define
 \be{ \label{hyperbolicdecay} 
 h_d(x):= \frac{x}{\log^{d-1}(x+1)}
 , \qquad x \in [1, \infty).}
\end{definition}

\begin{lemma} \label{ordersimplify}
The following holds:
\\
1.) \ \  
$g_d(x)/h_{d}(x) \to 1 \quad  \text{as} \quad x \to \infty.$ 
\\
2.) \ \ Let $\tilde{f}(x) = x \log^{d-1} x + x p( \log(x) )  + \beta$ with $p$ a polynomial of degree at most $d-2$, $ \beta \in \mathbb{R}$ and let $\tilde{g}$ be its inverse function defined for large $x \in \bbR_+$. Then we also have $\tilde{g}(x)/h_{d}(x) \to 1 \quad  \text{as} \quad x \to \infty.$
\end{lemma}
\begin{proof}
1.) \ \ For notational convenience we shall prove the equivalent result
\[ \frac{g_d(x) \log^{d-1}(x)}{x} \to 1 \quad \text{as} \quad x \to \infty. \]
By taking logarithms we change the problem from studying the asymptotics of a fraction to the asymptotics of the difference 
\be{ \label{equivalentasymp}
\log(g_d(x))-\log(h_d(x)) =  \log(g_d(x)) - \log x + (d-1) \log \log x \to 0 \quad \text{as} \quad x \to \infty. 
}
With this in mind we notice that the function $\log(g_d)$ (defined on $[0,\infty)$) is the  inverse function of $e_d(x):= f_d(\exp(x)) = x^{d-1} \exp x$ (defined on $[0,\infty)$). 

Notice that for $x$ large we have $e_d(x- (d-1) \log x)= \frac{(x- (d-1) \log x)^{d-1}}{x^{d-1}} \exp(x) \le \exp(x)$ which implies that $ x - (d-1) \log x \le \log(g_d( \exp(x))) $. Now if we let $\epsilon>0$ then we deduce that 
\[e_d(x- (d-1) \log x+ \epsilon)= \frac{(x- (d-1) \log x+\epsilon)^{d-1}}{x^{d-1}} \exp(x+\epsilon) \ge \exp(x) \quad \text{for} \quad  x \quad \text{large.} \]
This implies that $ x - (d-1) \log x + \epsilon \ge \log(g_d( \exp(x))) $ for $x$ large. We therefore conclude that for all $x$ sufficiently large we have
\[ x - (d-1) \log x \le \log(g_d( \exp(x))) \le x - (d-1) \log x + \epsilon,\]
from which (\ref{equivalentasymp}) follows since $\epsilon>0$ is arbitrary.

2.) \ \ Notice that by part 1. it suffices to show that $\tilde{g}(x)/g_{d}(x) \to 1 \ $ as $\ x \to \infty.$ Again, we shall show this by taking logarithms, reducing the proof to showing
\[ \log( \tilde{g}(x)) - \log( g_d(x)) \to 0 \quad \text{as} \quad  x \to \infty. \]
Notice that $\log( \tilde{g}(x))$ is the inverse function, defined for large $x$, of 
\[\tilde{e}(x):=\tilde{f}( \exp(x)) = x^{d-1} \exp(x) + p(x) \cdot \exp(x) + \beta, \] 
Then since 
 \[\tilde{e}'(x)= x^{d-1} \exp(x) + ((d-1) \cdot x^{d-2} +p'(x)+p(x)) \cdot \exp(x), \] 
 we can use the hypothesis that $p$ is of a lower order than $x^{d-1}$ to show that for every $\epsilon>0$, there is an $L(\epsilon)>0$ such that for all $x \ge L(\epsilon)$ we have $\epsilon \cdot \tilde{e}'(x -\epsilon) \ge |\tilde{e}(x) - e_d(x)|=|p(x) \cdot \exp(x) + \beta|$. We therefore deduce from the mean value theorem that for $x \ge \exp(L(\epsilon))$ we have
\begin{align*} 
 \tilde{e}(\log(g_d(x))-\epsilon)  \le e_d(\log(g_d(x)))= & x  \le  \tilde{e}( \log(g_d(x))+ \epsilon)   
 \\ & \Rightarrow   \log(g_d(x)) - \epsilon \le  \log(\tilde{g}(x))  \le \log(g_d(x)) + \epsilon,
\end{align*}
where we applied $\log(\tilde{g})$ to the inequality in the last step (this preserves the inequality since $\log(\tilde{g})$ is an increasing function of $x$ for $x$ large).
\end{proof}

\begin{lemma} \label{orderset}
1). \ \ For every $d \in \mathbb{N}$ we have
\be{ \label{hyplogestimate}
R_N:=\sum_{i=1}^N \frac{1}{i} (\log(N) - \log(i))^d \ = \ \frac{1}{d+1}  \log^{d+1} N + \mathcal{O}(\log^d N) \qquad N \to \infty.
}

2). \ \ Let $S_d(N)$ for $d, N  \in \mathbb{N}$ be defined by
\begin{equation} \label{simplehyperbolic} 
S_d(N):= \# \Big\{ m \in \mathbb{N}^d  :  \prod_{i=1}^d m_i  \le  N \Big\}. 
\end{equation}

Then for every $d \in \mathbb{N}$, there exists polynomials $ \underline{p}_d, \overline{p}_d$ both of degree $d-1$ with identical leading coefficient $1/(d-1)!$ such that
\begin{equation}   \label{hyperbolic_count}
 N \underline{p}_d( \log(N))   \le  S_d(N)  \le  N \overline{p}_d( \log(N)) .
 \end{equation}
 
 3). \ \ If we let $\sigma: \mathbb{N} \to \mathbb{N}^d$ correspond to the hyperbolic cross then (\ref{hyperbolicdecayrate}) holds.
\end{lemma}

\begin{proof}
1). \ \ Let $I_N:=\int_1^N \frac{1}{x} (\log(N) - \log(x))^d \, dx$. Since the integrand is a decreasing function of $x$ (with $N$ fixed) we find that by the Maclaurin integral test that
$0 \le R_N -I_N \le \log^d(N)$. This means that showing (\ref{hyplogestimate}) is equivalent to showing that
\[ \int_1^N \frac{1}{x} (\log(N) - \log(x))^d \, dx =  \frac{1}{d+1}  \log^{d+1} N + \mathcal{O}(\log^d N). \]
Now, by expanding out the factors of the integrand and integrating (recall that the integral of $x^{-1} \log^k x$ is $\frac{1}{k+1} \cdot \log^{k+1}x$) the integral becomes
\[ \log^{d+1}(N) \cdot \sum_{i=0}^d \frac{1}{i+1} \binom{d}{i} (-1)^{i}. \]
Since $\frac{1}{i+1} \binom{d}{i}  = \frac{1}{d+1} \binom{d+1}{i+1}$ we see that the sum simplifies to $ \frac{1}{d+1}$ and we are done.

2). \ \  We use induction on the dimension $d$. The case $d=1$ is immediate since $\underline{p}_1(x)=\overline{p}_1(x)=1$ satisfies inequality (\ref{hyperbolic_count}). Therefore suppose that inequality (\ref{hyperbolic_count}) holds for dimension $d=k$. We shall extend the result to $d=k+1$ using the equality:
\begin{equation} \label{hyperbolicdimreduce}
S_{k+1}(N)= \sum_{i=1}^N S_{k}\Big( \left\lfloor \frac{N}{i} \right\rfloor \Big).
\end{equation}
This equality follows from rewriting the set defining $S_{k+1}$ as the following disjoint union:
\[
 \Big\{ m \in \bbN^{k+1} : \prod_{i=1}^{k+1} m_i \le N \Big\} = \coprod_{j=1}^N \Bigg\{ m \in \bbN^{k+1} : m_{k+1}=j, \prod_{i=1}^k m_i \le \left\lfloor \frac{N}{i} \right\rfloor \Bigg\} .
\]

\textbf{Upper Bound:} We may assume without loss of generality that $\overline{p}_k$ has all coefficients positive. Therefore, by replacing $ \lfloor \frac{N}{i} \rfloor $ with $\frac{N}{i}$ and using the upper bound in (\ref{hyperbolic_count}), we can upper bound equation (\ref{hyperbolicdimreduce}) by
\[ \sum_{i=1}^N \frac{N}{i}   \cdot  \overline{p}_k \Big( \log \Big( \frac{N}{i} \Big) \Big) \ \le \ N \sum_{i=1}^N \frac{1}{i} \cdot  \overline{p}_k ( \log(N) - \log(i)). \]
We can then get the required upper bound by applying part 1) of the lemma to each term in the polynomial; for example the highest order term becomes
\[
\begin{aligned} 
 \sum_{i=1}^N \frac{N}{i} \cdot \frac{1}{(k-1)!}  ( \log(N) - \log(i))^{k-1} \le \frac{N }{k!} \log^{k}N + C N \log^{k-1} N, \qquad \forall N \in \mathbb{N},         
 \end{aligned}   
 \]
 for some constant $C>0$ sufficiently large. The other terms in $\overline{p}_k$ are handled similarly. 

\textbf{Lower Bound:} Notice that without loss of generality we can assume all the coefficients of $\underline{p}_k$ apart from the leading coefficient are negative. Using the lower bound in (\ref{hyperbolic_count}), we can lower bound equation (\ref{hyperbolicdimreduce}) by
\[ \sum_{i=1}^N \left\lfloor \frac{N}{i} \right\rfloor   \cdot  \overline{p}_k \Big( \log \Big( \left\lfloor \frac{N}{i} \right\rfloor \Big) \Big). \]
This means we can tackle the $<k-1$ order terms in the same way as in the upper bound since we can replace $ \left\lfloor \frac{N}{i} \right\rfloor $ with $\frac{N}{i}$ (recall we have assumed these terms are negative). Now we are left with bounding the highest order term:
\begin{equation}
\begin{aligned}
   \sum_{i=1}^N  \left\lfloor \frac{N}{i} \right\rfloor \frac{1}{(k-1)!}   (\log \Big( \left\lfloor \frac{N}{i} \right\rfloor \Big) )^k =   \sum_{i=1}^N  \left\lfloor \frac{N}{i} \right\rfloor  \frac{1}{(k-1)!} \cdot \Big[ \log \Big(  \frac{N}{i}  \Big)  -  \big( \log \Big(   \frac{N}{i}  \Big)  - \log \Big( \left\lfloor \frac{N}{i} \right\rfloor \Big) \big) \Big]^k.   
\end{aligned}                  
\end{equation}
Therefore expanding out the binomial term, setting the sign of all terms except the first to be negative, and noticing $\log \Big(   \frac{N}{i}  \Big)  - \log \Big( \left\lfloor \frac{N}{i} \right\rfloor \Big) \le 1$ for every $i,N$ we get the lower bound
\[                          
\begin{aligned}
\sum_{i=1}^N & \left\lfloor \frac{N}{i}  \right\rfloor \frac{1}{(k-1)!} \log^{k}  \Big(  \frac{N}{i}  \Big)  - \sum_{i=1}^N  \sum_{j=0}^{k-1} \left\lfloor \frac{N}{i} \right\rfloor  \binom{k}{j} \frac{1}{(k-1)!}  \log^{j} \Big(  \frac{N}{i}  \Big).
\end{aligned}
 \]
From here we can replace $\left\lfloor \frac{N}{i} \right\rfloor$ by $ \frac{N}{i} $ for the right term, $\left\lfloor \frac{N}{i} \right\rfloor$ by $\frac{N}{i}  -1 $  on the left term and use part 1) of the lemma again to prove the lower bound.
 
3.) \ \  From the second part of the lemma we know that for some degree $d-1$ polynomials $\underline{p}_d , \overline{p}_d$ with leading coefficient $1/(d-1)!$ we have
$ N \underline{p}_d (\log(N)) \le S_d(N) \le N \overline{p}_d ( \log(N)). $
 Now notice that if $m \in \mathbb{N}$ then because of consistency we must have $ S_d( F_H(\sigma(m))-1) \le m$ since $\sigma$ must first list all the terms $n$ in $\bbN^d$ with $F_H(n) \le F_H(\sigma(m))-1$ before listing $\sigma(m)$. Likewise we must have $m \le S_d( F_H(\sigma(m)))$ since the $S_d( F_H(\sigma(m)))$ terms with $F_H(n) \le F_H(\sigma(m)), n \in \bbN^d$ must be listed by $\sigma$ first, including $m$ , before any others. Consequently we deduce
 \begin{equation} \label{productbound}
 \begin{aligned}
 (F_H(\sigma(m)) -1) \underline{p}_d (\log(F_H(\sigma(m)) & -1))  \le m \le F_H(\sigma(m)) \overline{p}_d ( \log( F_H( \sigma(m)) )).
\end{aligned}
\end{equation}
 We now treat both sides separately. Looking at the LHS we get the estimate
 $F_H(\sigma(m)) -1 \le \tilde{g}_d(m),$
 where $\tilde{g}_d(m)$ is the inverse function (defined for large $m$) of 
 \[ 
 \begin{aligned}
 \tilde{f}_d(x):= & \frac{1}{(d-1)!} x \log^{d-1}(x) + \mbox{(degree $d-2$ poly)}(\log(x)),
\end{aligned}
\] 
 and so we may apply part 2. of Lemma \ref{ordersimplify} to deduce
 $ F_H(\sigma(m)) \le h_d( (d-1)! m ) \cdot (1 + \epsilon(m)),$
 where $\epsilon(m) \to 0$ as $m \to \infty$. The right hand side is handled similarly to get the same asymptotic lower bound on $F_H(\sigma(m))$, namely
$ F_H(\sigma(m)) \ge h_d( (d-1)! m ) \cdot (1 + \epsilon(m)),$ where $\epsilon(m) \to 0$ as $m \to \infty$. Since $\frac{h_d((d-1)! x)}{(d-1)! h_d(x)} \to 1$ as $x \to \infty$ the proof is complete.
\end{proof}

(\ref{hyperbolicdecayrate}) allows us to determine the optimal decay rate for when the optimal one dimensional decay rate is a power of $N$. 

\begin{theorem} \label{generaltensortheorem}
Returning to the framework of Corollary \ref{generaltensorcorollary}, if $f(n)=n^{- \alpha}$ for $n \in \bbN$, $F(n)= \prod_{i=1}^d f(n)$ for $n \in \bbN^d$ and $\sigma : \bbN \to \bbN^d$ corresponds to the hyperbolic cross then 
\be{ \label{finaltensordecay}
F(\sigma(N)) = \Bigg( \prod_{i=1}^d \sigma(N)_i \Bigg)^{- \alpha} \sim \big((d-1)! \cdot h_d(N) \big) ^{-\alpha}, \quad N \to \infty.
}
Consequently $h_d^{-\alpha}$ is representative of the optimal decay rate for the basis pair $(B_1^d,B_2^d)$. Furthermore, an ordering $\rho$ is strongly optimal for the basis pair $(B_1^d,B_2^d)$ if and only if there are constants $C_1, C_2>0$ such that
\begin{equation}  \label{hyperbolicdef}
C_1 \cdot h_d(N) \le \prod_{i=1}^d \Big((\rho_1^d)^{-1} \circ \rho(N) \Big)_i \le C_2 \cdot h_d(N), \quad N \in \bbN.
\end{equation}
\end{theorem}
\begin{proof}
(\ref{finaltensordecay}) follows immediately from (\ref{hyperbolicdecayrate}). This implies that $F \circ \sigma \approx h_d^{-\alpha}$.  The statement on the optimal decay rate then follows from the characterization result from Lemma \ref{generaltensor} applied to Lemma \ref{characterisationlemma}. The statement on strongly optimal orderings follows from Corollary \ref{generaltensorcorollary}.
\end{proof}

\begin{definition} \label{hyperbolicordering}
Using the framework of Lemma \ref{generaltensor}, any ordering $\rho: \bbN \to B_1^d$ such that (\ref{hyperbolicdef}) holds is called a 'hyperbolic' ordering. with respect to $\rho_1$. Notice that by (\ref{finaltensordecay}) that if $\sigma: \bbN \to \bbN^d$ corresponds to the hyperbolic cross then $\rho_1^d \circ \sigma$ is hyperbolic with respect to $\rho_1$.
\end{definition}

We now apply Theorem \ref{generaltensortheorem} to the one-dimensional cases we have already covered:

\subsection{Fourier-Wavelet Case}

\begin{theorem} \label{TensorResultsWavelet}
We use the setup of Lemma \ref{generaltensor}. 
Suppose $B_1=B_\rf(\epsilon)$, $B_2=B_\rw$ for some fixed $\epsilon \in I_{J,p}$, $\rho_1$ is a standard ordering of $B_1$ and $\tau_1$ is a leveled ordering of $B_2$. Let $U_d=[(B_1^d, \rho), (B_2^d, \tau)]$ where $\rho, \tau$ is hyperbolic with respect to $\rho_1, \tau_1$ respectively. Then we have, for some constants $C_1,C_2>0$,
\begin{equation}\label{FourWave1}
\frac{C_1 \log^{d-1}(N+1)}{N} \le \mu( \pi_N U_d ), \ \mu(U_d \pi_N) \le \frac{C_2 \log^{d-1}(N+1)}{N}, \qquad N \in \bbN.
\end{equation}
The above also holds if the basis $B_\rw$ is replaced by $B_{\rb \rw}$ and the condition $\epsilon \in I_{J,p}$ by $\epsilon \in (0,1/2]$.
\end{theorem}

\begin{proof}
Inequality (\ref{FourWave1}) follows from applying Theorem \ref{FourierWaveletOptimalBounds} to Theorem \ref{generaltensortheorem}.
\end{proof}

\subsection{Fourier-Polynomial Case}

\begin{theorem} \label{TensorResultsPoly}
We use the setup of Lemma \ref{generaltensor}. Suppose $B_1=B_\rf(\epsilon)$, $B_2=B_\rp$ for some fixed $\epsilon \in (0,0.45]$, $\rho_1$ is a standard ordering of the Fourier basis and $\tau_1$ is the natural ordering of the polynomial basis. Let $U_d=[(B_1^d, \rho), (B_2^d, \tau)]$ where $\rho, \tau$ is hyperbolic with respect to $\rho_1, \tau_1$ respectively. Then we have, for some constants $C_1,C_2>0$, that
\begin{equation}\label{FourLeg1}
 \frac{C_1 (\log^{d-1}(N+1))^{2/3}}{N^{2/3}} \le \mu( \pi_N U_d ), \ \mu(U_d \pi_N) \le \frac{ C_2 (\log^{d-1}(N+1))^{2/3}}{N^{2/3}}, \qquad N \in \bbN.
\end{equation}
\end{theorem}

\begin{proof} 
Inequality (\ref{FourLeg1}) follows from applying Theorem \ref{FourierPolynomialOptimalBounds} to Proposition \ref{generaltensortheorem}.
\end{proof}

\subsection{Examples of Hyperbolic Orderings} \label{hyperbolicexamples}

The generalisation introduced by Definition \ref{hyperbolicordering}, apart from allowing us to characterise all orderings that are strongly optimal, may seem to fulfil little other purpose. However, as we shall see in this section, this definition admits orderings which in specific cases are very natural and appear a little less abstract than an ordering derived from the hyperbolic cross.

\begin{example} \label{hypcrossZd} (Hyperbolic Cross in $\mathbb{Z}^d$)
Our first example is unremarkable but nonetheless important. In $d$ dimensions, take $B_1^d:=B_\rf^d$ as a d-dimensional tensor Fourier basis. Recall we can identify this basis with $\mathbb{Z}^d$ using the function $\lambda_d$. Suppose that we define a function $H_d: \mathbb{Z}^d \to \mathbb{R}$ by
\be{ \label{hddef}
H_d(m) = \prod_{i=1}^d | \max(|m_i|,1)|, 
}
and say that a bijective function $\sigma: \mathbb{N} \to \mathbb{Z}^d$ `corresponds to the hyperbolic cross in $\mathbb{Z}^d$' if it is consistent with $H_d$. Figure \ref{hyperbolic} shows the first few contour lines of $H_d$ in two dimensions.
With this definition we can then prove the analogous result of Lemma \ref{orderset}:

\begin{figure}

\centering
    \includegraphics[width=0.6\textwidth]{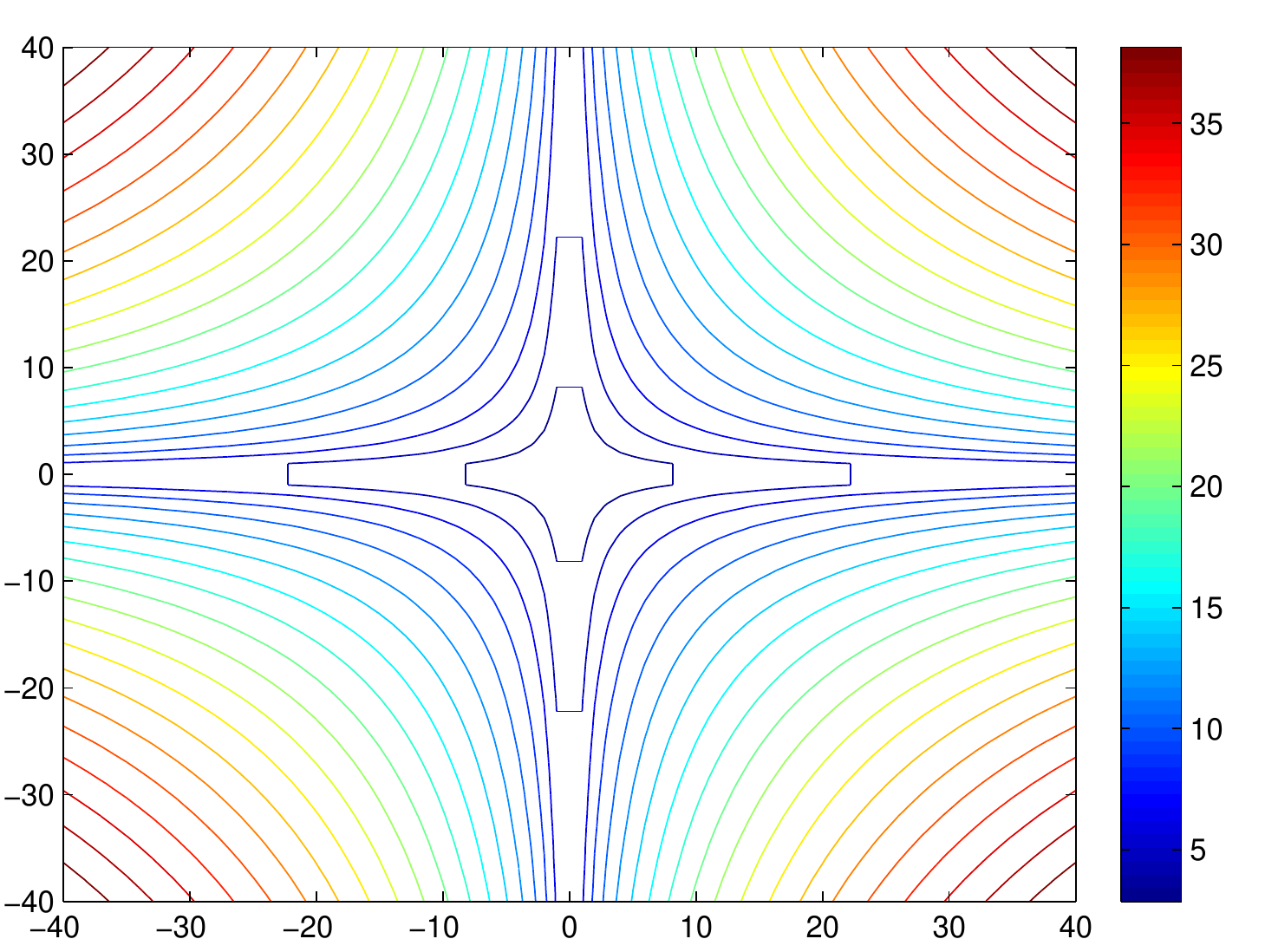} 
      \caption{Hyperbolic Fourier Ordering in Two Dimensions: A Contour Plot of $H_2$}
      \label{hyperbolic}
\end{figure}

\begin{lemma} \label{orderset2}
Let $\sigma: \mathbb{N} \to \mathbb{Z}^d$ correspond to the hyperbolic cross and let $h_d$ be as in (\ref{hyperbolicdecayrate}). Then we have
\be{ \label{hyperboliccrossZdecay}
\prod_{i=1}^d | \max(| \sigma(m)_i|,1)| \sim \frac{(d-1)!}{2^d} \cdot h_d(m)  \quad \mbox{as} \quad m \to \infty.  
}
Moreover, if $\rho_1$ is a standard ordering of $B_\rf$ and $\sigma: \mathbb{N} \to \mathbb{Z}^d$ corresponds to the hyperbolic cross. Then $\lambda_d^{-1} \circ \sigma$ is a hyperbolic ordering with respect to $\rho_1$.
\end{lemma}

\begin{proof}
Let $R_d(n)$ denote the number of lattice points in the hyperbolic cross of size $n$ in $\bbZ^d$, namely
 \[ R_d(n) := \# \{ m \in \mathbb{Z}^d  :  \prod_{i=1}^d \max(| m_i|,1)  \le n \}.  \]
Call the set in the above definition $ \mathcal{H}_d(n)$. If we remove the hyperplanes $\{m_i=0 \}$ for every $i$ from $\mathcal{H}_d(n)$, we are left with $2^d$ quadrants in $\mathbb{Z}^d$ which are congruent to set in equation (\ref{simplehyperbolic}). From the second part of Lemma \ref{orderset} we therefore have
 \[ R_d(n)  \ge  2^d n \underline{p}_d( \log(n)). \]
 Next notice that the intersection of $\mathcal{H}_d(n)$ with each hyperplane $\{ m_i=0 \}$ can be identified with $\mathcal{H}_{d-1}(n)$ and so we also have the upper bound 
 \[ 
 \begin{aligned}
 R_d(n) & \le 2^d  n \overline{p}_d( \log(n)) + d \cdot R_{d-1}(n)  \quad \Rightarrow \quad R_{d}(n) \le n \overline{r}_d( \log (n)),
 \end{aligned}
 \]
 for some degree $d-1$ polynomial $\overline{r}_d$ with leading coefficient $\frac{2^{d}}{(d-1)!}$. Combining the upper and lower bounds we see that for some polynomials $\underline{r}_d, \overline{r}_d$ of degree $d-1$ with leading coefficient $ \frac{2^{d}}{(d-1)!}$ we have 
 \[ n \underline{r}_d (\log(n)) \le R_d(n) \le n \overline{r}_d( \log(n) ). \]
 Therefore for $m \in \mathbb{N}$ since 
 \[R_d(H_d(\sigma(m))-1) \le m \le R_d(H_d(\sigma(m))),\] 
 we have
 \[ 
\begin{aligned} 
 (H_d(\sigma(m))-1) & \underline{r}_d (\log(H_d(\sigma(m))-1)) \le m \le H_d(\sigma(m)) \overline{r}_d (\log(H_d(\sigma(m)))).
 \end{aligned} 
 \]
Consequently we can apply Lemma \ref{ordersimplify} to both sides to derive (\ref{hyperboliccrossZdecay}) like in the proof of Lemma \ref{orderset}.
 For the last part of the Lemma notice that since $\rho_1$ is a standard ordering then $\max(|\lambda_1 \circ \rho_1(N)|,1) \approx N$. This means that the bounds on $\mu(\pi_N U)$ in Theorem \ref{FourierWaveletResults} can be rephrased as (for some constants $C_1,C_2>0$)
 
 \[
 C_1 \cdot (\max(|n|,1))^{-1} \le \sup_{g \in B_\rw} | \langle \lambda_1^{-1}(n), g \rangle |^2 \le  C_2 \cdot (\max(|n|,1))^{-1}, \quad n \in \bbZ,
 \]
 and by Lemma \ref{generaltensor} this extends to the dD tensor case:
  \be{ \label{zhypcrosscharacterise}
 C^d_1 \cdot \prod_{i=1}^d (\max(|n_i|,1))^{-1} \le \sup_{g \in B^d_\rw} | \langle \lambda_d^{-1}(n), g \rangle |^2 \le  C^d_2 \cdot \prod_{i=1}^d (\max(|n_i|,1))^{-1}, \quad n \in \bbZ^d.
 }
 This describes a characterization of the optimal decay of $(B_\rf^d(\epsilon),B^d_\rw)$. Lemma \ref{characterisationlemma} tells us that $\lambda_d^{-1} \circ \sigma$ is strongly optimal for $(B_\rf^d(\epsilon),B^d_\rw)$, which by Theorem \ref{generaltensortheorem} is hyperbolic with respect to $\rho_1$.
\end{proof}
\end{example}

\begin{example} (Tensor Wavelet Ordering)
Now we look at an example of a less obvious hyperbolic ordering. We first introduce some notation to describe a tensor wavelet basis:
For $j \in \mathbb{N}, \ k \in \mathbb{Z}$ let $\phi^0_{j,k}:= \phi_k, \ \phi^1_{j,k}:= \psi_{j,k}$. Now for $s \in \{ 0,1 \}^d , \ j \in \mathbb{N}^d, \ k \in \mathbb{Z}^d$ define
\[ \Psi^s_{j,k}  :=  \bigotimes_{i=1}^d \phi^{s_i}_{j_i,k_i}.\] 
Then it follows that for $J \in \bbN$ fixed, we have
\begin{align} \label{tensorwaveletbasisdefine}
B^d_\rw := \left\{ \begin{array}{cc}   &  \mathrm{Supp}(\phi^{s_i}_{j_i,k_i}) \cap (-1,1) \neq \emptyset \quad \forall i , \\  \Psi^s_{j,k} :  &  s_i=0 \Rightarrow j_i=J, \quad s_i=1 \Rightarrow j_i\ge J \\ &  j \in \mathbb{N}^d, s \in \{ 0,1 \}^d , \ k \in \mathbb{Z}^d  \end{array} \right \} , 
\end{align}
 The same approach can be applied to the boundary wavelet basis $B_{\rb \rw}$ to generate a boundary tensor wavelet basis $B^d_{\rb \rw}$, although we must include the extra boundary terms, which can be done by letting $s \in \{ 0,1,2,3 \}^d$ where $\phi^2_{J,n}$ would be a boundary scaling function term and $\phi^3_{j,n}$ a boundary wavelet term.
\begin{lemma} \label{tensorwavelethyp}
Let $\rho_1$ be any leveled ordering of a one-dimensional Haar wavelet basis $B^d_{\rw}$. Setting $\overline{j}=\sum_{i=1}^d j_i$ define $F_\text{hyp}:B^d_{\rw} \to \mathbb{R}$ by the formula
\[ F_\text{hyp}( f) =  \overline{j} \quad   \mbox{if }  \quad f= \Psi^s_{j,k}. 
  \]
Then any ordering $\rho: \mathbb{N} \to B^d_{\rw}$ that is consistent with $F_\text{hyp}$ is a hyperbolic ordering with respect to $\rho_1$.
\end{lemma}
\begin{remark}
Such an ordering $\rho$ is used to implement a tensor wavelet basis in Section \ref{numericalsection}.
\end{remark}
\begin{remark}
For the sake of simplicity we only work with the Haar wavelet case, although we could cover the boundary wavelet case with the same argument.
\end{remark}
\begin{proof}
By recalling inequality (3.10) in \cite{onedimpaper} or by using Lemma \ref{levelgrowth} in the case $d=1$ we know that there are constants $C_1, C_2>0$ such that for $\rho_1(N) = \phi^{s}_{j,k}, s \in \{0,1\}, j \in \bbN, k \in \bbZ$,
\[
C_1 2^{ j} \le N \le C_2 2^{j}, \qquad N \in \bbN.
\]
Therefore, writing $\rho_1^d(m) = \Psi^{s(m)}_{j(m),k(m)}$,
\[ C^d_1 2^{ \overline{j(m)}} \le \prod_{i=1}^d m_i \le C^d_2 2^{ \overline{j(m)}}, \qquad m \in \bbN^d. \]
Consequently if we rewrite this with an actual ordering $\rho(N)= \Psi^{s(N)}_{j(N),k(N)}$ for $N \in \mathbb{N}$ we deduce
\begin{equation} \label{down2j} 
C^d_1 2^{ \overline{j(N)}} \le \prod_{i=1}^d \Big( (\rho_1^d)^{-1} \circ \rho(N) \Big)_i \le C^d_2 2^{ \overline{j(N)}} ,
\end{equation}
and so we have reduced the problem to determining how $\overline{j(N)}$ scales with $N$.  
Notice that from our ordering of the wavelet basis that $\overline{j(N)}$ is a monotonically increasing function in $N$ and moreover, for every value of $\overline{j(N)}$ there are $r_d(\overline{j(N)})2^{ \overline{j(N)}}$ terms in $B_w^d$ with this value of  $\overline{j(N)}$ in the wavelet basis, where
\[ 
\begin{aligned}
r_d(N):=  \# \Bigg\{ (j,s) \in \mathbb{N}^d \times \{0,1 \}^d : \quad \overline{j} = N , \quad j_i \ge J, \quad  (s_i-1)(j_i-J)=0 \quad \forall i=1,...,d \Bigg\},
\end{aligned}
\] 
This is where we are using that the support of the Haar wavelet is $[0,1]$ and so there are $2^j$ shifts of $\phi_{j,0}, \psi_{j,0}$ in $B_\rw$. Notice that $r_d(N)$ is a polynomial of degree $d-1$. With this in mind notice we can define, consistent for $n \in \mathbb{N}, \ n \ge J$,
\begin{align*} 
T_d(x) & ``=" \sum_{i=J}^x r_d(i)2^{i} := \ p_d(x) 2^{x} + \alpha_d ,
\end{align*}
for some degree $d-1$ polynomial $p_d$ and constant $\alpha_d$. This is possible by taking the formula for the geometric series expansion and differentiating repeatedly. By the consistency property of $\rho$ we deduce the inequality
\begin{align*}
 T_d( \overline{j(N)}-1)  \le  N   \le  T_d(\overline{j(N)}) \quad &\Rightarrow \quad \overline{j(N)}-1  \le  T_d^{-1}(N)  \le  \overline{j(N)} \\
&\Rightarrow  \quad 2^{\overline{j(N)}-1}  \le  2^{T_d^{-1}(N)}  \le  2^{\overline{j(N)}}.
\end{align*}
Notice that $2^{T_d^{-1}(x)}$ is the inverse function of $T_d( \log_2 x)$ which is of the form 
$ x \cdot p_d( \log_2x) + \alpha,$
Therefore, applying parts 2. \& 3. of Lemma \ref{ordersimplify} gives, for some constants $D_1,D_2 >0$ and $N$ large,
\begin{equation} \label{gasymp}
  \quad (1+ \epsilon_1(N)) \cdot D_1 \cdot 2^{\overline{j(N)}}  \le  g_d(N)   \le (1+\epsilon_2(N)) \cdot  D_2 \cdot 2^{\overline{j(N)}} ,
  \end{equation}
where $\epsilon_1(N), \epsilon_2(N) \to 0$ as $N \to \infty$. Combining this with (\ref{down2j}) shows that we have a hyperbolic ordering.
\end{proof}

\end{example}

\subsection{Plotting Tensor Coherences}

Let us consider a simple illustration of this theory applied to a 2D tensor Fourier-Wavelet case $(B^2_\rf,B^2_\rw)$. We can identify the 2D Fourier Basis $B^2_\rf$ with $\bbZ^2$ using the function $\lambda_2$,  so the row incoherences can also be identified with $\bbZ^2$ and therefore they can be imaged directly in 2D, as in Figure \ref{tensorincoherences}. 

\begin{figure}[!t]
\begin{center}
\begin{subfigure}[t]{0.32\textwidth}
\begin{center}
\includegraphics[width=\textwidth]{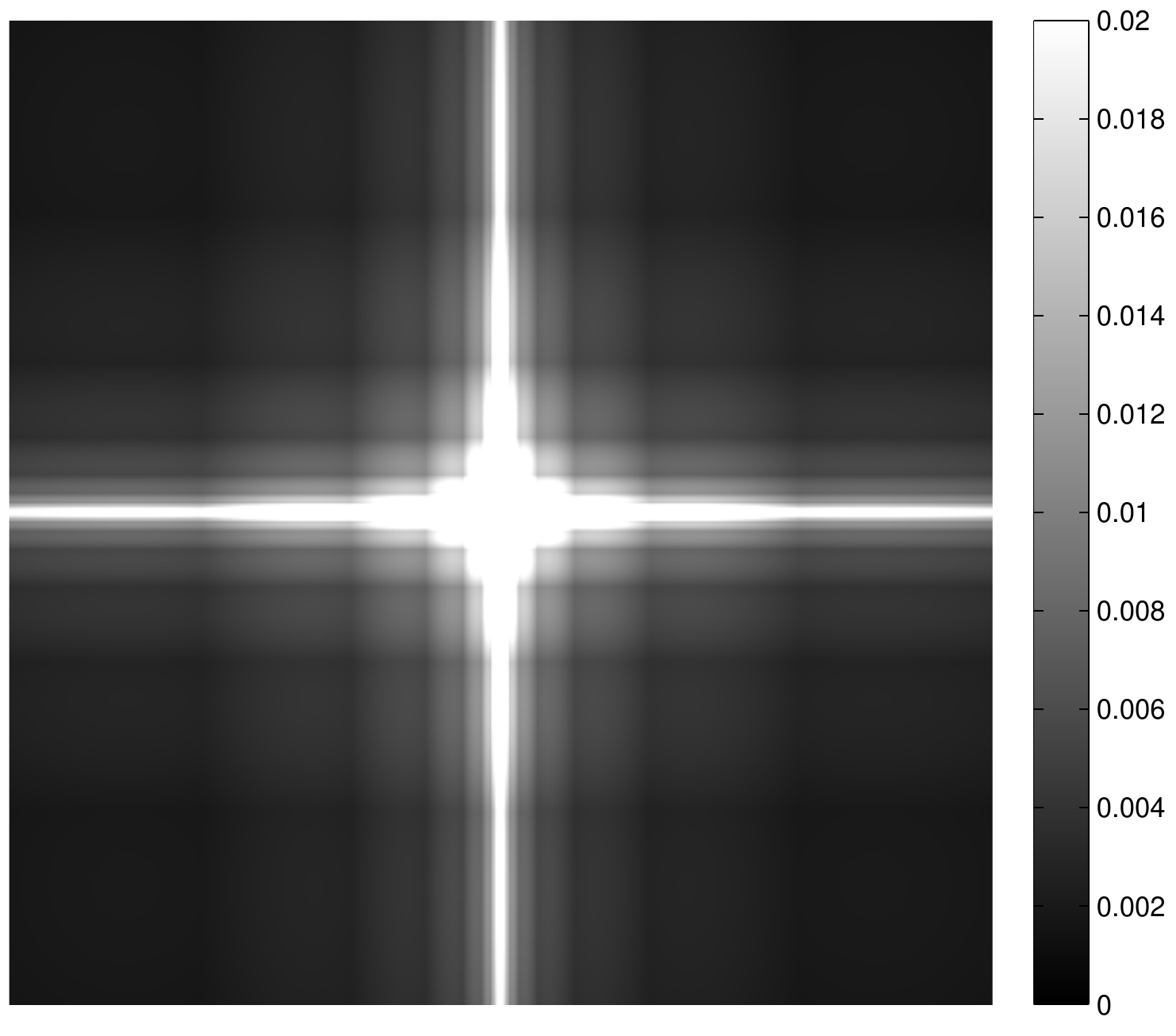}
 \caption{\footnotesize Original Coherences}
\end{center}
\end{subfigure}
\begin{subfigure}[t]{0.32\textwidth}
\begin{center}
 \includegraphics[width=\textwidth]{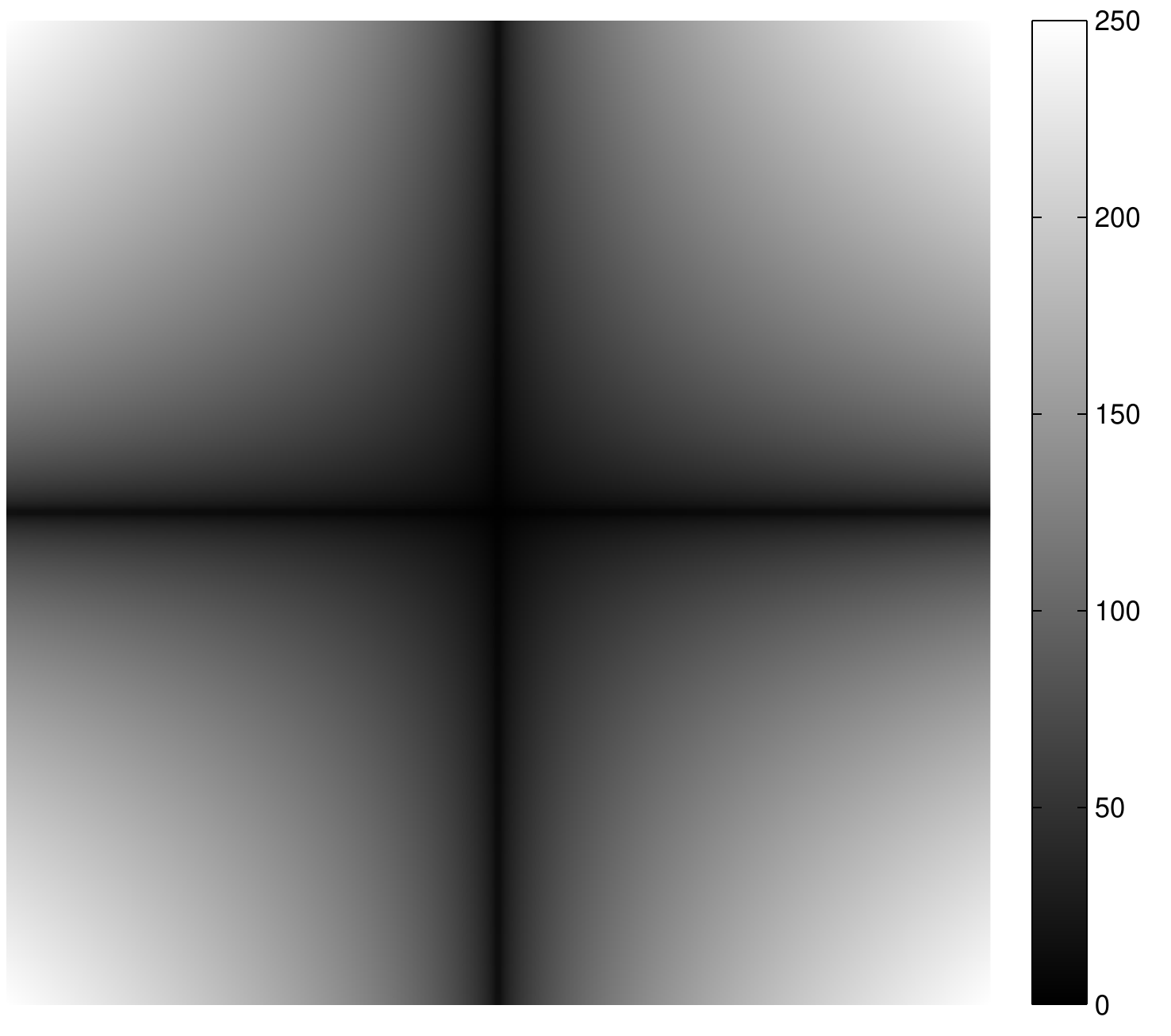} 
  \caption{\footnotesize Hyperbolic Scaling }
  \end{center}
\end{subfigure}
\begin{subfigure}[t]{0.32\textwidth}
\begin{center}
\includegraphics[width=\textwidth]{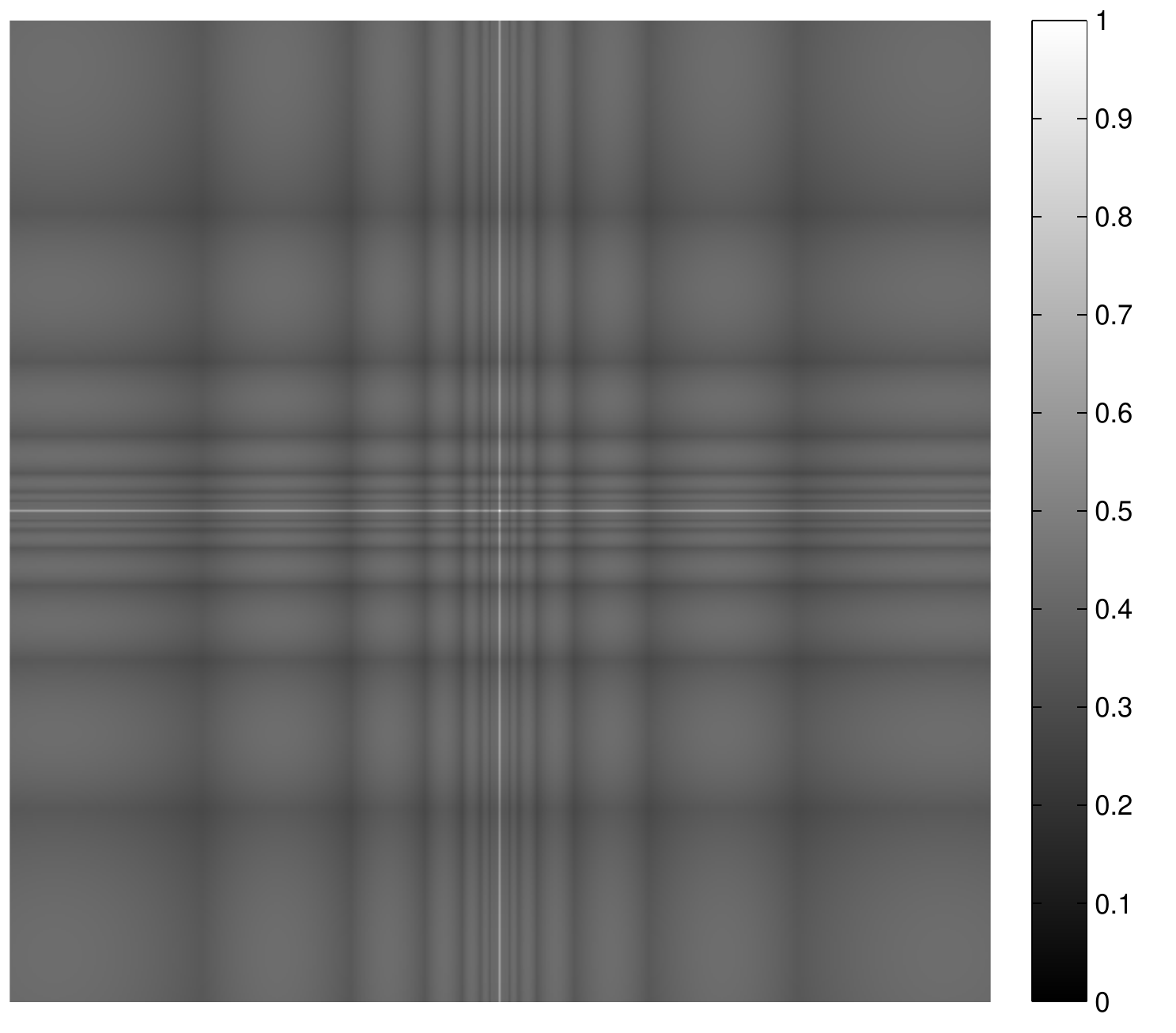} 
  \caption{\footnotesize Scaled Coherences \\ (Product of (a) \& (b)) }
  \end{center}
\end{subfigure}
\end{center}
\caption{2D Tensor Fourier - Tensor Haar Incoherences. We show the subset $\{-250,-249,...,249,250\}^2 \subset \bbZ^2$. Notice that the scaled coherences have no vanishing values (no pure black) and no values that blow up (no pure white, baring the center value which = 1) indicating that we have characterised the coherence in terms of the hyperbolic scaling used. Formally this is shown by equation (\ref{zhypcrosscharacterise}). The coherences shown in the Figure are square rooted to reduce contrast (i.e. we image $ \sqrt{\mu(\pi_N U)}$ instead of $\mu(\pi_N U)$).}
\label{tensorincoherences}
\end{figure}

\section{Multidimensional Fourier - Separable Wavelet Case: Proof of Theorem \ref{separablesummary}} \label{separable}

We repeat the notation of the one-dimensional case, with scaling function $\phi$ (in one dimension) \& Daubechies wavelet $\psi$:
\[ \phi_{j,k}(x) =2^{j/2}\phi(2^j x - k)  , \quad \psi_{j,k} (x)  =  2^{j/2} \psi(2^j x - k). \]
We can construct a d-dimensional scaling function $\Phi$ by taking the tensor product of $\phi$ with itself, namely
\[ \Phi(x) \ := \ \Big( \bigotimes_{j=1}^d \phi \Big)(x) = \prod_{j=1}^d \phi(x_j), \qquad x  \in \mathbb{R}^d, \]
which has corresponding multiresolution analysis $(\tilde{V}_j)_{j \in \mathbb{Z}}$ with diagonal scaling matrix $A \in \mathbb{R}^{d \times d}$ with $A_{i,j}=2 \delta_{i,j}$.

Let $\phi^0:=\phi, \  \phi^1:=\psi$ and for $s \in \{ 0,1 \}^d,  \ j \ge J, \ k \in \mathbb{Z}^{d}$ where $J \in \bbN$ is fixed define the functions 
\be{ \label{separabledefine}
\Psi^s_{j,d}:= \bigotimes_{i=1}^d \phi^{s_i}_{j,k_i}.
}
If we write (for $s \in \{ 0,1 \}^d \setminus \{0\}, \ j \ge J$) 
\[W^s_j := \overline{ \text{Span} \{ \Psi^s_{j,k}  :  k \in \mathbb{Z}^d \} }.\] Then it follows that 
\[ \tilde{V}_{j+1} = \tilde{V}_j \oplus \bigoplus_{s \in \{ 0,1 \}^d \setminus \{0\}}  W^s_j, \quad L^2(\mathbb{R}^d) =  \overline{\tilde{V}_J \oplus \bigoplus_{\substack{s \in \{ 0,1 \}^d \setminus \{0\}  \\ j \ge J}} W^s_j}.\]
This corresponds to taking $2^d-1$ wavelets for our basis in d dimensions (see \cite{dDwav}). As before we take the spanning functions from the above whose support has non-zero intersection with $[-1,1]^d$ as our basis $B_2$ (called a `separable wavelet basis'):
\begin{align} \label{separablewaveletbasisdefine}
B^d_{\text{sep}} := \left\{ \begin{array}{cc}   &  \mathrm{Supp}(\phi^{s_i}_{j,k_i}) \cap (-1,1) \neq \emptyset \quad \forall i , \\  \Psi^s_{j,k} :  &  s=0 \Rightarrow j=J, \\ &  j \in \mathbb{N}, s \in \{ 0,1 \}^d , \ k \in \mathbb{Z}^d  \end{array} \right \} , 
\end{align}

\begin{remark}
We can also construct a separable boundary wavelet basis in the same manner like in the one-dimensional case however, for the sake of simplicity, we stick to the above relatively simple construction throughout (although all the coherence results we cover here also hold for the separable boundary wavelet case as well).

\end{remark}

\subsection{Ordering the Separable Wavelet Basis: Proving Theorem \ref{separablesummary} Part (i)}
 \label{separablewaveletordering}

We note a few key equalities from the one-dimensional case that will come in handy:
\be{ \label{1Dequalities}
\mathcal{F}\phi_{j,k}( \omega )  = e^{-2\pi \ri  2^{-j} k \omega} 2^{-j/2} \mathcal{F} \phi(2^{-j} \omega), \qquad
\mathcal{F}\psi_{j,k}( \omega )  = e^{-2\pi \ri 2^{-j} k \omega} 2^{-j/2} \mathcal{F} \psi(2^{-j} \omega),
}
where $\mathcal{F}$ here denotes the Fourier Transform, i.e. for $f \in L^2(\bbR^d)$ we define
$$
\mathcal{F}f(\omega) = \int_{\mathbb{R}^d} f(x) e^{-2\pi i \omega \cdot x }  \, dx, \qquad \omega \in \bbR^d.
$$
Recall $\chi_k$ from Definition \ref{fourier}. We observe that by (\ref{separabledefine})
\begin{equation} \label{fourierprod}
\begin{aligned}
 \langle \Psi^s_{j,k} , \chi_n \rangle  & = \epsilon^{d/2} \cdot \mathcal{F} \Psi^s_{j,k}(\epsilon n)
 = \epsilon^{d/2} \prod_{i=1}^d \mathcal{F} \phi^{s_i}_{j,k_i} (\epsilon n_i), \qquad n \in \bbZ^d,
 \\  \Rightarrow & \sup_{n \in \bbN^d} |\langle \Psi^s_{j,k} , \chi_n \rangle|^2 = \epsilon^{d} 2^{-dj} \cdot \prod_{i=1}^d \sup_{n \in \bbN} |\mathcal{F} \phi^{s_i} (\epsilon 2^{-j} n)|^2.
 \end{aligned}
\end{equation}
By careful treatment of the product term we can determine the optimal decay of $(B^d_\text{sep}, B_\rf^d (\epsilon))$, using the following result:

\begin{proposition} \label{dDSeparableWaveletFourierCharacterisation}
There are constants $C_1,C_2>0$ such that for all $\epsilon \in I_{J,p}, \Psi^s_{j,k} \in B^d_\text{sep}$ we have 
\[
C_1 \cdot \epsilon^{d} 2^{-dj} \le  \sup_{n \in \bbN^d} |\langle \Psi^s_{j,k} , \chi_n \rangle|^2 \le C_2 \cdot \epsilon^{d} 2^{-dj}.
\]
Consequently, fixing $\epsilon$, the function $F_\text{power}: B^d_\text{sep} \to \bbR$ defined by $F_\text{power}(\Psi^s_{j,k})=2^{-dj}$ characterizes the optimal decay of $(B_\text{sep}^d, B_\rf^d(\epsilon))$.
\end{proposition}
\begin{proof}
Let $A=\max(\sup_{\omega \in \bbR} |\mathcal{F} \phi (\omega)|^2, \sup_{\omega \in \bbR} |\mathcal{F} \psi (\omega)|^2)$. Then (\ref{fourierprod}) gives us the upper bound
\[
\sup_{n \in \bbN^d} |\langle \Psi^s_{j,k} , \chi_n \rangle|^2 \le \epsilon^{d} 2^{-dj} \cdot A^d.
\]
This leaves the lower bound. This can be achieved if we can show that there exists constants $D_1, D_2>0$ such that for all $\epsilon \in I_{J,p}$\footnote{Notice that replacing $J$ with $j \ge J$ below would have been redundant.}
\be{ \label{fourierinfimum}
\mathcal{F}_1(\epsilon):=\sup_{n \in \bbN} |\mathcal{F} \phi (\epsilon 2^{-J} n)| \ge D_1, \quad \mathcal{F}_2(\epsilon):= \sup_{n \in \bbN} |\mathcal{F} \psi (\epsilon 2^{-J} n)| \ge D_2.
}
By the Riemann-Lebesgue Lemma the functions $\mathcal{F}_1, \mathcal{F}_2$ are continuous on $I_{J,p}$ and 
\[
\mathcal{F}_1(\epsilon) \to \sup_{\omega \in \bbR} |\mathcal{F} \phi (\omega)|>0 \quad \text{as} \quad \epsilon \to 0.
\]
Likewise for $\mathcal{F}_2$. Therefore $\mathcal{F}_1, \mathcal{F}_2$ can be extended to continuous functions over the closed interval $I_{J,p} \cup \{0\}$. Finally we notice that $\mathcal{F}_1(\epsilon)>0, \mathcal{F}_2(\epsilon)>0$ for every $\epsilon \in I_{J,p}$ otherwise we would deduce that $\phi$ or $\psi$ has no support in $[-1,1]$ since the span of $B_\rf^d(\epsilon)$ covers $L^2[-1,1]$. This means that the infimums over $I_{J,p} \cup \{0\}$ are attained and are strictly positive, proving (\ref{fourierinfimum}) and the lower bound.
\end{proof}
Let $F_\text{level}:B^d_\text{sep} \to \bbR$ be defined by $F_\text{level}(\Psi^s_{j,k})=j$. Lemma \ref{characterisationlemma} tells that an ordering that is consistent with $1/F_\text{power}$, i.e. consistent with $F_\text{level}$ will be strongly optimal.

\begin{definition} \label{sepleveled}
We say that an ordering $\rho: \mathbb{N} \to B^d_{\text{sep}}$ is `leveled' if it is consistent with $F_\text{level}$. 
\end{definition}

\begin{lemma} \label{levelgrowth}
Let $\rho: \mathbb{N} \to B^d_{\text{sep}}$ be leveled. Then there are constants $D_1, D_2>0$ such 
\be{ \label{levelgrowthequation}
D_1 \cdot N \le 2^{d F_\text{level}(\rho(N))} \le D_2 \cdot N.
}

\end{lemma}
\begin{proof}
Let $a \in \mathbb{N}$ denote the length of the support of $\phi, \psi$. Notice that for each $j \in \mathbb{N}$ and $s \in \{ 0,1 \}^d$, there are $(2^{j+1} +a-1)^d $ shifts of $\Psi^{s}_{j,0}$ whose support lies in $[-1,1]^d$.  For convenience we use the notation $j(N):= F_\text{level}(\rho(N))$ and shall also be using the simple bounds $ 2^{j(N)+1} \le 2^{j(N)+1}+a-1 \le  2^{j(N)+a}$. Now for every $N \in \bbN$ with $j(N)>J$, we must have had all the terms of the form $f \in B^d_{\text{sep}}, F_\text{level}(f)=j(N)-1$ come before $N$ in the leveled ordering and there are at least $(2^d-1) \cdot 2^{dj(N)}$ of these terms, implying that 
\[ (2^d-1) \cdot 2^{dj(N)} \le N.\]
This completes the upper bound for $j(N)>J$. Likewise for every $N \in \bbN$ with $j(N)\ge J$ there can be no more than 
 \[ 
2^d  \cdot \sum_{i=J}^{j(N)} 2^{d(i+a)} \le 2^d \cdot 2^{d(j(N)+a+1)}= 2^{d(a+2)} \cdot 2^{dj(N)},
 \] 
terms such that $F_\text{level}(f) \le j(N)$. This shows that $N \le 2^{d(a+2)} \cdot 2^{dj(N)}$, completing the upper bound for $j(N)>J$. Extending (\ref{levelgrowthequation}) to all $N \in \bbN$ (i.e. $j(N) \ge J$) is trivial since we have only omitted finitely many terms so a change of constants will suffice.
\end{proof}
\begin{corollary} \label{leveledresults}
Any ordering $\rho$ of $B^d_\text{sep}$ that is leveled is strongly optimal for the basis pair $(B^d_\text{sep},B^d_\rf(\epsilon))$. Furthermore, the optimal decay rate of $(B^d_\text{sep},B_\rf^d(\epsilon))$ is represented by the function $f(N)=N^{-1}$.
\end{corollary}
\begin{proof}
Lemma \ref{characterisationlemma} applied to Proposition \ref{dDSeparableWaveletFourierCharacterisation} tells us that $\rho$ is strongly optimal and moreover the optimal decay rate is represented by $F_\text{power}(\rho(N))$ which by Lemma \ref{levelgrowth} is of order $N^{-1}$.
\end{proof}

\subsection{Ordering the Fourier Basis: Proving Theorem \ref{separablesummary} Part (ii)} \label{linearproof}

We now want to find the optimal decay rate of $(B^d_\rf(\epsilon), B^d_{\text{sep}})$ which means looking at orderings of the Fourier basis. It might be tempting to try and extend the standard ordering definition from the one dimensional Fourier basis. Recall as well that, using the function $\lambda_d$ defined in (\ref{multidimlambda}), ordering $B^d_\rf(\epsilon)$ is equivalent to ordering $\bbZ^d$.

If we let $s \in \{ 0,1 \}^d , \ j \in \bbN, \ k \in \bbZ^d$, then in order to bound the coherence $\mu(\pi_NU)$ we need to be bounding terms of the form
\be{ \label{sepprodexample2}
| \langle  \Psi^{s}_{j,k} , \lambda_d^{-1} (n) \rangle |^2 
=   \epsilon^d 2^{-dj} \prod_{i=1}^d | \mathcal{F} \phi^{s_i}(2^{-j} \epsilon n_i) |^2.
}
In the one-dimensional case in \cite{onedimpaper} the following decay property of the Fourier transform of the scaling function $\phi$ was used:

\begin{lemma} \label{FTdecayLemma}
If $\phi$ is any Daubechies scaling function with corresponding mother wavelet $\psi$ then there exists a constant $K>0$ such that for all $\omega \in \bbR \setminus \{0 \}$,
\be{ \label{FTdecay}
| \mathcal{F} \phi(\omega)|, |\mathcal{F} \psi (\omega)| \le \frac{K}{|\omega|}.
}
Furthermore, suppose that for some $\alpha>0$ we have, for some constant $K>0$, the decay $| \mathcal{F} \phi(\omega)| \le K | \omega|^{-\alpha}$ for all $\omega \in \bbR \setminus \{0\}$. Then, for a larger constant $K>0$, $| \mathcal{F} \psi(\omega)| \le K | \omega|^{-\alpha}$ for all $\omega \in \bbR \setminus \{0\}$.
\end{lemma}
\begin{proof}
The first result is a direct result of Lemma 3.5 in \cite{onedimpaper}. The last statement follows immediately from the equality (taken from equation (3.14) in \cite{onedimpaper}):
\begin{equation} \label{fourierscalingwavelet}
|\mathcal{F}\psi(2 \omega)|= |m_0(\omega + 1/2) \cdot \mathcal{F}\phi(\omega)|,
\end{equation}
where $m_0$ is the low pass filter corresponding to $\phi$ which satisfies $|m_0(\omega)| \le 1$ for all $\omega \in \bbR$.
\end{proof}

Therefore let us first consider the case where we use (\ref{FTdecay}) to bound every term in the product, giving us ($n \in \bbZ^d, n_i \neq 0, i=1,...,d$)

\be{ \label{hyperbolicbound}
| \langle  \Psi^{s}_{j,k} , \lambda_d^{-1} (n) \rangle |^2   \le \epsilon^d 2^{-dj} \prod_{i=1}^d \frac{K^{2}}{|\epsilon 2^j n_i|}  =\frac{K^{2d}}{ \prod_{i=1}^d | n_i|}.
}
Making adjustments to prevent dividing by zero by using $\sup_{\omega \in \bbR} \max(| \mathcal{F} \phi(\omega)|, | \mathcal{F} \psi(\omega)|)\le 1$ (for $\phi$ this follows from Proposition 1.11 in \cite{wav}. We extend this to $\psi$ using equation (\ref{fourierscalingwavelet})), this can then be rephrased as
\be{ \label{hyperbolicbound2}
\sup_{g \in B_\text{sep}^d} | \langle  g , \lambda_d^{-1} (n) \rangle |^2  \le \frac{\max (K^{2d},1)}{ \prod_{i=1}^d \max(| n_i|,1)}, \qquad n \in \bbZ^d.
}
This tells us that the function $F_\text{hyp}: \bbZ^d \to \bbR, F_\text{hyp}(n)= (\prod_{i=1}^d \max(|n_i|,1))^{-1}$ dominates the optimal decay of $(B_\rf^d, B_\text{sep}^d)$ (see Definition \ref{Characterisation} for the definition of domination).  
 Therefore if we want to maximise the utility of this bound then we should use an ordering $\sigma$ of $\bbZ^d$ so that  $\prod_{i=1}^d \max(|\sigma(N)_i|,1)$ is increasing, namely an ordering corresponding to the hyperbolic cross in $\bbZ^d$ (see Example \ref{hypcrossZd}). However, using such an ordering will not give us the $N^{-1}$ decay rate that we got from the one dimensional case:

\begin{proposition} \label{Hyperbolic4Separable}
Let $\sigma: \bbN \to \bbZ^d$ correspond to the hyperbolic cross in $\bbZ^d$ and define an ordering $\rho$ of $B^d_\rf(\epsilon)$ by $\rho:=\lambda_d^{-1} \circ \sigma$, where $\epsilon \in I_{J,p}$. Next let $U=[(B^d_\rf(\epsilon),\rho),(B^d_{\text{sep}},\tau)]$ for any ordering $\tau$ and fix $\epsilon$. Then there are constants $C_1, C_2 > 0$
\[  \frac{C_1 \log^{d-1}(N+1)}{N} \le \mu( Q_N U )  \le    \frac{C_2 \log^{d-1}(N+1)}{N}, \qquad N \in \bbN. \]
\end{proposition}
As this result is primarily for motivation, its proof is left to the appendix.

Since this approach gives us suboptimal results, we return to our bound of (\ref{sepprodexample2}). Instead of using (\ref{FTdecay}) on every term in the product, why not just use it once on the term that give us the best decay instead? To bound the remaining terms we can simply use $\sup_{\omega \in \bbR} \max(| \mathcal{F} \phi(\omega)|, | \mathcal{F} \psi(\omega)|)\le 1$ . This approach gives us the following bound
\be{ \label{linearbound1}
| \langle  \Psi^{s}_{j,k} , \lambda_d^{-1} (n) \rangle |^2   \le \epsilon^d 2^{-dj}  \cdot \min_{i=1,...d} \frac{K^{2}}{|\epsilon 2^j n_i|}
 =\epsilon^{d-1} 2^{-(d-1)j}  \cdot \frac{K^{2}}{ \max_{i=1,...,d} | n_i|}, \qquad n \in \bbZ^d.
}
As we shall see in Lemma \ref{normest}, choosing $\rho$ so that we maximise the growth of the $\max_{i=1,...,d} | n_i|$ leads to $ \max_{i=1,...,d} | n_i| \ge E \cdot N^{1/d}$ for some constant $E>0$ and so  (\ref{linearbound1}) is bounded above by $\text{constant} \cdot N^{-1/d}$, which is very poor decay. However, if we instead replace (\ref{FTdecay}) by the stronger condition

\begin{equation} \label{dDFTdecay}
 |\mathcal{F} \phi( \omega )|  \le  \frac{K}{| \omega |^{d/2}}  , \qquad \omega \in \bbR \setminus \{0 \} .
 \end{equation}
then we can obtain the following upper bound\footnote{noting that (\ref{dDFTdecay}) also holds for $\psi$ by (\ref{fourierscalingwavelet}).}
 \be{ \label{linearbound2}
| \langle  \Psi^{s}_{j,k} , \rho(N) \rangle |^2   \le \epsilon^d 2^{-dj}  \cdot \min_{i=1,...d} \frac{K^{2d}}{|\epsilon 2^j n_i|^d}
 =  \frac{K^{2d}}{ \max_{i=1,...,d} | n_i|^d}.
}
Let us write $\|n\|_\infty:=\max_{i=1,...,d} | n_i |$. The above can be rephrased as
\be{ \label{linearbound3}
\sup_{g \in B_\text{sep}^d} | \langle  g , \lambda_d^{-1}(n) \rangle |^2  \le \frac{\max(K^{2d},1)}{ \max(\|n\|_\infty^d,1)}, \qquad n \in \bbZ^d .
}
 Therefore we deduce that $F_\text{lin}: \bbZ^d \to \bbR, F_\text{lin}(n)=(\max(\|n\|_\infty^d,1))^{-1}$ dominates the optimal decay of $(B_\rf^d (\epsilon), B_\text{sep}^d)$. In fact in can be shown that $F_\text{lin}$ also \textit{characterizes} the optimal decay (i.e. a lower bound of the same form is possible) by using the following preliminary Lemma:

\begin{lemma} \label{wavelower}
For any compactly supported wavelet $\psi$ there exists an $R \in \mathbb{N}$ such that for all $q \ge R, \ (q \in \mathbb{N})$ we have 
\be{ \label{wavelowerequation}
L_q := \inf_{\omega \in [2^{-(q+1)},2^{-q}] } | \mathcal{F}\psi(\omega)| \ > \ 0. 
}
\end{lemma}
\begin{proof}
See Lemma 3.6 in \cite{onedimpaper}.
\end{proof}

\begin{proposition} \label{dDSeparableFourierWavelet}
We fix the choice of wavelet basis $B^d_\text{sep}$ and recall the function $\lambda_d : B^d_\rf (\epsilon) \to \bbZ^d$ from (\ref{multidimlambda}).

1.) \ \  Then there are constants $C_1(\phi)>0, D(J)>0$ such that for all $\epsilon \in I_{J,p}$ and $n \in \bbZ^d$ with $\| n \|_\infty \ge D \epsilon^{-1}$ we have
\be{ \label{linearbound4}
\sup_{g \in B_\text{sep}^d} | \langle  g , \lambda_d^{-1}(n) \rangle |^2  \ge \frac{C_1}{ \|n\|_\infty^d}.
}

Therefore  (by fixing $\epsilon$) the function $F_\text{lin}$ is dominated by the optimal decay of $(B^d_\rf(\epsilon), B^d_\text{sep})$.

2.) \ \ Suppose that $\phi$ satisfies (\ref{dDFTdecay}). Then there is a constant $C_2(\phi)>0$ such that for all $\epsilon \in I_{J,p}$ and $n \in \bbZ^d$,
\be{ \label{linearbound5}
\sup_{g \in B_\text{sep}^d} | \langle  g , \lambda_d^{-1}(n) \rangle |^2  \le \frac{C_2}{ \max(\|n\|_\infty^d,1)}.
}
Therefore  (by fixing $\epsilon$) the function $F_\text{lin}$ characterizes the optimal decay of $(B^d_\rf(\epsilon), B^d_\text{sep})$.
\end{proposition}
\begin{proof}

2.) \ \ Follows from (\ref{linearbound3}).

1.) \ \ If we set $j= \lceil \log_2 \epsilon \| n \|_\infty \rceil +q$ for some $q \in \bbN$ fixed we observe that $|  \epsilon  2^{-j} n_i | \in [0,2^{-q}]$ for every $i=1,...,d$ and, since we are using the max norm, $|  \epsilon  2^{-j} n_i | \in [2^{-q-1},2^{-q}]$ for at least one $i$, say $i'$ . Set $s_i=0$ for $i \neq i'$ and $s_{i'}=1$. Then, assuming $j \ge J$, by (\ref{sepprodexample2}) we have the lower bound.
\begin{equation} \label{lastcall}
\begin{aligned}
  &| \langle  \Psi^s_{j,0} , \lambda_d^{-1}(n) \rangle |^2  \ge \frac{2^{-d(q+1)}}{ \|n\|_\infty^d} \prod_{i=1}^d | \mathcal{F} \phi^{s_i} (\epsilon 2^{-j} n_i)|^2
 \\ & \qquad \ge \frac{ 2^{-d(q+1)}}{\|n\|_\infty^d} \cdot \inf_{ \omega \in (2^{-q-1},2^{-q}]}| \mathcal{F} \psi(\omega)|^2 
  \cdot
 \inf_{ \omega \in [0,2^{-q}]}| \mathcal{F} \phi (\omega)|^{2(d-1)}.
 \end{aligned}
 \end{equation}
Recall that by Lemma \ref{wavelower} there exists a $q \in \mathbb{N}$ such that $L_q>0$ and $\inf_{ \omega \in [0,2^{-q}]}| \mathcal{F} \phi (\omega)| >0 $\footnote{We are using the fact that $|\mathcal{F} \phi(0)|=1$ and continuity of $\mathcal{F} \phi$ here which follows from $\phi \in L^1(\bbR)$.} and therefore (\ref{linearbound5}) follows as long as $j \ge J$. 
To ensure that $j= \lceil \log_2 ( \epsilon \|n\|_\infty) \rceil + q$ satisfies $j \ge J$ we must therefore impose the constraint that $n$ is sufficiently large. $j \ge J$ is satisfied if
\[ J \le \log_2 ( \epsilon \|n\|_\infty) \quad \Rightarrow \quad \|n\|_\infty \ge  2^{J} \epsilon^{-1}. \]
\end{proof}

\begin{remark} \label{2doptimal}
If $d=2$ then (\ref{dDFTdecay}) always holds by Lemma \ref{FTdecayLemma}. This means we have characterized every 2D Separable wavelet case (for Daubechies Wavelets).
\end{remark}

\begin{remark}
A similar upper bound in two dimensions based on the norm of $n \in \mathbb{Z}^2$ has already been considered in a discrete framework for separable Haar wavelets \cite{discrete}.
\end{remark}
 
Let $F_\text{norm}(n):=\max(\|n\|_\infty,1)$ . By \ref{characterisationlemma} we know that if (\ref{dDFTdecay}) holds then the optimal decay of $(B_\rf^d,B^d_\text{sep})$ is determined by the fastest growth of $F_\text{norm}$. This motivates the following:
 
 \begin{lemma} \label{normest}
Let $\sigma : \mathbb{N} \rightarrow \mathbb{Z}^d$ be consistent with $F_\text{norm}$. Then there are constants $E_1 , E_2 >0$ such that
\begin{equation} \label{normestresult}
 E_1 \cdot N^{1/d}  \le  \max(\| \sigma(N) \|_\infty,1)  \le  E_2 \cdot N^{1/d}, \qquad \forall N \in \mathbb{N}.
 \end{equation}

\end{lemma}
\begin{proof}
If $ \| \sigma(N) \|_\infty = L \ge 2$, then $\sigma$ must have enumerated beforehand all points $m$ in $ \mathbb{Z}^d$ with $\|m\|_\infty \le L-1$ and there are $(2L -1)^d$ of such points. This means that
\[
N \ge (2 L -1)^d \quad \Rightarrow \quad \| \sigma(N) \|_\infty \le \frac{N^{1/d}+1}{2}, \qquad N \in \bbN.
\]
which proves the upper bound when $ \| \sigma(N) \|_\infty = L \ge 2$. The lower bound is tackled similarly by noting $\sigma$ must first list all $m \in \bbZ^d$ with $\| m \|_\infty \le L$, including $\sigma(N)$ which shows
\[
N \le (2 L +1)^d \quad \Rightarrow \quad \| \sigma(N) \|_\infty \ge \frac{N^{1/d}-1}{2}, \qquad N \in \bbN.
\]
 This proves (\ref{normest}) for $ \| \sigma(N) \|_\infty = L \ge 2$. Extending this to all $ N \in \bbN$ is trivial since we have only omitted finitely many terms, so changing the constants will suffice since all terms are strictly positive.
\end{proof}

\begin{definition}[Linear Ordering]
Any ordering $\rho: \bbN \to B_\rf^d(\epsilon)$ such that $\sigma=\lambda_d \circ \rho$ satisfies (\ref{normestresult}) is called a `linear ordering'.
\end{definition}
 
\begin{corollary} \label{linearresults}
Assuming (\ref{dDFTdecay}) holds for the scaling function corresponding to $B^d_\text{sep}$, an ordering $\rho$ of $B_\rf^d(\epsilon)$ is strongly optimal for the basis pair $(B_\rf^d(\epsilon),B^d_\text{sep})$ if and only if it is linear. Furthermore, the optimal decay rate of $(B_\rf^d(\epsilon),B^d_\text{sep})$ is represented by the function $f(N)=N^{-1}$.
\end{corollary} 
\begin{proof}
If we apply part 2.) of Proposition \ref{dDSeparableFourierWavelet} to Lemma \ref{characterisationlemma} we kow that if $\sigma: \bbN \to \bbZ^d$ is consistent with $1/F_\text{lin}=F_\text{norm}^d$, i.e. consistent with $F_\text{norm}$, then $F_\text{lin}(\sigma(\cdot))=1/F_\text{norm}^d(\sigma(\cdot))$ represents the optimal decay rate. Lemma \ref{normest} tells us that this optimal decay is $1/(N^{1/d})^d=1/N$. Furthermore, Lemma \ref{characterisationlemma} says that an ordering $\rho$ is strongly optimal for $(B_\rf^d(\epsilon),B^d_\text{sep})$ if and only if $F_\text{lin}(\lambda_d \circ \rho(\cdot)) \approx F_\text{lin}(\sigma(\cdot))$ which holds if and only if $F_\text{norm}(\lambda_d \circ \rho(\cdot)) \approx F_\text{norm}(\sigma(\cdot))$, namely $\rho$ is linear.
\end{proof}
 
Corollary \ref{linearresults} gives us the same optimal decay as in one dimension, which is in contrast to the multidimensional tensor case, where the best we can do is have $d-1$ extra log factors. We can use this result to cover the two dimensional case in full:
\begin{corollary} \label{twodimresults}
In 2D the optimal decay rate of $(B^2_\rf,B^2_\text{sep})$ is represented by $f(N)= N^{-1}$. This optimal decay rate is obtained by using a linear ordering. In fact an ordering $\rho$ of $B^2_\rf$ is strongly optimal in 2D if and only if it is linear. 
\end{corollary}
\begin{proof}
Using Lemma \ref{FTdecayLemma} we observe that the decay condition (\ref{dDFTdecay}) holds automatically if $d=2$. Therefore we may apply Corollary \ref{linearresults} directly.
\end{proof}
This result \emph{does not extend to higher dimensions:}
 
\begin{example} \label{3DHaar}
If we do not have condition (\ref{dDFTdecay}) then our argument can break down very badly:
For Haar wavelets we have an explicit formula for the Fourier transform of the one-dimensional mother wavelet,
\[ \mathcal{F} \phi( \omega) = \frac{ \exp(2 \pi i \omega)-1}{2 \pi i \omega }. \]
Therefore we have that (\ref{dDFTdecay}) is not satisfied for $d \ge 3$ and furthermore we have (for $\epsilon<1$ and $J \in \bbN$ fixed)
\begin{equation} \label{FTdecayfail}
 | \mathcal{F} \phi( \epsilon 2^{-J} k ) | \ge \frac{1}{2\pi \epsilon k} ,
 \end{equation}
for infinitely many $k \in \mathbb{N}$. 
Now consider the case of $d$D separable Haar wavelets with a linear ordering $\rho$ of the Fourier Basis. Then, for $m \in \mathbb{N}$ such that $\lambda_d \circ \rho(m)=(\lambda_d \circ \rho(m)_1,0,\cdots,0)$ we know that by (\ref{FTdecayfail}) there are infinitely many $m$ such that
\be{ \label{poorlowerbound}
\begin{aligned} 
| \langle \Phi , \rho(m) \rangle |^2 & = \epsilon^d |\mathcal{F} \phi(\epsilon 2^{-J} \lambda_d \circ \rho(m)_1)|^2 \cdot |\mathcal{F} \phi(0)|^2(d-1)
\\ & \ge \epsilon^d \cdot \frac{1}{(2\pi \epsilon |\lambda_d \circ \rho(m)_1|)^2} \ge \frac{\epsilon^{d-2} E}{4 \pi^2 m^{2/d}},
\end{aligned}
}
for some constant $E$ using Lemma \ref{normest}. Therefore an upper bound of the form $\text{Constant} \cdot N^{-1}$ is not possible for a linear scaling scheme if $d \ge 3$. 
This can be rectified by applying a semi-hyperbolic scaling scheme, as in the next subsection.
\end{example}

\subsection{Examples of Linear Orderings - Linear Scaling Schemes} \label{linearsection}

A wide variety of sampling schemes that are commonly used happen to be linear. In particular we demonstrate that sampling according to how a shape scales linearly from the origin always corresponds to a linear ordering (see Figure \ref{linearscalingimages}):

\begin{figure}[!t]
\begin{center}
\begin{subfigure}[t]{0.4\textwidth}
\begin{center}
\includegraphics[width=\textwidth]{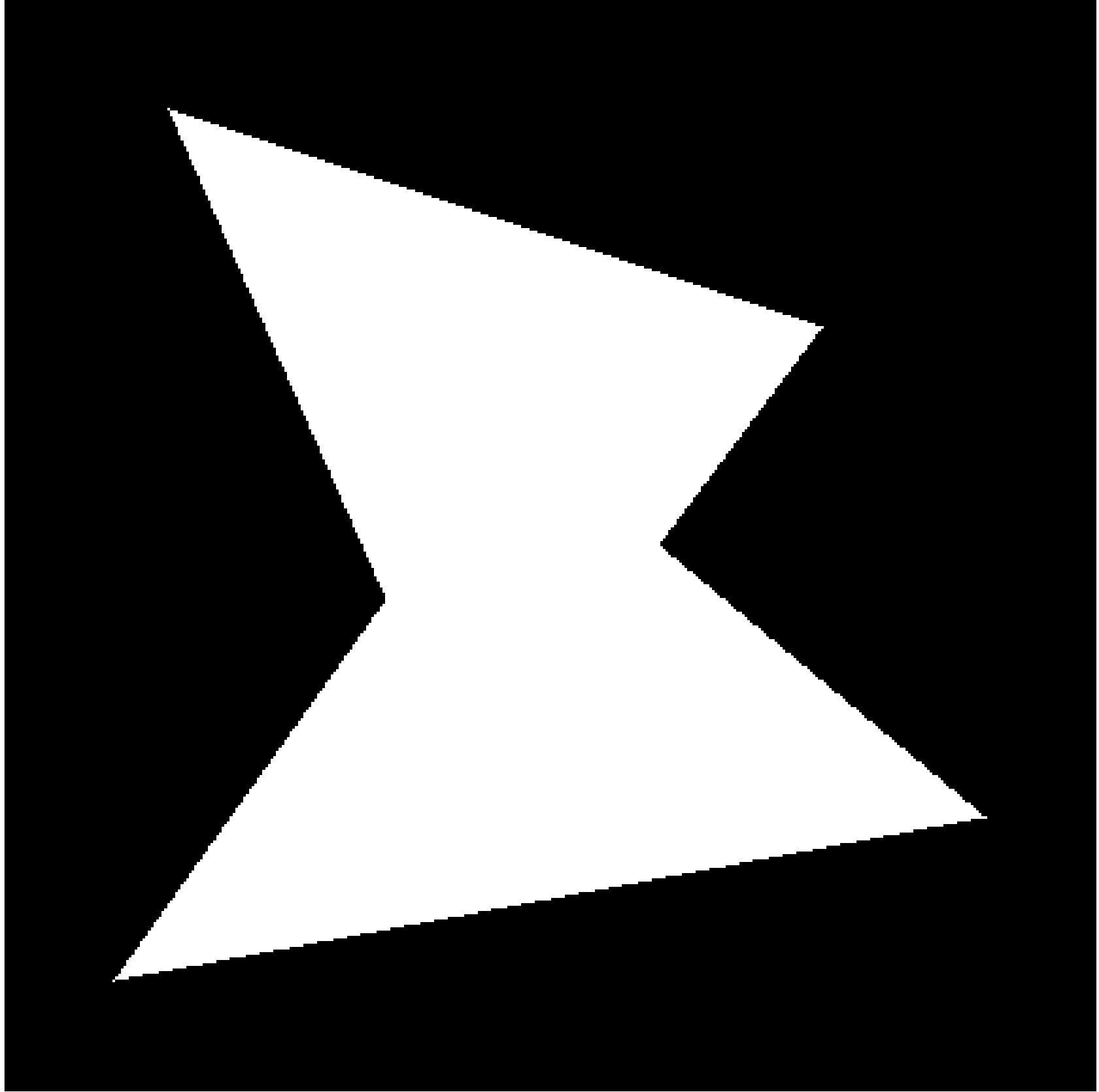}
 \caption{\footnotesize Scaling shape}
\end{center}
\end{subfigure}
\begin{subfigure}[t]{0.4\textwidth}
\begin{center}
 \includegraphics[width=\textwidth]{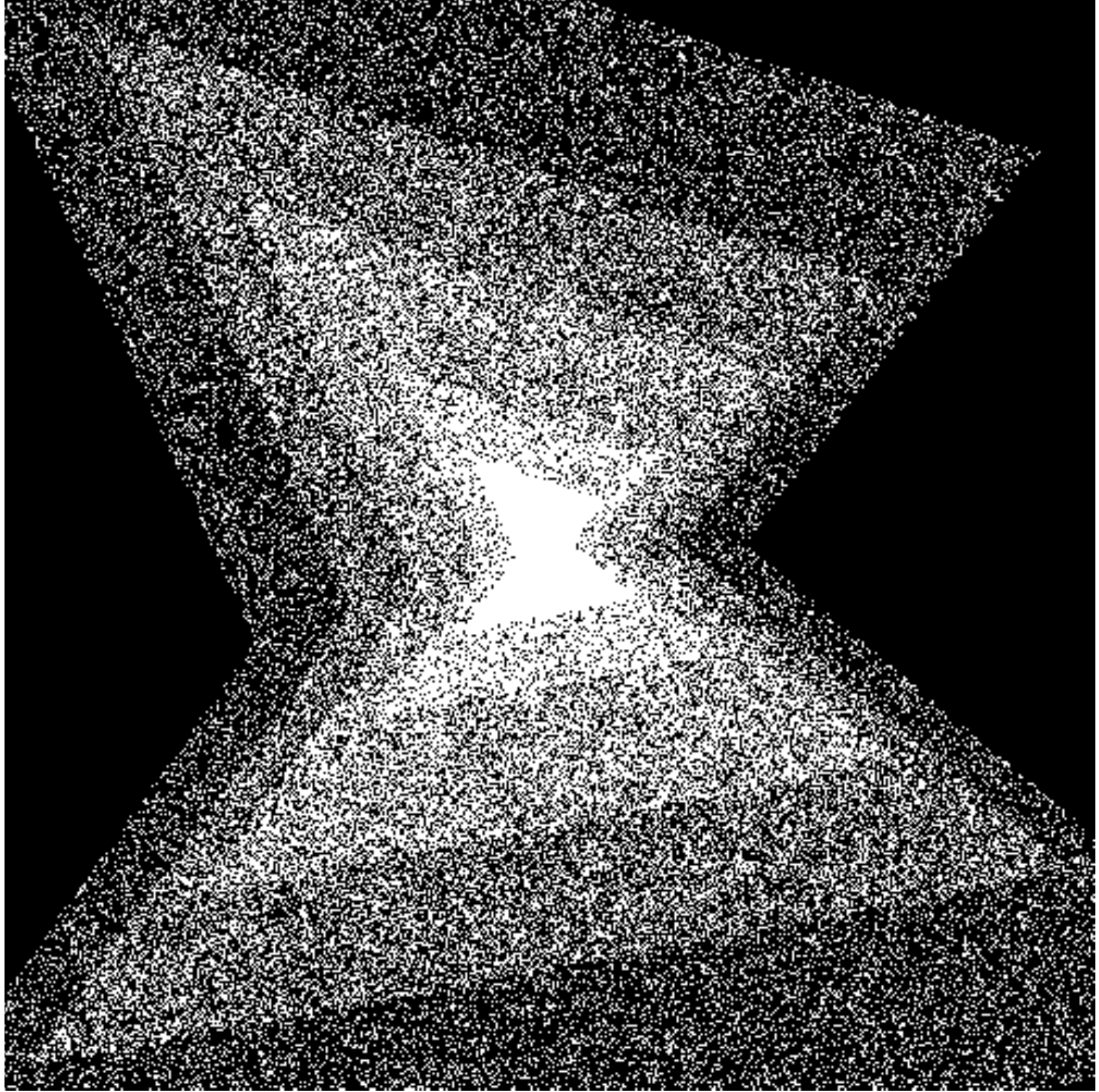} 
  \caption{\footnotesize Sampling according to a linear scaling scheme with scaling shape (a)}
  \end{center}
\end{subfigure}
\end{center}
\caption{A Simple Linear Scaling Scheme} 
\label{linearscalingimages}
\end{figure}

\begin{definition} \label{lineardef}
Let $D \subset \mathbb{R}^d$ be bounded with $0$ in its interior and define $S_{D}: \mathbb{Z}^d: \to \mathbb{R}$
\[ S_{D}(x) := \inf \big\{ \kappa >0 :  x \in \kappa D \ \big\}. \]
An ordering $\sigma: \bbN \to \bbZ^d$ is said to `correspond to a linear scaling scheme with scaling shape $D$' if it is consistent with $S_D$. Furthermore, an ordering $\rho : \mathbb{N} \rightarrow B^d_\rf(\epsilon)$ is said to `correspond to a linear scaling scheme with scaling shape $D$' if it is consistent with $S_D \circ \lambda_d$.

\end{definition}

\begin{remark} \label{linearnorm}
If we put a norm $\| \cdot \|$ on $\mathbb{Z}^d$ and take an ordering consistent with this norm then this ordering corresponds to a linear scaling scheme with scaling shape $\{ x \in \mathbb{R}^d \ : \ \|x\|=1 \}.$
\end{remark}
\begin{lemma} \label{normestold}
Let $\rho: \bbN \to B^d_\rf(\epsilon)$ corresponds to a linear scaling scheme with scaling shape $D$. Then $\rho$ is linear.

\end{lemma}
\begin{proof}
Let $\sigma=\lambda_d \circ \rho$.  Because the scaling shape $D$ is bounded and contains $0$ in its interior we have that there exists constants $C_1,C_2>0$ such that $C_1 \mathcal{S} \subset D \subset C_2 \mathcal{S}$ where $\mathcal{S}$ is defined to be the unit hypercube, i.e.
$ \mathcal{S} :=  \{ x \in \mathbb{R}^d  :  \|x\|_\infty=1 \}.$
Therefore  if $ \| \sigma(N) \|_\infty = L$, then since $D \subset C_2 \mathcal{S}$ we have that $S_D(\sigma(N))\ge L C_2^{-1}$. Applying this to $C_1 \mathcal{S} \subset D$  we deduce that $\sigma$ must have enumerated beforehand all points $m$ in $ \mathbb{Z}^d$ with $\|m\|_\infty<LC_1 C_2^{-1}$ and there are at least $(2 ( L C_1 C_2^{-1} -1 )+1)^d$ of such points. This means that
\[
N \ge (2 ( \| \sigma(N) \|_\infty C_1 C_2^{-1} -1 ) +1)^d \quad \Rightarrow \quad \| \sigma(N) \|_\infty \le \frac{N^{1/d}+1}{2 C_1 C_2^{_1}} \le \frac{N^{1/d}}{C_1 C_2^{-1}}, \qquad N \in \bbN.
\]
which proves the upper bound. The lower bound is tackled similarly to prove (\ref{normestresult}).
\end{proof}

\subsection{2D Separable Incoherences}

By Remark \ref{2doptimal} we have shown that linear orderings are strongly optimal for all 2D Fourier - separable wavelet cases, so this is a good point to have a quick look at a few of these in Figure \ref{separableincoherences}. 

\begin{figure}[!t]
\begin{center}
\begin{subfigure}[t]{0.32\textwidth}
\begin{center}
\includegraphics[width=\textwidth]{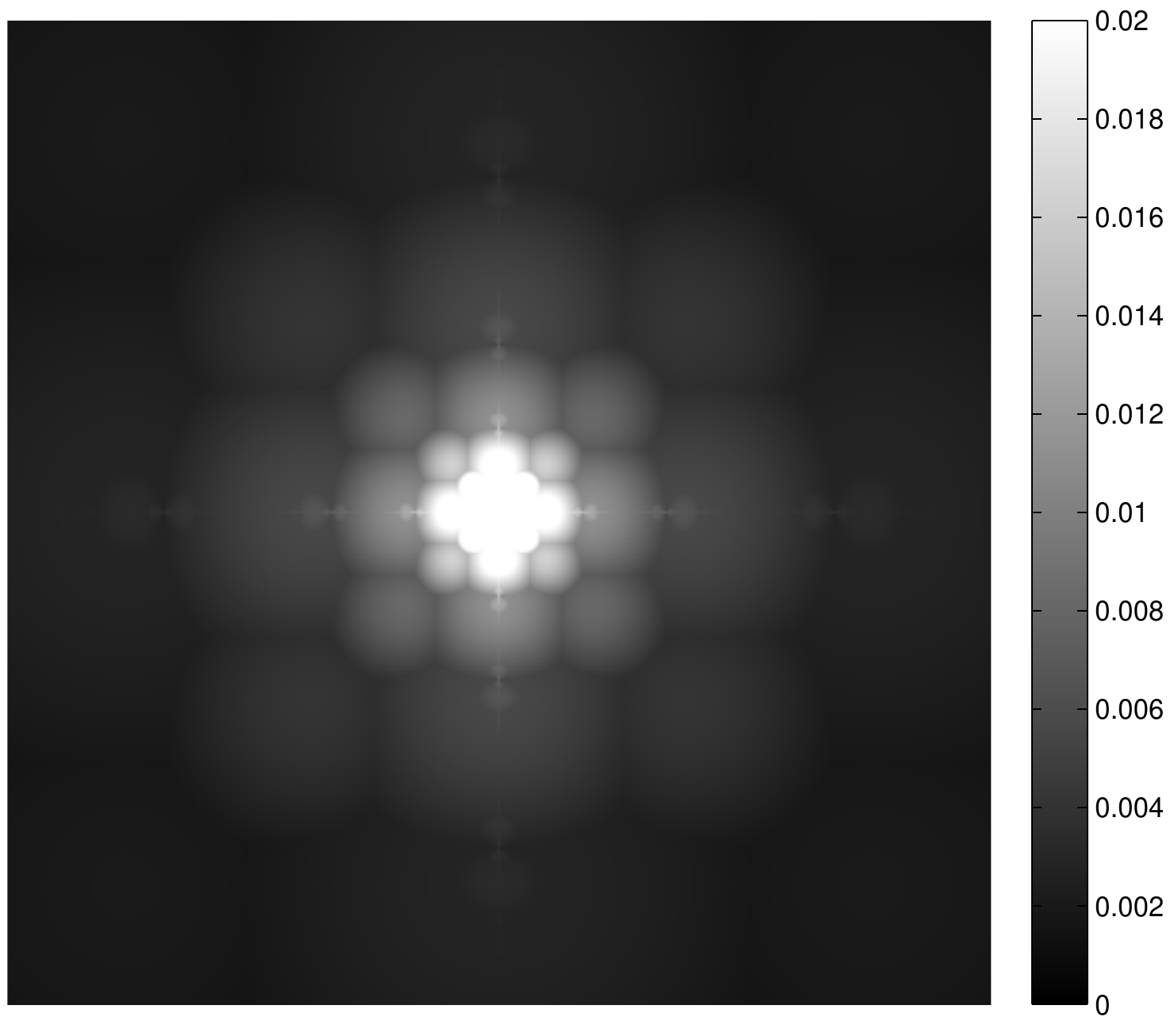}
 \caption{\footnotesize Haar Coherences}
\end{center}
\end{subfigure}
\begin{subfigure}[t]{0.32\textwidth}
\begin{center}
 \includegraphics[width=\textwidth]{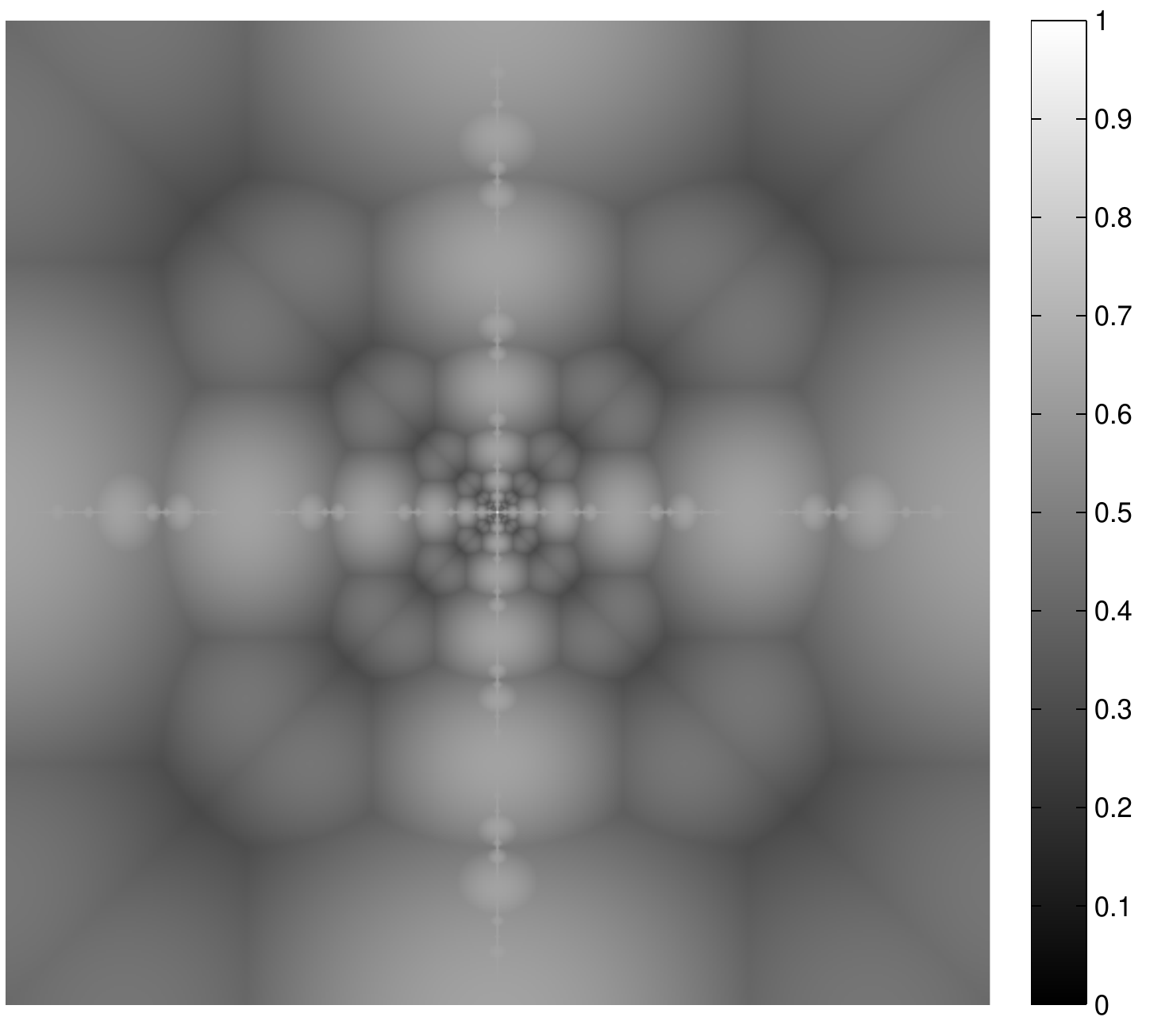} 
  \caption{\footnotesize Scaled Haar Coherences}
  \end{center}
\end{subfigure}
\begin{subfigure}[t]{0.32\textwidth}
\begin{center}
\includegraphics[width=\textwidth]{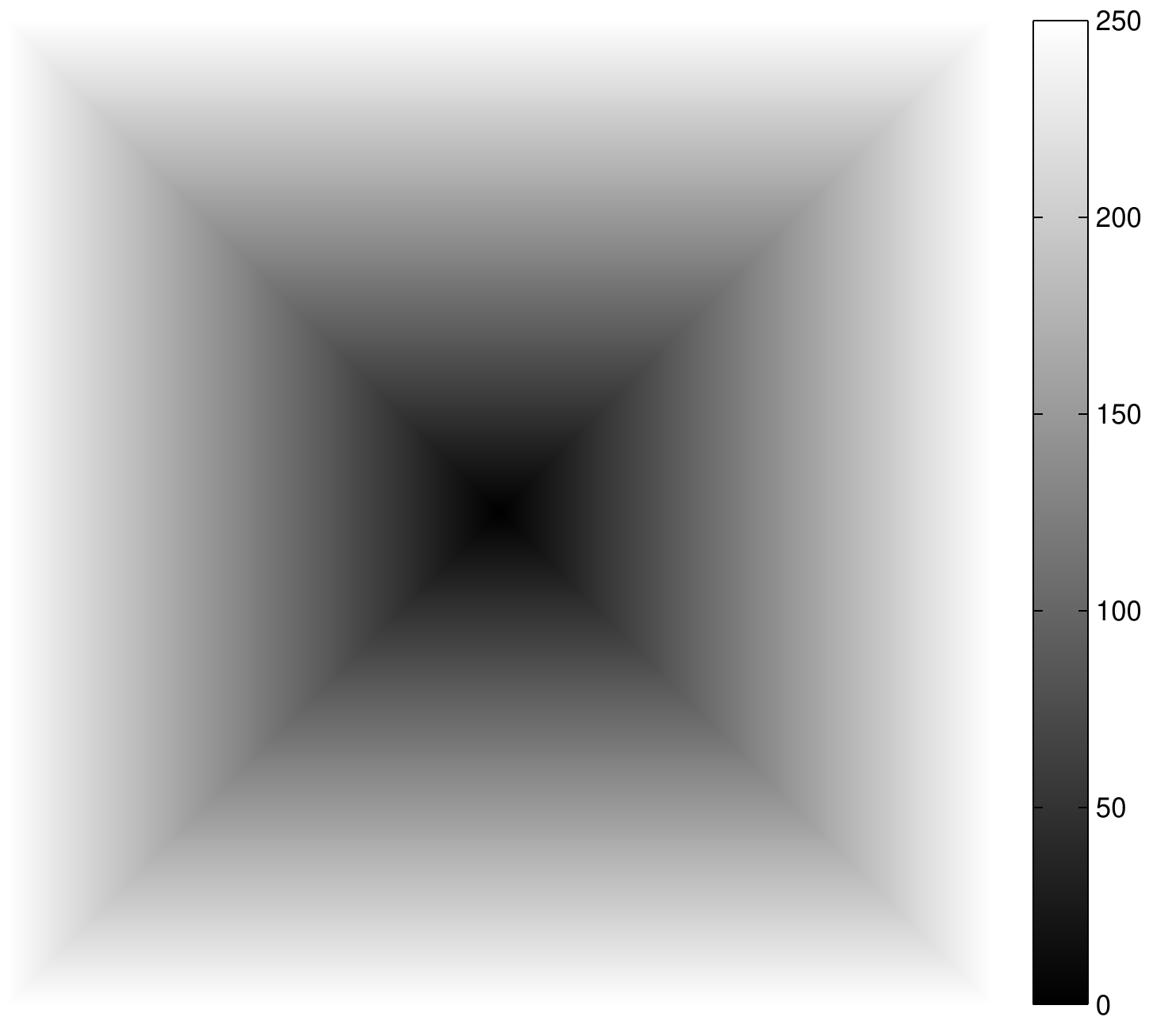} 
  \caption{\footnotesize Linear Scaling used for Both Bases}
  \end{center}
\end{subfigure}
\begin{subfigure}[t]{0.32\textwidth}
\begin{center}
\includegraphics[width=\textwidth]{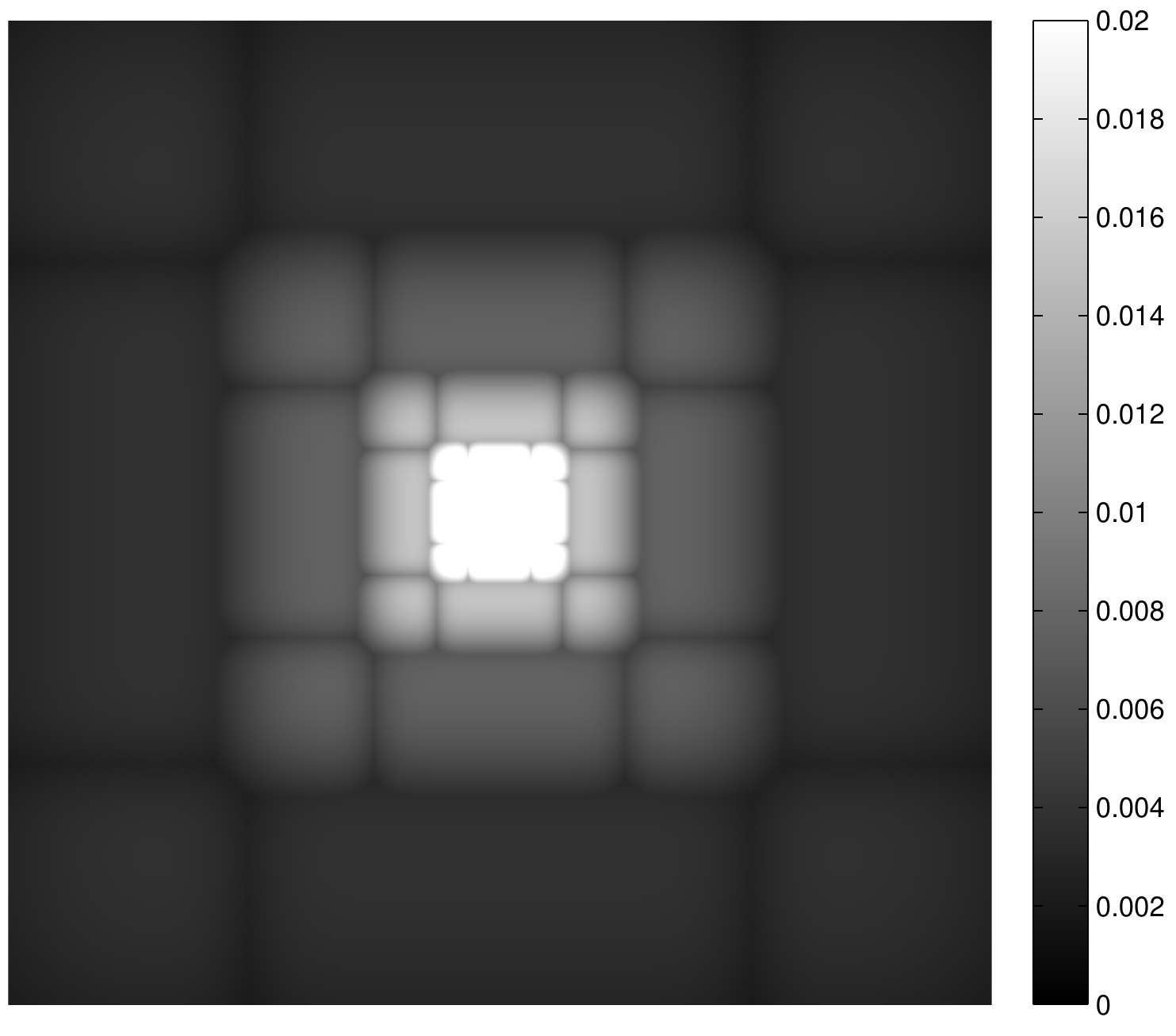}
 \caption{\footnotesize Daubechies16 Coherences }
\end{center}
\end{subfigure}
\begin{subfigure}[t]{0.32\textwidth}
\begin{center}
 \includegraphics[width=\textwidth]{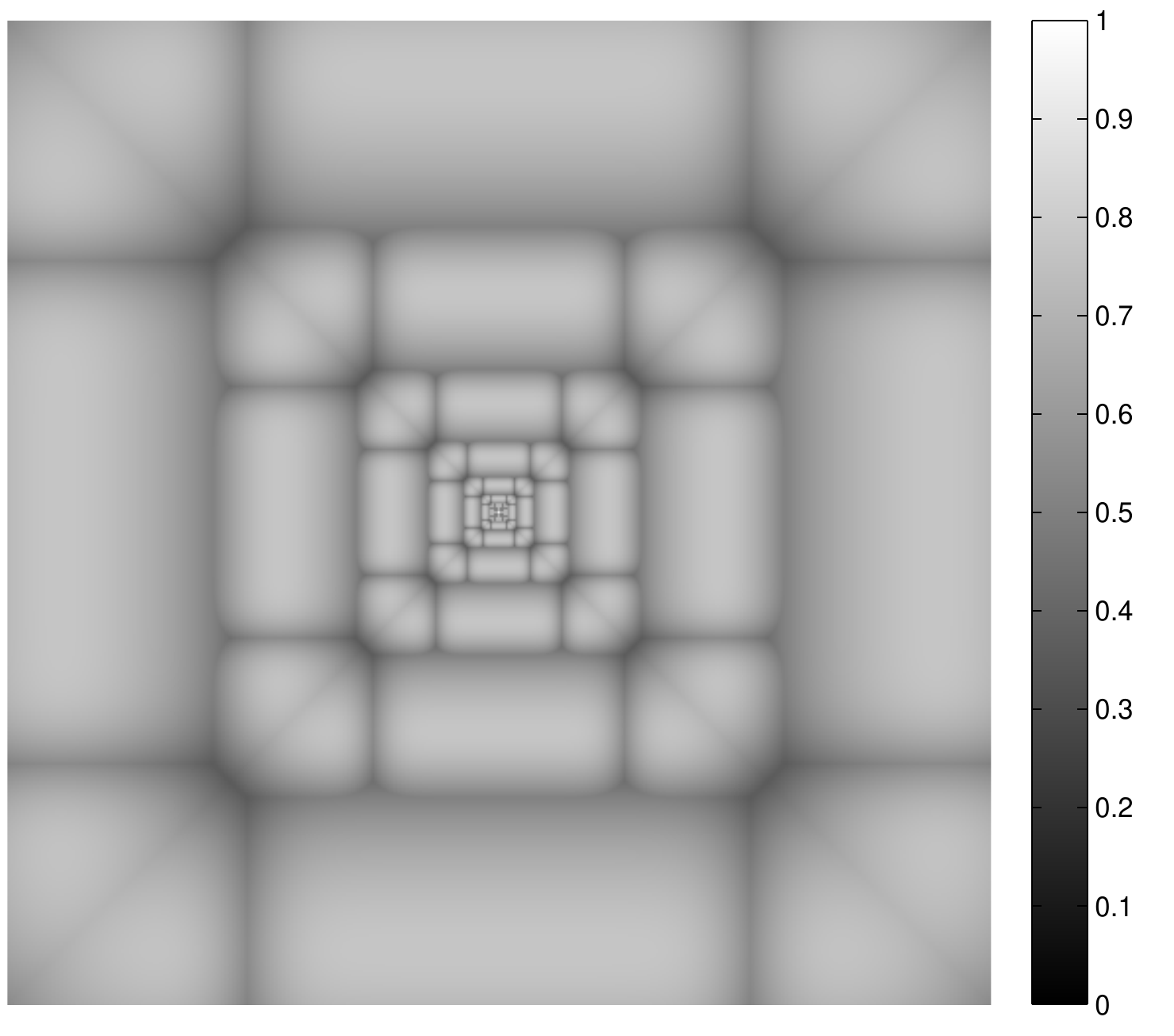} 
  \caption{\footnotesize  Scaled \\ Daubechies16  Coherences }
  \end{center}
\end{subfigure}
\begin{subfigure}[t]{0.32\textwidth}
\begin{center}
\end{center}
\end{subfigure}
\end{center}
\caption{2D Fourier -  Separable Wavelet Coherences. We show the subset $\{-250,-249,...,249,250\}^2 \subset \bbZ^2$. Notice again that the scaled coherences are bounded above zero and below 1 indicating that we have characterised the incoherence in terms of the linear scaling used, as shown in Proposition \ref{dDSeparableFourierWavelet}. The incoherences shown in the Figure are square rooted to reduce contrast.}

\label{separableincoherences}

\end{figure}

\subsection{Semi-Hyperbolic Orderings: Proof of Theorem \ref{separablesummary} Part (iii)} \label{semihypsection}

By Example \ref{3DHaar} we know that if (\ref{dDFTdecay}) does not hold then our approach of using a linear ordering can fail. We therefore return once more to (\ref{sepprodexample2}). Let us now try to use an approach that is halfway between our two previous linear/hyperbolic approaches. Let $r \in \{1,...,d-1\}$ be fixed. We shall first impose a decay condition that is stronger than (\ref{FTdecay}) but weaker than (\ref{dDFTdecay}):
\be{ \label{semiFTdecay}
|\mathcal{F}\phi (\omega)| \le \frac{K}{| \omega|^{d/2r} }, \qquad \omega \in \mathbb{R} \setminus \{ 0 \} .
}
Instead of just taking out the dominant term of the product in (\ref{sepprodexample2}), let us take out the $r$ smallest terms:
\be{ \label{semihypbound}
\begin{aligned}
| \langle  \Psi^{s}_{j,r} , \lambda_d^{-1}(n) \rangle |^2  & \le \epsilon^d 2^{-dj} \cdot \min_{\substack{i_1,...,i_r \in \{1,...,d\} \\ i_1<...<i_r}} \prod_{r=1}^r \frac{K^{2}}{|\epsilon 2^j n_{i_r}|}^{d/r}
\\ & =  K^{2r} \cdot \Big( \max_{\substack{i_1,...,i_r \in \{1,...,d\} \\ i_1<...<i_r}} \prod_{r=1}^r | n_{i_r}| \Big)^{-d/r}, \qquad n \in \bbZ^d, n_i \neq 0, i=1,...,d.
\end{aligned}
}
Again we can extend this bound to all $n \in \bbZ^d$:
\be{ \label{semihypbound2}
\sup_{g \in B_\text{sep}^d} | \langle  g , \lambda_d^{-1}(n) \rangle |^2  \le \max(K^{2r},1) \cdot  \Big( \max_{\substack{i_1,...,i_r \in \{1,...,d\} \\ i_1<...<i_r}} \prod_{r=1}^r \max(| n_{i_r}|,1) \Big)^{-d/r}, \qquad n \in \bbZ^d.
}
We deduce that the function $F_{\text{hyp},r}: \bbZ^d \to \bbR, F_{\text{hyp},r}(n)=\Big( \max_{\substack{i_1,...,i_r \in \{1,...,d\} \\ i_1<...<i_r}} \prod_{r=1}^r \max(| n_{i_r}|,1) \Big)^{-d/r}$ dominates the optimal decay of of $(B_\rf^d, B_\text{sep}^d)$.

\begin{definition} \label{semihyperbolic}
Let us define, for $r,d \in \bbN, r \le d$ the function
\[ H_{d,r}(n):= \max_{\substack{i_1,...,i_r \in \{1,...,d\} \\ i_1<...<i_r}}
\prod_{j=1}^r \max(|n_{i_j}|,1) , \qquad n \in \bbZ^d.
\]
Then we say an ordering $\sigma: \bbN \to \bbZ^d$ is `semi-hyperbolic of order $r$ in $d$ dimensions' if it is consistent with $H_{d,r}$.
\end{definition}
Figure \ref{Consistent3D} presents some isosurface plots of $H_{3,r}$ for the various values of $r$
Notice that a semi-hyperbolic ordering of order d in d dimensions corresponds to the hyperbolic cross in $\bbZ^d$ (see Example \ref{hypcrossZd}). Furthermore, if $\sigma: \bbN \to \bbZ^d$ is a semi-hyperbolic ordering of order $1$ in $d$ dimensions then, by Remark \ref{linearnorm}, $\sigma$ corresponds to a linear scaling scheme because $H_{d,1}(n)= \|n\|_\infty$ for the componentwise max norm $\| \cdot \|_\infty$ on $\bbR^d$. Like in the linear and hyperbolic cases discussed in the previous sections, we want to determine how $H_{d,r}(\sigma(n))$ scales with $n \in \bbN$. 

\begin{figure}[!t]
\begin{center}
\begin{subfigure}[t]{0.32\textwidth}
\begin{center}
\includegraphics[width=\textwidth]{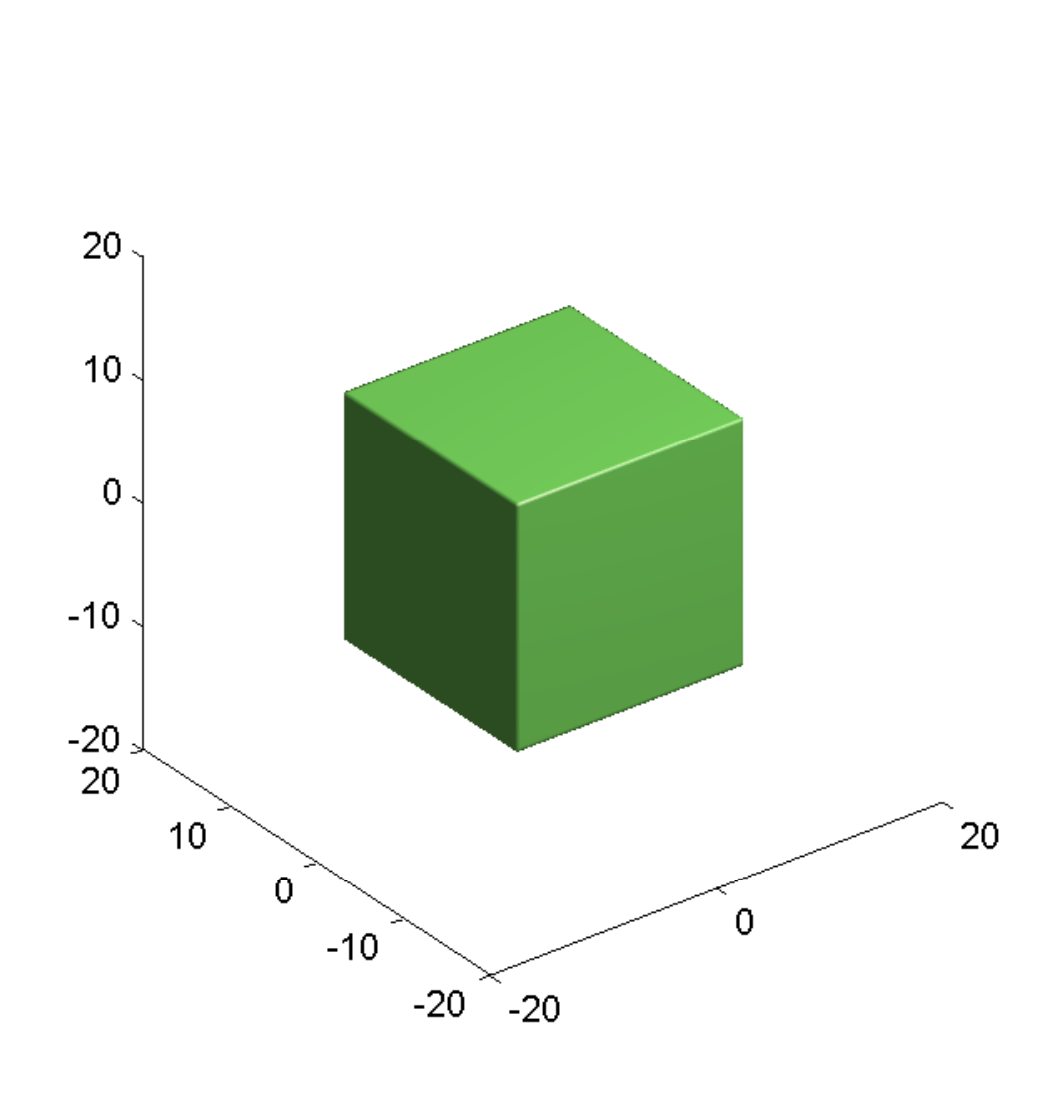}
 \caption{\footnotesize Case $r=1$ (Linear) ;  \\ Isosurface value=10. }
\end{center}
\end{subfigure}
\begin{subfigure}[t]{0.32\textwidth}
\begin{center}
 \includegraphics[width=\textwidth]{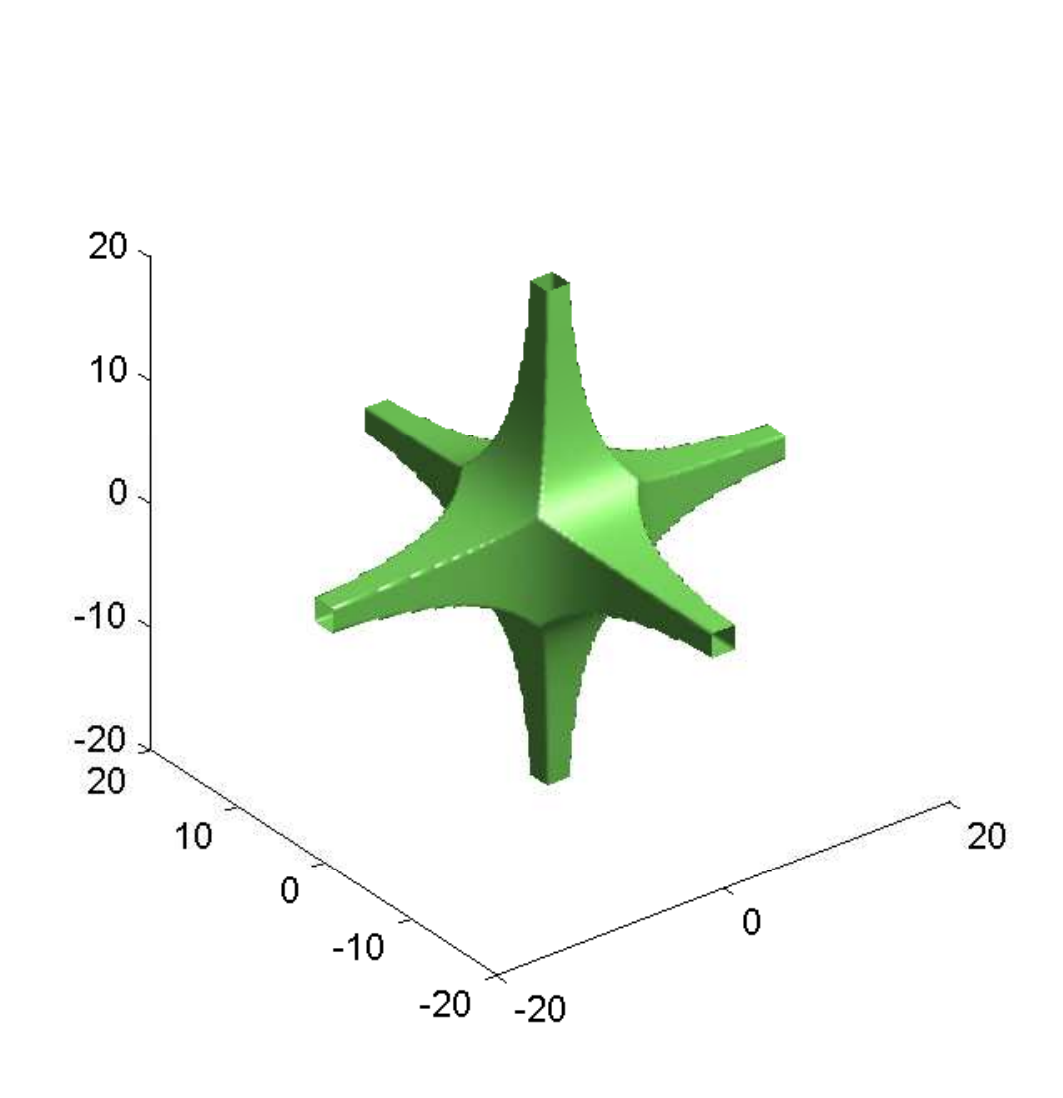} 
  \caption{\footnotesize Case $r=2$ (Semi-Hyperbolic) ; \\ Isosurface value=20. }
  \end{center}
\end{subfigure}
\begin{subfigure}[t]{0.32\textwidth}
\begin{center}
\includegraphics[width=\textwidth]{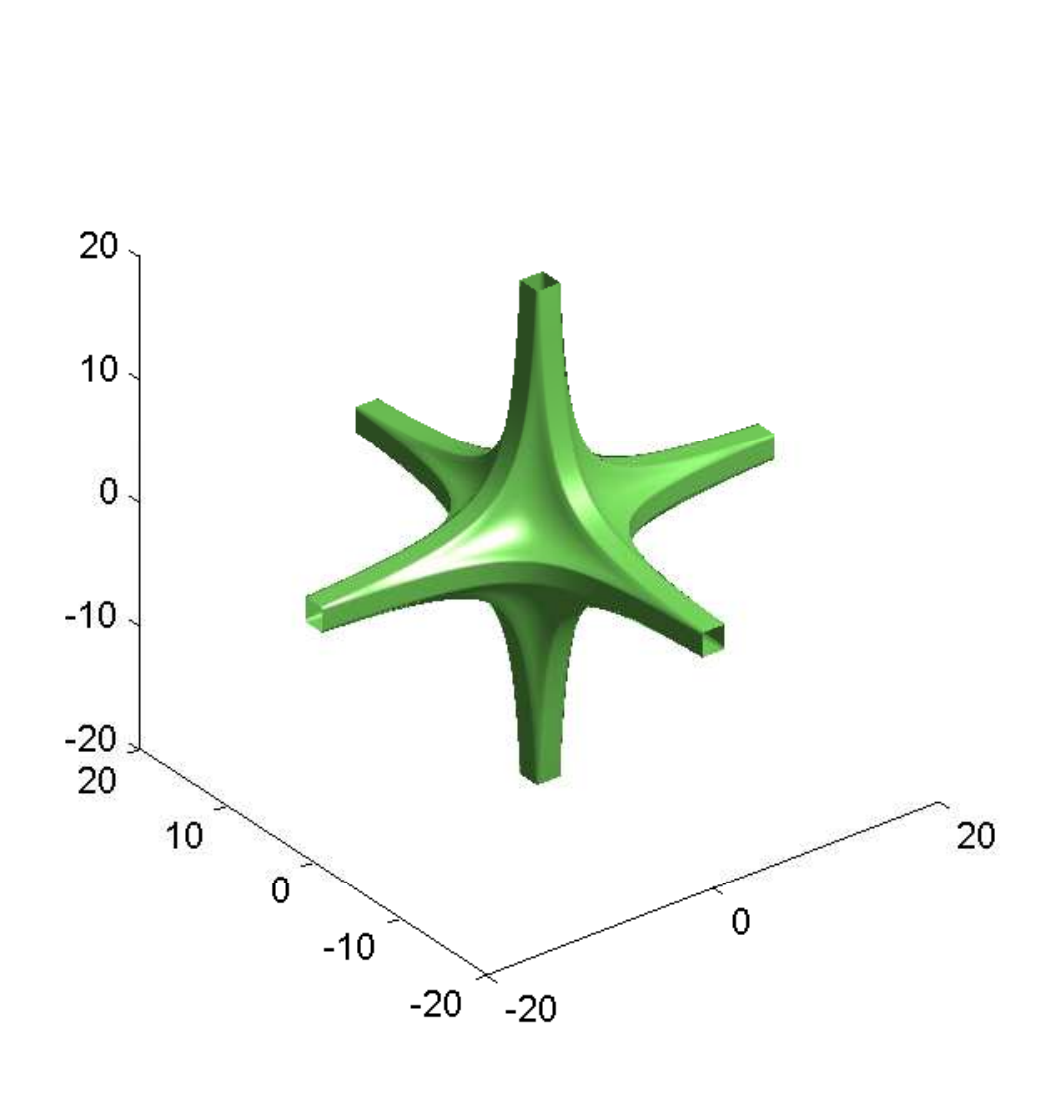} 
  \caption{\footnotesize Case $r=3$ (Hyperbolic) ; \\ Isosurface value=20. }
  \end{center}
\end{subfigure}
\end{center}
\caption{Isosurfaces of $H_{3,r}$, $r=1,2,3$ describing the three types of ordering available in 3D} 
\label{Consistent3D}
\end{figure}

\begin{lemma} \label{semihypbehaviour}
1). \ \ Let $r,d\in \bbN, r \le d-1$ be fixed. Let us define
\[ S_{d,r}(n):= \# \{ m \in \bbZ^d : H_{d,r}(m) \le n \}, \qquad n \in \bbN.\]
Then there is a constant $C>0$ such that
\[ n^{d/r} \le S_{d,r}(n) \le C \cdot n^{d/r}, \qquad n \in \mathbb{N}. \]
2). \ \ If $\sigma: \bbN \to \bbZ^d$ is semi-hyperbolic of order $r$ with $r \le d-1$ then there are constants $C_1,C_2>0$ such that
\[ C_1 \cdot n^{r/d} \le H_{d,r}(\sigma(n)) \le C_2 \cdot n^{r/d}, \qquad n \in \bbN. \]
\end{lemma}
\begin{proof}
1). \ \ For notational simplicity we prove the same bounds but with $S_{d,r}$ replaced by the smaller set
\[ \tilde{S}_{d,r}(n):= \# \{ m \in \bbN^d : H_{d,r}(m) \le n \}, \qquad n \in \bbN. \] 
The same bounds for $S_{d,r}$ then follows immediately, albeit with a larger constant $C>0$. The lower bound is straightforward since the set defining $\tilde{S}_{d,r}(n)$ contains the set $\{ m \in \bbN^d : m_i \le n^{1/r}, \ i=1,..,d \}$.
We prove the upper bound by induction on $r$. The case $r=1$ is clear because $\tilde{S}_{d,1}(n)$
 is simply the number of points inside a $d$-dimensional hypercube with side length $n$. Suppose the result holds for $r=r'-1$. We use the following set inclusion:
 \be{
 \begin{aligned}
  \{ m \in \bbZ^d : H_{d,r'}(m) \le n \} \subset & \{ m \in \bbN^d : m_i \le n^{1/r'}, \ i=1,..,d \}
  \\ & \cup \bigcup_{i=1}^d \{m \in \bbN^d : n^{1/r'} \le m_i \le n, \ H_{d-1,r'-1}(\tilde{m}_i) \le n/m_i \},
  \end{aligned}
 }
where $\tilde{m}_i$ here refers to $m$ with the $i$th entry removed. The cardinality of the first set on the right is just $n^{d/r'}$ and so we are done if we can show that for some constant $C>0$,
\[ \# \{m \in \bbN^d : n^{1/r'} \le m_1 \le n, \ H_{d-1,r'-1}((m_2,...,m_d)) \le n/m_1 \} \le C n^{d/r'}, \qquad n \in \bbN \]
We achieve this by applying our inductive hypothesis:
\be{
\begin{aligned}
 \# \{m & \in \bbN^d : n^{1/r'} \le m_1 \le n, \ H_{d-1,r'-1}((m_2,...,m_d)) \le n/m_1 \} \\
 & \le \sum_{  i= \lfloor n^{1/r'} \rfloor} ^n  S_{d-1,r'-1}(\lfloor n/i \rfloor) \le C' \cdot \sum_{ i= \lfloor n^{1/r'} \rfloor}^n \big( n/i \big)^{(d-1)/(r'-1)}
 \\ & \le C'n^{(d-1)/(r'-1)} \cdot \int_{n^{1/r'}-2}^n x^{-(d-1)/(r'-1)} \, dx 
 \\ & \le C'n^{(d-1)/(r'-1)} \cdot (n^{1/r'}-2)^{(1-(d-1)/(r'-1))}, \qquad ( \text{noting } r' \le d-1)
 \end{aligned}
 }
 where $C'>0$ is some constant. We can replace $(n^{1/r'}-2)$ by $n^{1/r'}$ in the above by changing the constant $C'$ and assuming $n>2^{r'}$. Finally, we notice that the exponents add to the desired expression:
 \[ \frac{d-1}{r-1} + \frac{1}{r'}\Big( 1- \frac{d-1}{r'-1} \Big)= \frac{d-1}{r'-1} - \frac{ d-r'}{r'(r'-1)} = \frac{d}{r'}.\]
 This gives the required upper bound for $n>2^{r'}$. Since the terms involved are all positive, we can just increase the constant $C'$ to include the cases $n \le 2^{r'}$. This shows that the result holds for $r=r'$ and the induction argument is complete.
 
 2.) \ \  By consistency we know that
 \[ S_{d,r}(H_{d,r}(\sigma(n))-1) \le n \le S_{d,r}(H_{d,r}(\sigma(n))), \qquad n \in \bbN \]
 and therefore we can directly apply part 1 to deduce
 \[ ( H_{d,r}(\sigma(n))-1)^{d/r} \le n \le C \cdot ( H_{d,r}(\sigma(n)))^{d/r}, \qquad n \in \bbN, \]
 from which the result follows.
\end{proof}
Armed with this result, we can now completely tackle the separable wavelet case.

\begin{theorem} \label{semihyperbolicthm}
Suppose that the scaling function $\phi$ corresponding to the separable wavelet basis $B^d_{\text{sep}}$, satisfies (\ref{semiFTdecay}) for some constant $K\ge 0$ and $r \in \{1,...,d-1\}$.
Next let $\sigma: \bbN \to \bbZ^d$ be semi-hyperbolic of order $r$ in $d$ dimensions and $\rho:=\lambda_d^{-1} \circ \sigma$. Finally, we let $U=[(B^d_\rf(\epsilon),\rho),(B^d_{\text{sep}}, \tau)]$, where $\tau$ is an ordering of $B^d_{\text{sep}}$. Let us also fix $\epsilon \in I_{J,p}$. Then there are constants $C_1, C_2$ such that 
\[ \frac{C_1}{N} \le \mu(Q_N U)  \le \frac{C_2}{N},  \quad N \in \bbN. \]
Furthermore it follows that the ordering $\rho$ is optimal for the basis pair $(B^d_\rf(\epsilon),B^d_\text{sep})$.
\end{theorem}
\begin{proof}  
Applying part 1.) from Proposition \ref{dDSeparableFourierWavelet} (with $\epsilon$ fixed) to part 2.)  of Lemma \ref{characterisationlemma} immediately gives us the lower bound for the semihyperberbolic ordering since this bound also holds for the optimal decay rate. Furthermore this lower bound holds for any other ordering and therefore if we have the upper bound then the ordering $\rho$ is automatically optimal. We now focus on the upper bound. 

By (\ref{semihypbound2}) we know that the optimal decay of $(B_\rf^d(\epsilon),B^d_\text{sep})$ is dominated by $F_\text{hyp,r}$. Therefore by part 1.) of Lemma \ref{characterisationlemma} if $\sigma: \bbN \to \bbZ^d$ is consistent with $1/F_\text{hyp,r}$, i.e. $\sigma$ is semihyperbolic of order $r$ then we can bound the row incoherence $\mu(\pi_N U)$ by $F_\text{hyp,r}(\sigma(N))= H_{d,r}^{-d/r}(\sigma(N)) \approx (N^{r/d})^{-d/r}=N^{-1}$ by Lemma \ref{semihypbehaviour}. Since $N^{-1}$ is decreasing this bound extends to $\mu(Q_N U)$.
\end{proof}

Finally we can summarise our results on the $(B_\rf^d (\epsilon),B^d_{\text{sep}})$ case as follows:
\begin{theorem} \label{SeparableResults}
Let $\rho$ be a Linear ordering of the d-dimensional Fourier basis $B_\rf^d(\epsilon)$ with $\epsilon \in I_{J,p}$, $\tau$ a leveled ordering of the d-dimensional separable wavelet basis $B^d_{\text{sep}}$ and $U=[(B_\rf^d(\epsilon),\rho),(B^d_{\text{sep}},\tau)]$. Furthermore, suppose that the decay condition (\ref{FTdecay}) holds for the wavelet basis. Then, keeping $\epsilon>0$ fixed, we have, for some constants $C_1, C_2>0$ the decay
\be{ \label{dDSeparableFourierPolynomialLinearBounds}  
\frac{C_1}{N} \le \mu(\pi_N U), \ \mu(U \pi_N), \le \frac{C_2}{N} \qquad \forall N \in \mathbb{N}.
}
Let us now instead replace $\rho$ by a semi-hyperbolic ordering of order $r$ in $d$ dimensions with $r \in \{1,....,d-1\}$ and assume the weaker decay condition (\ref{semiFTdecay}). Then, keeping $\epsilon>0$ fixed, we have, for some constants $C_1, C_2>0$ the decay
\be{ \label{dDSeparableFourierPolynomialSemiHypBounds}  
\frac{C_1}{N} \le \mu(Q_N U), \ \mu(U \pi_N), \le \frac{C_2}{N} \qquad \forall N \in \mathbb{N},
}
and furthermore $\rho$ is optimal for the basis pair $(B_\rf^d(\epsilon),B^d_{\text{sep}})$. Since, for any separable Daubechies wavelet basis, (\ref{semiFTdecay}) always holds for $r=d-1$  any semi-hyperbolic ordering of order $d-1$ in $d$ dimensions will produce (\ref{dDSeparableFourierPolynomialSemiHypBounds}).
\end{theorem}
\begin{proof}
(\ref{dDSeparableFourierPolynomialLinearBounds}) follows from Corollaries \ref{leveledresults} and \ref{linearresults}. (\ref{dDSeparableFourierPolynomialSemiHypBounds}) follows from Corollary \ref{leveledresults} and Theorem \ref{semihyperbolicthm}. To show that (\ref{dDSeparableFourierPolynomialSemiHypBounds}) always holds for a $d-1$ degree semi-hyperbolic ordering in $d$ dimensions, we note that the weakest decay on the scaling function $\phi$ is $| \mathcal{F} \phi (\omega)| \le K \cdot | \omega |^{-1}$ (see Lemma \ref{FTdecayLemma}) and therefore (\ref{semiFTdecay}) is automatically satisfied for $r=d-1$.
\end{proof}

\subsection{Optimal Orderings \& Wavelet Smoothness} \label{orderingsandsmoothness}

Theorem \ref{SeparableResults} demonstrates how certain degrees of smoothness, in terms of decay of the Fourier transform, allows us to show certain orderings are optimal and this smoothness requirement becomes increasingly more demanding as the dimension increases. But if a certain ordering is optimal for the basis pair $(B_\rf^d, B^d_\text{sep})$, does this mean that the wavelet must also have some degree of smoothness as well? The answer to this question turns out to be yes, and it is the goal of this section to prove this result.

We shall rely heavily on the following simple result from \cite[Thm. 9.4]{korner}:

\begin{theorem} \label{kornertheorem}
Let\footnote{$\bbR/\bbZ$ denotes the unit circle which we write as $[0,1)$ with the quotient topology induced by $M: \bbR \to [0,1), M(x)=x (\text{mod} \ 1)$.} $f: \bbR/\bbZ \to \bbC$ be continuous and for $k \in \bbZ$ define $\hat{f}(k)= \int_0^1 f(x) \exp(2 \pi \ri k x) \, dx$. If $\sum_{k=-\infty}^\infty |k||\hat{f}(k)| < \infty$ then $f \in C^1$. Consequently, using $\widehat{f'}(k)=(2 \pi \ri k)^{-1} \cdot \hat{f}(k)$, if $\sum_{k=-\infty}^\infty |k|^n|\hat{f}(k)| < \infty$ then $f \in C^n$.
\end{theorem}

Now the main result itself:

\begin{theorem} \label{ordering2smooth}
Let $\sigma : \bbN \to \bbZ^d$ be semihyperbolic of order $r<d$ in $d$ dimensions and let $\rho:= \lambda_d^{-1} \circ \sigma : \bbN \to B_\rf^d(\epsilon)$ where $\epsilon \in I_{J,p}$. Then if $\rho$ is optimal for the basis pair $(B_\rf^d(\epsilon),B_\text{sep}^d)$ then $\phi \in C^l$ for any $l \in \bbN \cup \{0\}$ with $l+1 < d/2r$.
\end{theorem}
\begin{proof}
By Theorem \ref{SeparableResults} we know that the optimal decay rate for the basis pair is $N^{-1}$, therefore if $\rho$ is optimal we must have, for some constant $C_1>0$, the bound
\[
\sup_{g \in B^d_\text{sep}} | \langle g, \lambda_d^{-1} \circ \sigma(N) \rangle |^2 \le C_1 \cdot N^{-1}, \qquad N \in \bbN.
\]
Next since $\sigma$ is semihyperbolic we also know that, by Lemma \ref{semihypbehaviour}, there is a constant $C_2>0$ such that 
\[
 H_{d,r}(\sigma(N)) \le C_2 \cdot N^{r/d}, \qquad N \in \bbN.
\]
Consequently we deduce,
\be{ \label{conversion}
\begin{aligned}
\sup_{g \in B^d_\text{sep}} | \langle g, \lambda_d^{-1} \circ \sigma(N) \rangle |^2 \le C_1 \cdot N^{-1} \le C_1 C_2^{-d/r} \cdot H^{-d/r}_{d,r}(\sigma(N)), \qquad N \in \bbN.
\\
\Rightarrow \sup_{g \in B^d_\text{sep}} | \langle g, \lambda_d^{-1}(n) \rangle |^2 \le \frac{C_1 C_2^{-d/r}}{ \big( \max_{\substack{i_1,...,i_r \in \{1,...,d\} \\ i_1<...<i_r}}
\prod_{j=1}^r \max(n_{i_j},1) \big)^{d/r}}, \qquad n \in \bbZ^d.
\end{aligned}
}
Letting $g=\Psi^s_{J,0}$ where $s=\{0,...,0\}$ and $n=(k,0,...,0)$ for $k \in \bbZ$ we see that (\ref{conversion}) becomes 
\be{ \label{phiprior}
\epsilon^d 2^{-dJ} | \mathcal{F}\phi(2^{-J} \epsilon k) |^2 \le \frac{C_1 C_2^{-d/r}}{\max(|k|,1)^{d/r}},  \qquad k \in \bbZ. 
}
Since the scaling function $\phi$ has compact support in $[-p+1,p]$ and $\epsilon \in I_{J,p}$, $\phi_{J,0}$ can be viewed as a function on $\bbR/\bbZ$ and (\ref{phiprior}) describes a bound on the Fourier coefficients of $\phi$. Formally, if we write $\varphi(x):=\phi(2^J \epsilon^{-1}(x-1/2))$, then since $\epsilon \in I_{J,p}$ we have that $\varphi$ is supported in $[0,1]$ and (\ref{phiprior}) becomes, for some constant $D(\epsilon,J,p)>0$:
\[
| \mathcal{F}\varphi(k) |^2 = | \widehat{\varphi}(k)|^2 \le \frac{D}{\max(|k|,1)^{d/r}},  \qquad k \in \bbZ. 
\] 
If $\phi \in C^{0}$ then the result follows from Theorem \ref{kornertheorem}. If $\phi \notin C^0$, i.e. $\phi$ corresponds to a Haar wavelet basis, then (\ref{phiprior}) cannot hold with $d/2r > 1$ as this would contradict (\ref{FTdecayfail}).
\end{proof}

\begin{corollary} \label{hyperbolictendency}
Let the scaling function $\phi$ corresponding to the Daubechies wavelet basis $B^d_\text{sep}$ be fixed. Then for every order $r \in \bbN$, there exists a dimension $d' \in \bbN, d'>r$ such that for all $d \ge d'$, we have that a semihyperbolic ordering $\sigma $ of order $r$ in $d$ dimensions is such that $\rho= \lambda_d^{-1} \circ \sigma$ is not optimal for the basis pair $(B_\rf^d(\epsilon),B^d_\text{sep})$.
\end{corollary}
\begin{proof}
If this result was not true then we would deduce by Theorem \ref{ordering2smooth} that the wavelet $\phi$ satisfies $\phi \in C^{\infty}$, which is a contradiction because no compactly supported wavelet can be infinitely smooth \cite[Thm. 3.8]{wav}.
\end{proof}

\subsection{Hierarchy of Semihyperbolic Orderings}

One other notable point from Theorem \ref{SeparableResults} is that we can have multiple values of $r$ such that if $\sigma$ is semi-hyperbolic of order $r$ in $d$ dimensions then $\rho= \lambda_d^{-1} \circ \sigma$ is optimal for the basis pair $(B_\rf^d, B^d_\text{sep})$, so which one should we choose? We know that in the case of sufficient smoothness linear orderings are strongly optimal and therefore this suggests that the lower the order $r$ the stronger the optimality result. We now seek to prove this conjecture.

\begin{lemma} \label{semihypinequality}
Let $r,r',d \in \bbN, r \le r' \le d$. Then for all $n \in \bbZ^d$ we have that $H_{d,r'}^r(n) \le H_{d,r}^{r'}(n)$.
\end{lemma}
\begin{proof}
Let $n \in \bbZ^d$ be fixed. For each $j=1,..d$ let $i_j$ denote the $j$th largest terms of the form $\max(|n_{i_j}|,1)$. Observe that 
\[
H_{d,r}(n)= \prod_{j=1}^r \max(|n_{i_j}|,1), \quad H_{d,r'}(n)= \prod_{j=1}^{r'} \max(|n_{i_j}|,1),
\]
\[
\Rightarrow \frac{ H_{d,r}^{r'}(n)}{H_{d,r'}^r(n)} = \frac{\prod_{j=1}^r \max(|n_{i_j}|,1)^{r'-r}}{\prod_{j=r+1}^{r'}  \max(|n_{i_j}|,1)^r}.
\]
Finally we observe that the numerator and denominator have the same number ($r(r'-r)$) of terms in the product and that each term in the numerator is greater than each term in the denominator, proving the inequality.
\end{proof}

\begin{corollary} \label{hierarchycorollary}
Let $r,r',d \in \bbN, r \le r' < d$ and $\sigma, \sigma'$ be semihyperbolic of orders $r,r'$ in $d$ dimensions respectively. If $\rho=\lambda_d^{-1} \circ \sigma$ is optimal for the basis pair $(B_\rf^d(\epsilon), B_\text{sep}^d)$ then so is $\rho'=\lambda_d^{-1} \circ \sigma'$.

\end{corollary}
\begin{proof}
Recalling (\ref{conversion}) we know that there is a constant $C>0$ such that 
\[
\sup_{g \in B^d_\text{sep}} | \langle g, \lambda_d^{-1} \circ \sigma(N) \rangle |^2 \le C \cdot H^{-d/r}_{d,r}(\sigma(N)), \qquad N \in \bbN,
\]
\[
\Rightarrow  \sup_{g \in B^d_\text{sep}} | \langle g, \lambda_d^{-1} \circ \sigma'(N) \rangle |^2 \le C' \cdot H^{-d/r'}_{d,r'}(\sigma'(N)), \qquad N \in \bbN,
\]
where we have used Lemma \ref{semihypinequality} on the second line. We then apply Lemma \ref{semihypbehaviour} to deduce the result.

\end{proof}

\begin{figure}[!t]
\begin{center}
\begin{subfigure}[t]{0.49\textwidth}
\begin{center}
\includegraphics[width=\textwidth]{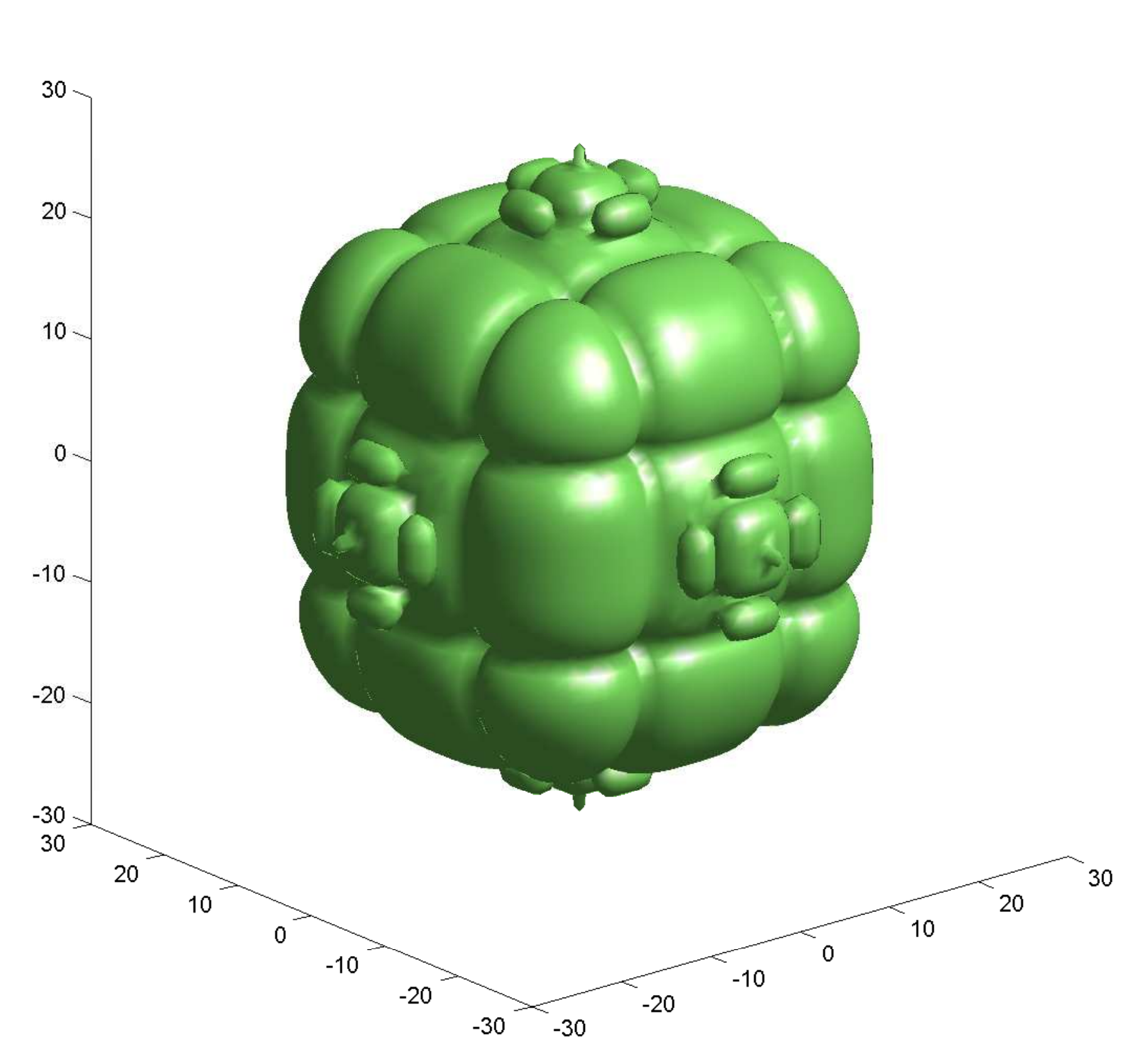}
 \caption{\footnotesize Daubechies4 - Isosurface Value $ = 5 \cdot 10^{-3}$}
\end{center}
\end{subfigure}
\begin{subfigure}[t]{0.49\textwidth}
\begin{center}
 \includegraphics[width=\textwidth]{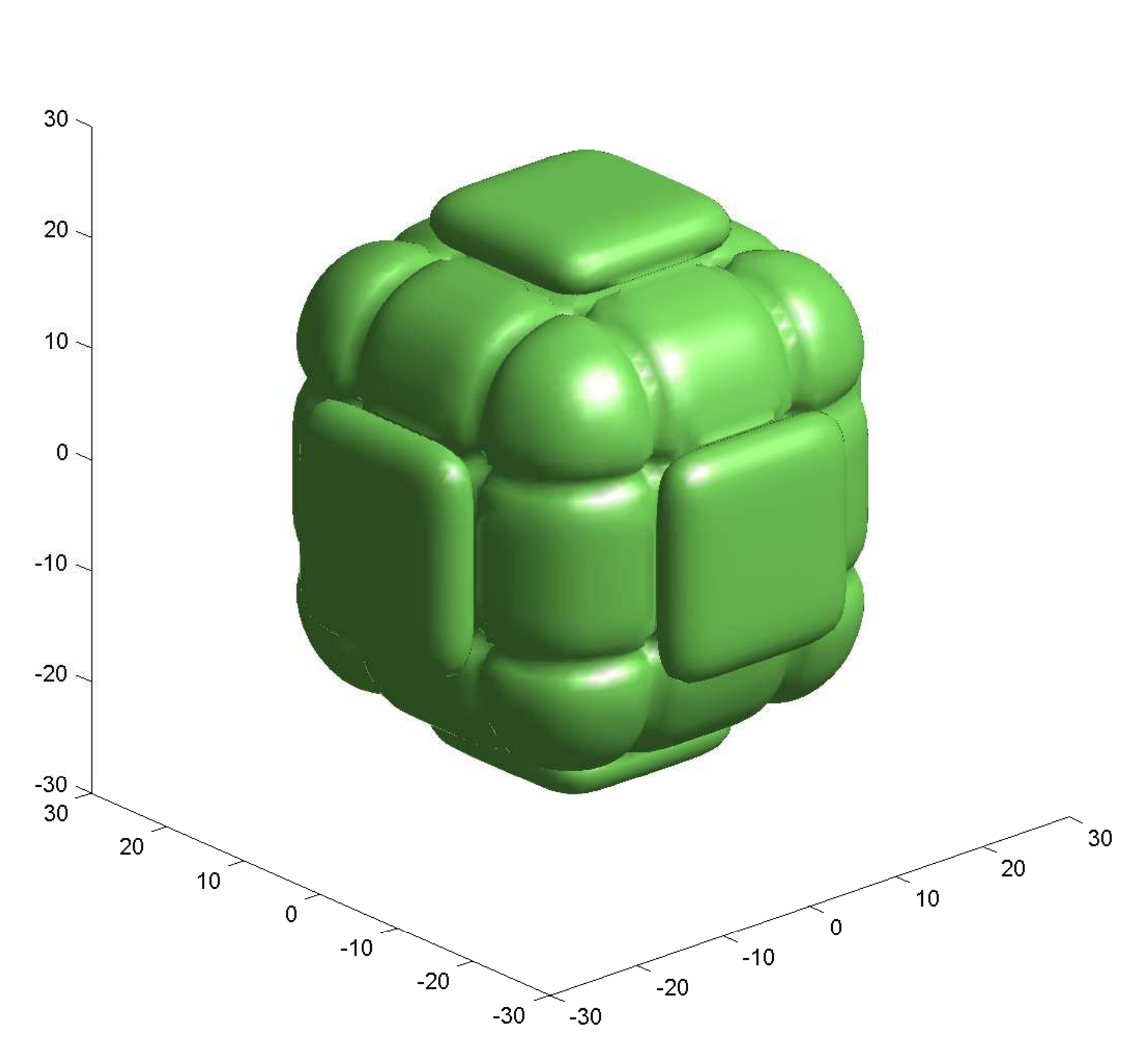} 
  \caption{\footnotesize Daubechies8 - Isosurface Value $ = 5 \cdot 10^{-3}$}
  \end{center}
\end{subfigure}
\begin{subfigure}[t]{0.49\textwidth}
\begin{center}
\includegraphics[width=\textwidth]{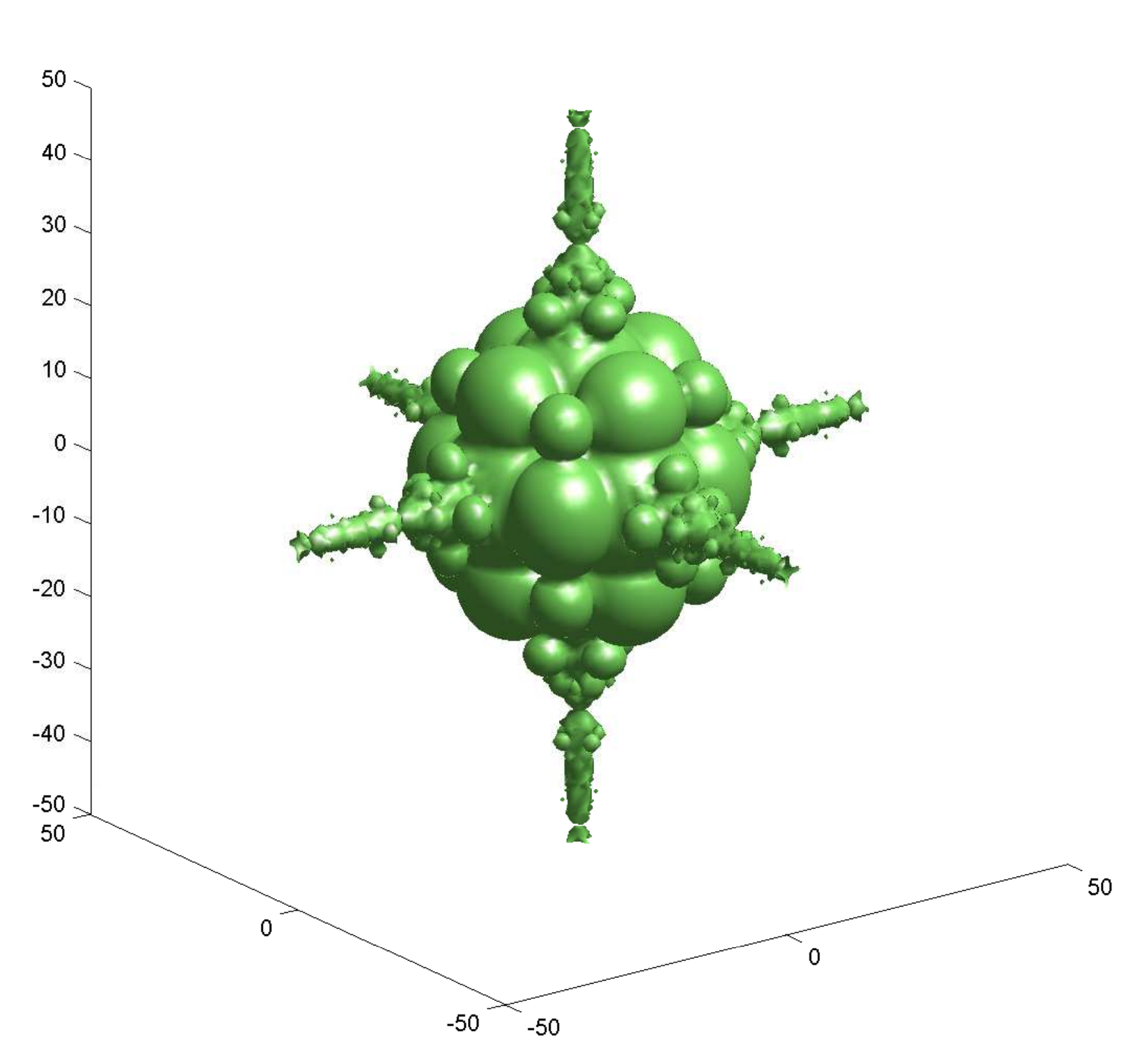} 
  \caption{\footnotesize Haar - Isosurface Value $ = 5 \cdot 10^{-3}$}
  \end{center}
\end{subfigure}
\begin{subfigure}[t]{0.49\textwidth}
\begin{center}
\includegraphics[width=\textwidth]{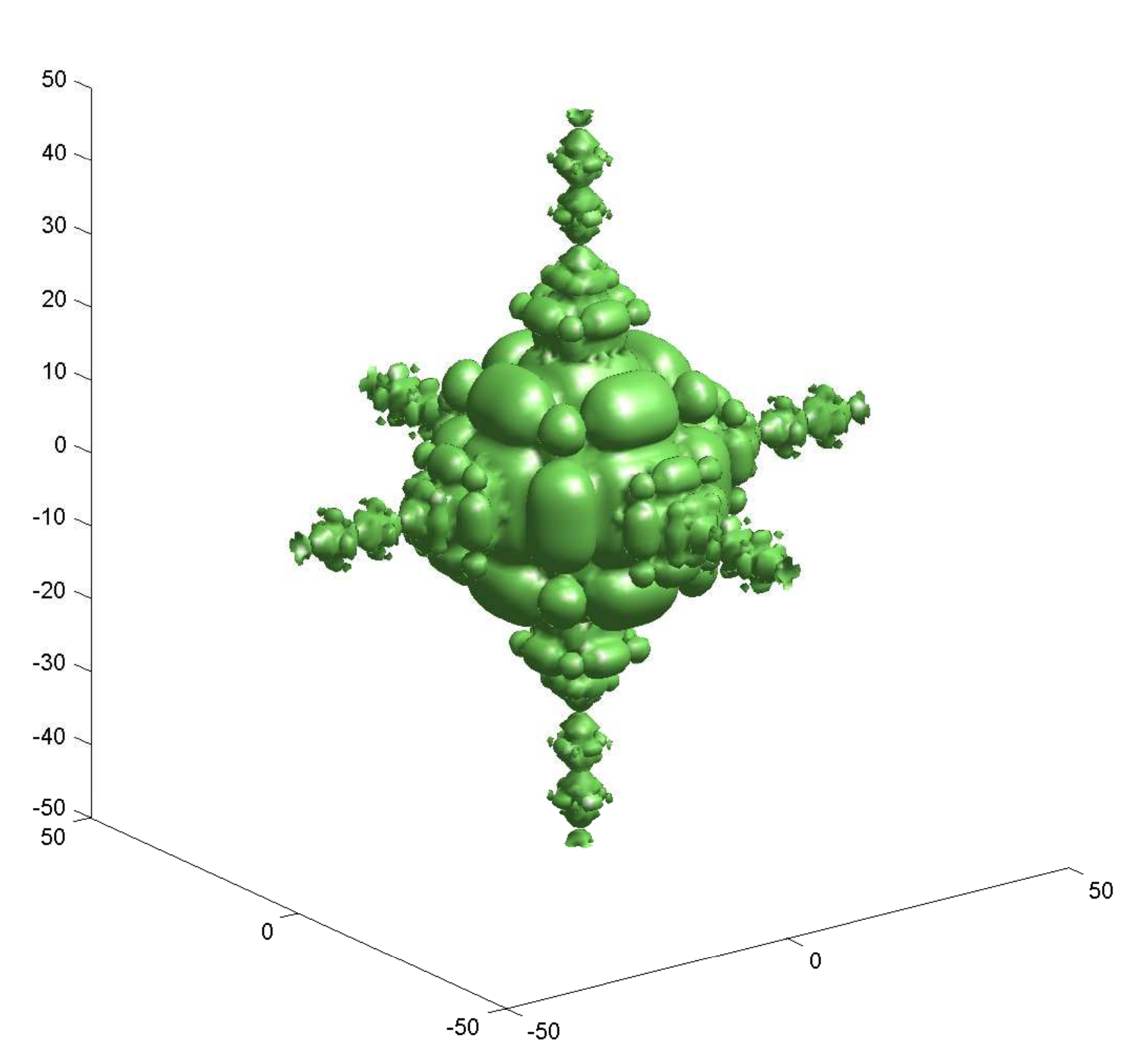}
 \caption{\footnotesize 3D slice of 5D Daubechies4 - Isosurface Value $ = 5 \cdot 10^{-4}$ }
\end{center}
\end{subfigure}
\end{center}
\caption{3D Fourier -  Separable Wavelet Incoherence Isosurface Plots. We draw the isosurface plots over the subset $\{-50,-49,...,49,50\}^3 \subset \bbZ^3$. These pictures should be compared with the ordering plots in Figure \ref{Consistent3D}. Notice that for the smoother wavelets in (a) \& (b), the growth matches that of a linear ordering however the 3D Haar case lacks this smoothness, resulting in semi-hyperbolic scaling in (c). If we keep the wavelet basis fixed and let the dimension increase, the scaling becomes increasingly hyperbolic, as seen in (d) and proved in Corollary \ref{hyperbolictendency}. }

\label{3dseparableincoherences}

\end{figure}  

\begin{remark}
Corollary \ref{hierarchycorollary} tells us that if there are several orders $r$ that give us optimality then the smallest $r$ possible, say $r^*$, is the strongest result.
\end{remark}

\subsection{3D Separable Incoherences}

We have found optimal orderings for every multidimensional Fourier- separable wavelet case however, we have not shown that (apart from in the linear case with sufficient Fourier decay) that the ordering is strongly optimal and we have not characterized the decay. Therefore it is of interest to see how the incoherence scales in further detail by directly imaging them in 3D. We do this by drawing levels sets in $\bbZ^3$, as seen in Figure \ref{3dseparableincoherences}.

\section{Asymptotic Incoherence and Compressed Sensing in Levels} \label{numericalsection}

We now return to the original compressed sensing problem which was described in the introduction of this paper and aim to study how asymptotic incoherence can influence the ability to subsample effectively. We shall be  working  exclusively in 2D for this section.

Consider the problem of reconstructing a function $f \in L^2([-1,1]^2)$ from its samples $\{ \langle f, g \rangle : g \in B^2_\rf (2^{-1}) \} $. The function $f$ is reconstructed as follows: Let $U:=[(B^2_\rf(2^{-1}),\rho), (B_2, \tau)]$ for some orderings $\rho, \tau$ and a reconstruction basis $B_2$. The number $2^{-1}$ is present here to ensure the span of $B_\rf$ contains $L^2([-1,1]^2)$.  Next let $\Omega \subset \bbN$ denote the set of subsamples from $B^2_\rf(2^{-1})$ (indexed by $\rho$), $P_\Omega$ the projection operator onto $\Omega$ and $\hat{f}:=( \langle f , \rho(m) \rangle )_{m \in \bbN}$. We then attempt to approximate $f$ by $\sum_{n=1}^\infty \tilde{x}_n \tau(n)$ where $\tilde{x} \in \ell^1(\bbN)$ solves the optimisation problem
\be{ \label{basicl1full}
\min_{x \in \ell^1(\bbN)} \| x \|_1  \quad \text{subject to} \quad P_\Omega Ux= P_\Omega \hat{f} .
}
\begin{figure}[!t]
\begin{center}
\begin{subfigure}[t]{0.3\textwidth}
\begin{center}
\includegraphics[width=\textwidth]{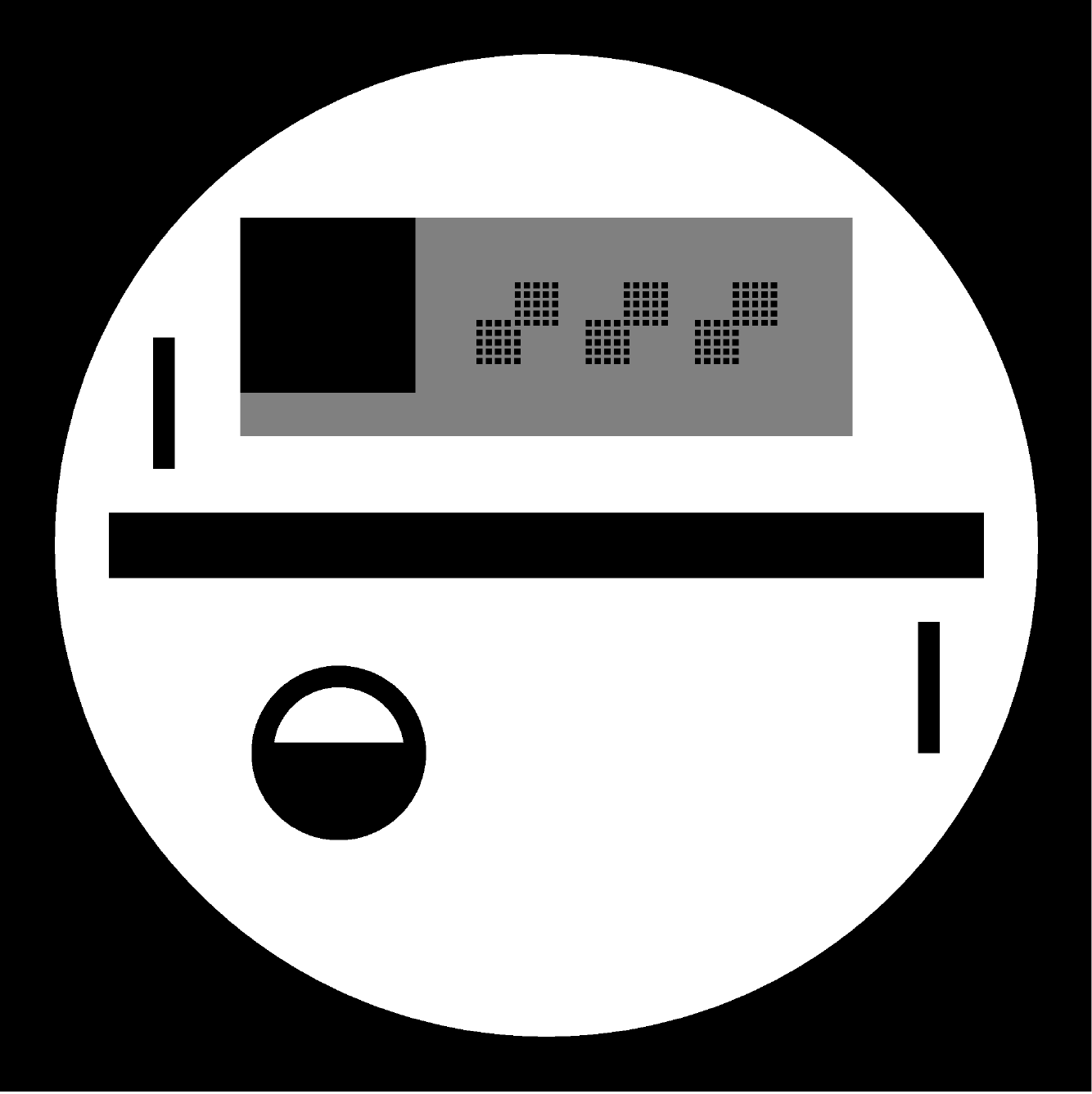}
 \caption{\footnotesize Rasterized Phantom \\ Resolution = $2^{12} \times 2^{12}$}
\end{center}
\end{subfigure}
\begin{subfigure}[t]{0.3\textwidth}
\begin{center}
 \includegraphics[width=\textwidth]{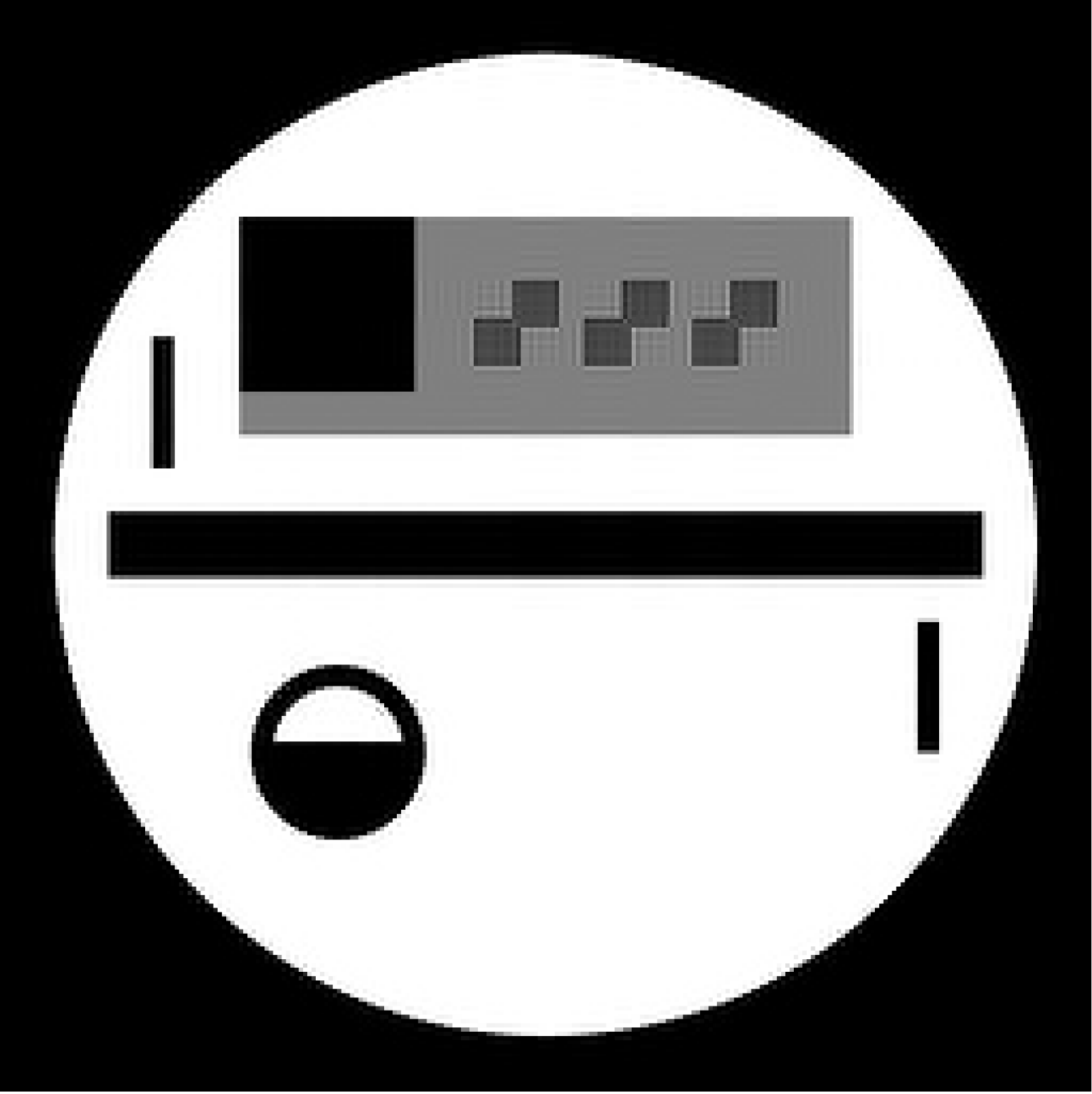} 
  \caption{\footnotesize Reconstruction from pattern A \\ $L^1$ error = 0.0735}
  \end{center}
\end{subfigure}
\begin{subfigure}[t]{0.3\textwidth}
\begin{center}
\includegraphics[width=\textwidth]{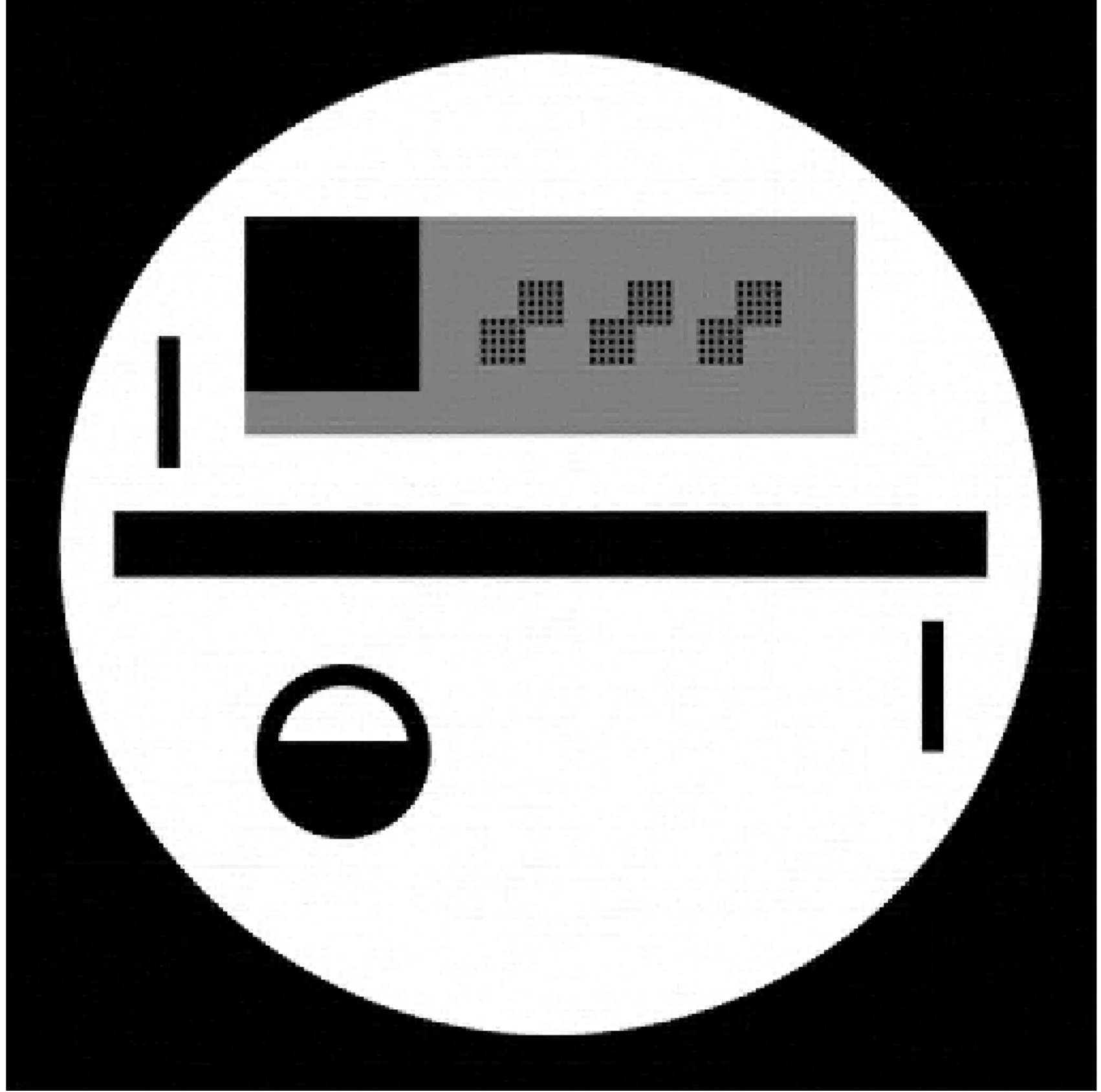} 
  \caption{\footnotesize Reconstruction from pattern B \\ $L^1$ error = 0.0620}
  \end{center}
\end{subfigure}
\begin{subfigure}[t]{0.3\textwidth}
\begin{center}
\includegraphics[width=\textwidth]{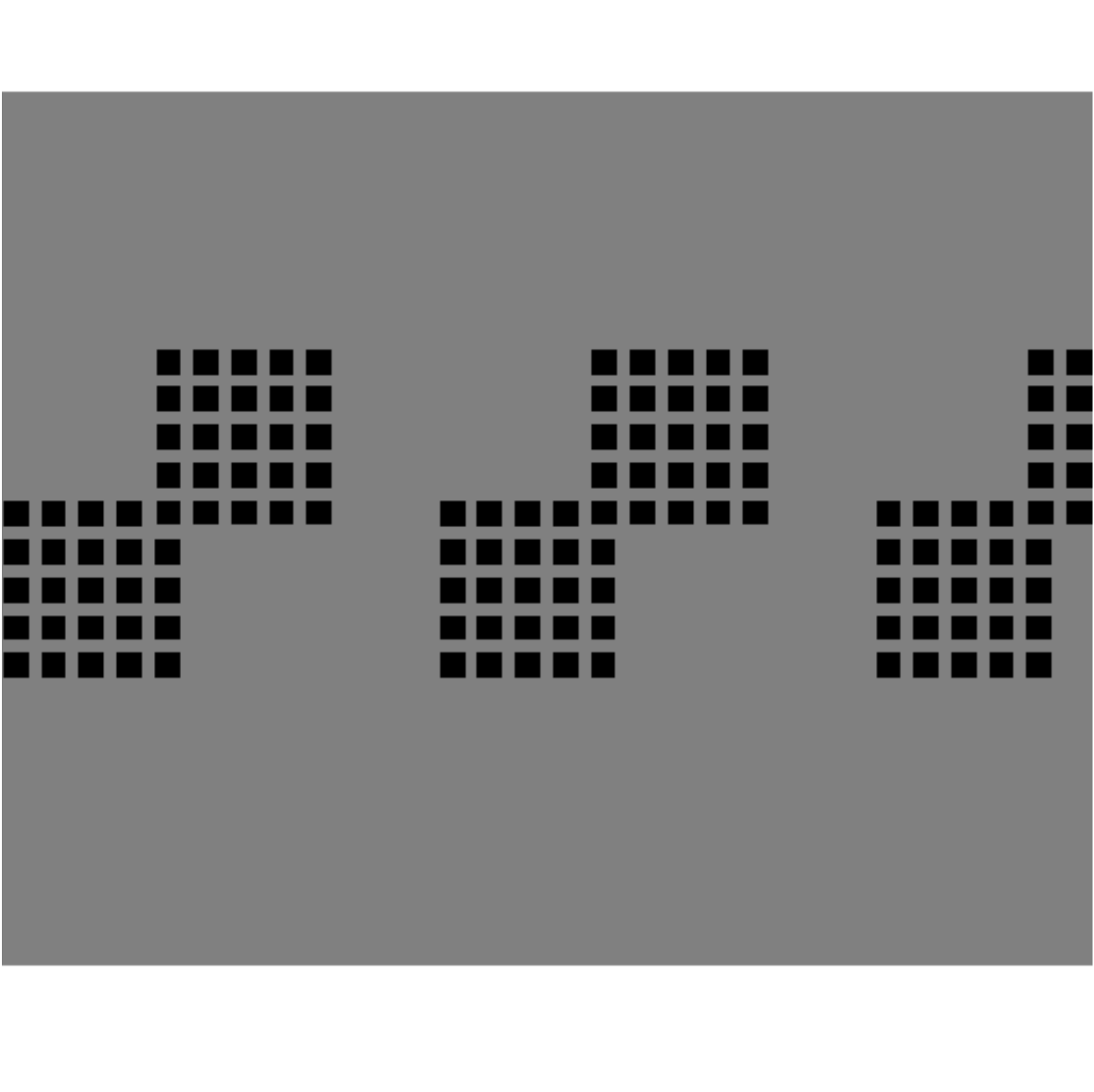}
 \caption{\footnotesize Rasterized Phantom - Closeup}
\end{center}
\end{subfigure}
\begin{subfigure}[t]{0.3\textwidth}
\begin{center}
 \includegraphics[width=\textwidth]{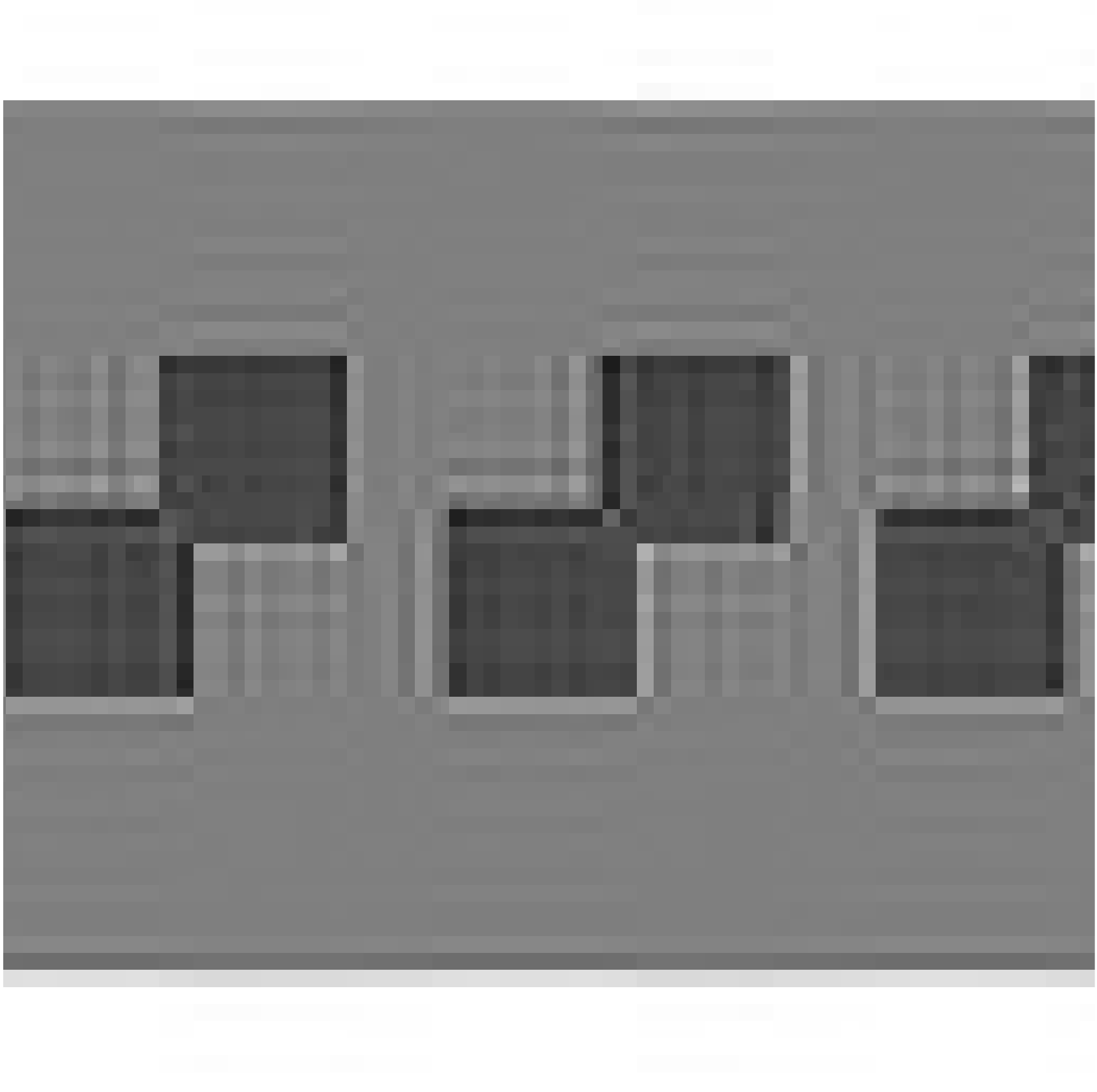} 
  \caption{\footnotesize Closeup of (b)}
  \end{center}
\end{subfigure}
\begin{subfigure}[t]{0.3\textwidth}
\begin{center}
\includegraphics[width=\textwidth]{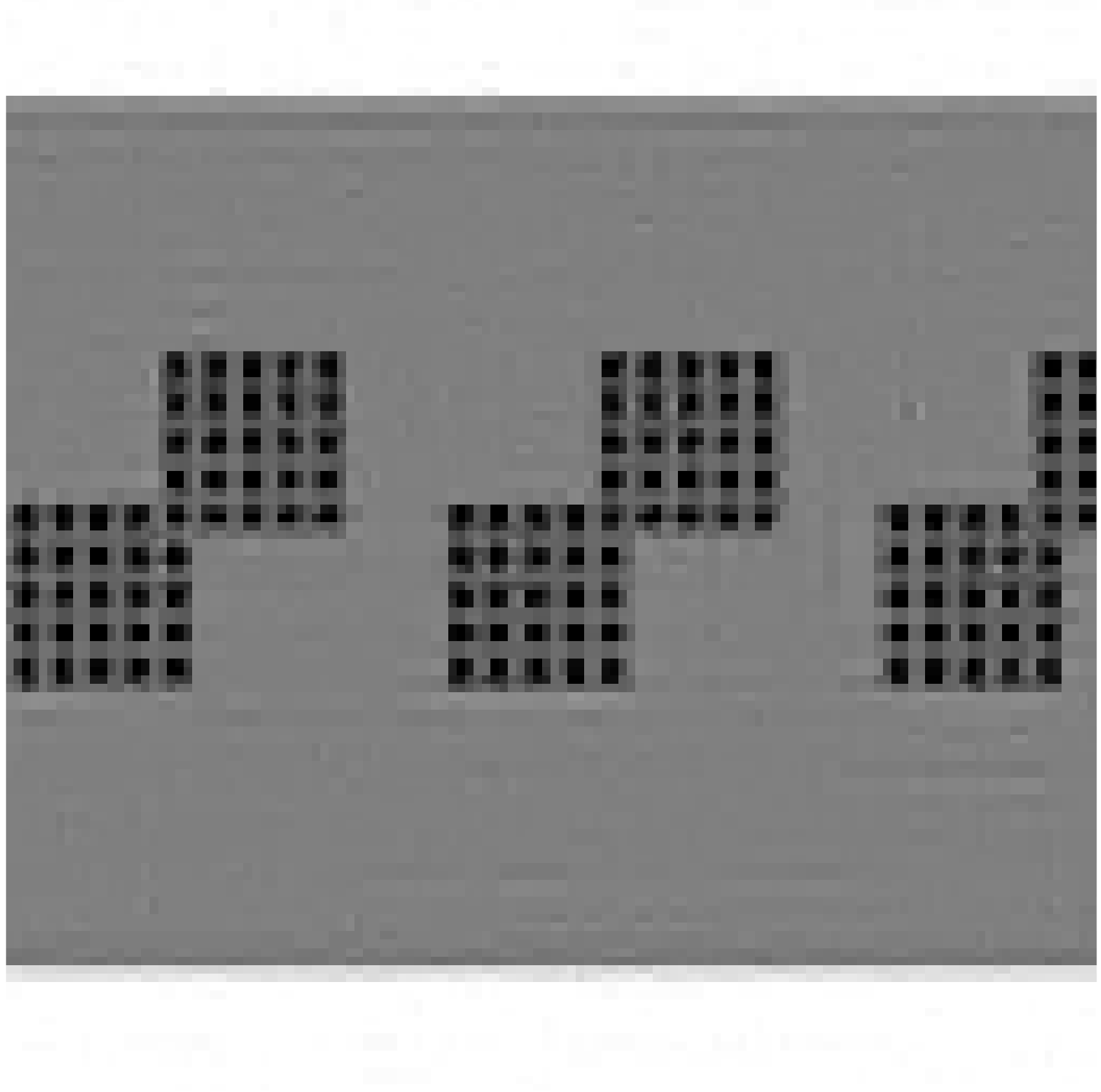} 
  \caption{\footnotesize Closeup of (c)}
  \end{center}
\end{subfigure}
\begin{subfigure}[t]{0.3\textwidth}
\begin{center}

  \end{center}
\end{subfigure}
\begin{subfigure}[t]{0.3\textwidth}
\begin{center}
\includegraphics[width=\textwidth]{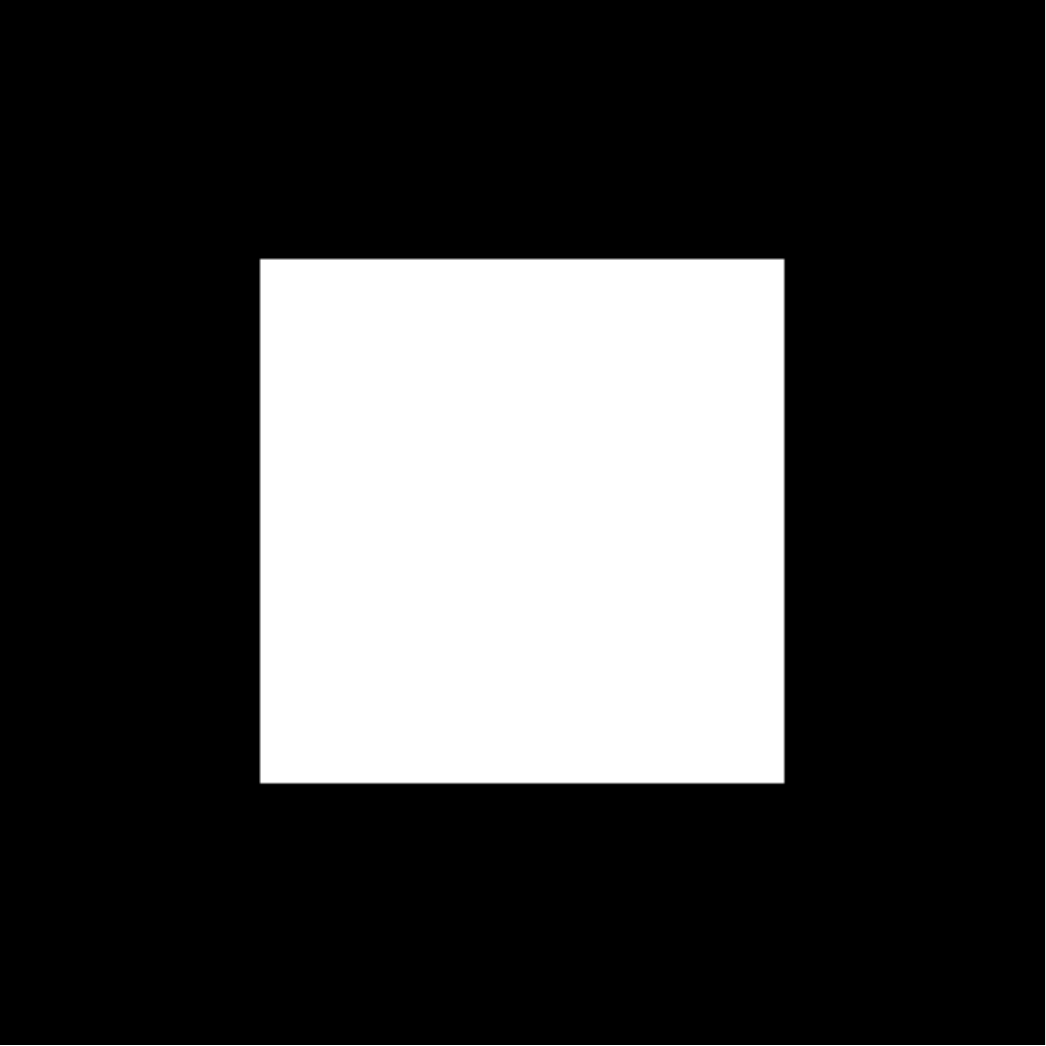}
 \caption{\footnotesize Sampling Pattern A \\ Number of Samples: 40401}
\end{center}
\end{subfigure}
\begin{subfigure}[t]{0.3\textwidth}
\begin{center}
 \includegraphics[width=\textwidth]{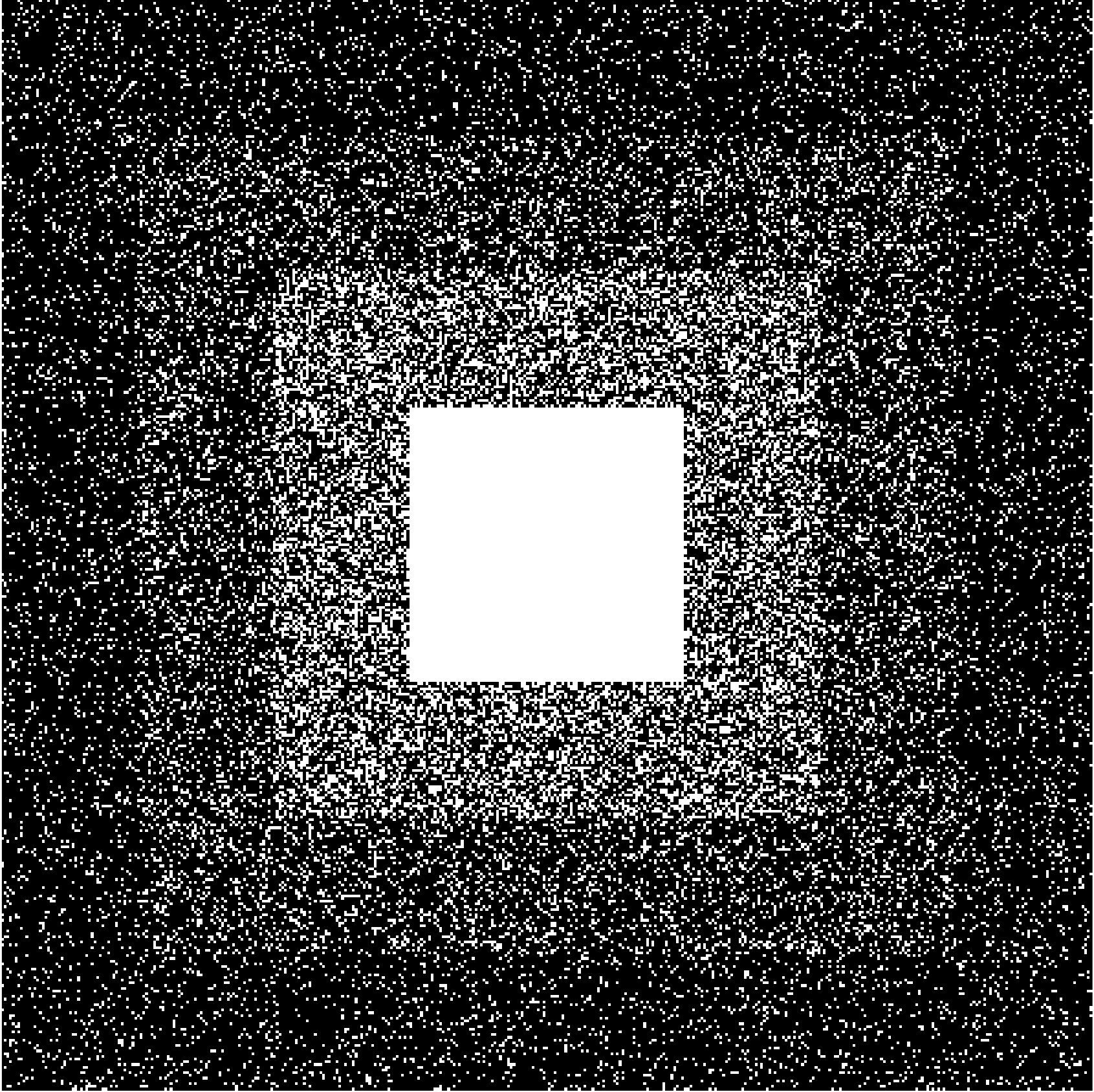} 
  \caption{\footnotesize  Sampling Pattern B \\ Number of Samples: 39341}
  \end{center}
\end{subfigure}
\end{center}
\caption{Simple Resolution Phantom Experiment. Samples are from the subset $\{-200,-199,...,199,200\}^2 \subset \bbZ^2$. Notice that the checkerboard feature are captured by the leveled sampling pattern but not by pattern (a), even though it uses fewer samples. Reconstructions are at a resolution of $2^{10} \times 2^{10}$.}
\label{resphantimages}
\end{figure}

Since the optimisation problem is infinite dimensional we cannot solve it numerically so instead we proceed as in \cite{BAACHGSCS} and truncate the problem, approximating $f$ by 
$\sum_{n=1}^R \tilde{x}_n \tau(n)$ (for $R \in \bbN$ large) where $\tilde{x} = (\tilde{x}_n)_{n=1}^R$ now solves the optimisation problem
\be{ \label{basicl1}
\min_{x \in \bbC^R} \| x \|_1  \quad \text{subject to} \quad P_\Omega U P_R x= P_\Omega \hat{f} .
}
We shall be using the SPGL1 package \cite{SPGL} to solve (\ref{basicl1}) numerically.

\subsection{Demonstrating the Benefits of Multilevel Subsampling}

We shall first demonstrate directly how subsampling in levels is beneficial in situations with asymptotic incoherence ($\mu(Q_N U) \to 0$) but poor global incoherence ($\mu(U)$ is relatively large). The image $f$ that we will attempt to reconstruct is made up of regions defined by Bezier curves with one degree of smoothness, as in \cite{GLPU}. This image is intended is model a resolution phantom\footnote{`resolution' here refers to `resolving' a signal from a MRI device.} which is often used to calibrate MRI devices \cite{resphantom}. A rasterization of this phantom is provided in image (a) of Figure \ref{resphantimages}.

We reconstruct with 2D separable Haar wavelets, ordered according to its resolution levels, from a base level of 0 up to a highest resolution level of 8. The Fourier basis is ordered by the linear consistency function $H_{2,1}$, which gives us a square leveling structure when viewed in $\bbZ^2$. We choose these orderings because we know that they are both strongly optimal for the corresponding bases, and therefore should allow reasonable degrees of subsampling when given an (asymptotically) sparse problem.

By looking at Figure \ref{resphantimages}, we observe that subsampling in levels (pattern (b)) allows to pick up features that would be otherwise impossible from a direct linear reconstruction from the first number of samples (pattern (a)) and moreover the $L^1$ error is smaller.

\subsection{Tensor vs Separable - Finding a Fair Comparison}

We would like to study how different asymptotic incoherence behaviours can impact how well one can subsample. In 2D it would be unwise to compare 2 different separable wavelet bases, since we know that they have the same optimal orderings and decay rates in 2D (see Corollary \ref{linearresults}). Therefore we are left with comparing a separable wavelet basis to a tensor basis. The incoherence decay rates for the 2D Haar cases are shown in the table below for Linear and Hyperbolic orderings of the Fourier basis $B^2_\rf$:

\begin{table}[H] \label{incoherencetable2D}
\begin{center}
\large
\begin{tabular}{c|cc}
\hline
\multicolumn{3}{c}{2D Haar Basis Incoherence Decay Rates} \\
\hline
 Ordering  & Tensor & Separable  \\
\hline
Linear      & $ N^{-1/2} $   & $ N^{-1} $     \\
   Hyperbolic       & $\log(N+1) \cdot N^{-1}$        & $\log(N+1) \cdot N^{-1} $     \\
\hline
\end{tabular}
\end{center}
\caption{ \footnotesize The decay rates for the hyperbolic case comes from Theorem \ref{TensorResultsWavelet} and Proposition \ref{Hyperbolic4Separable}. For the linear case, the separable result comes from Theorem \ref{SeparableResults} and the tensor result can be deduced from Lemma \ref{normest} applied to (\ref{zhypcrosscharacterise}), although we do not provide the details here. 
}
\end{table}

Observe that for linear orderings, there is a large discrepancy between the decay rates, however they are the same for hyperbolic orderings. Therefore, comparing separable and tensor reconstructions appears to be a good method for testing the behaviour of differing speeds of asymptotic incoherence. 

However, there is one serious problem, namely the choice of image $f$ that we would like to reconstruct. Recall from (\ref{conditions31_levels}) that the ability to subsample depends on both the coherence structure of the pair of bases and the sparsity structure of the function $f$ we are trying to reconstruct. Ideally, to isolate the effect of asymptotic incoherence we would like to choose an $f$ that has the same sparsity structure in both a tensor and separable wavelet basis. If $f$ was chosen to be the resolution phantom like before then the tensor wavelet approximation would be a poor comparison to that of the separable wavelet reconstruction (due to a poor resolution structure). Therefore we need to choose a function that we expect to reconstruct well in tensor wavelets, for example a tensor product of one dimensional functions.

\begin{figure}[!t]
\begin{center}
\begin{subfigure}[t]{0.32\textwidth}
\begin{center}
\includegraphics[width=\textwidth]{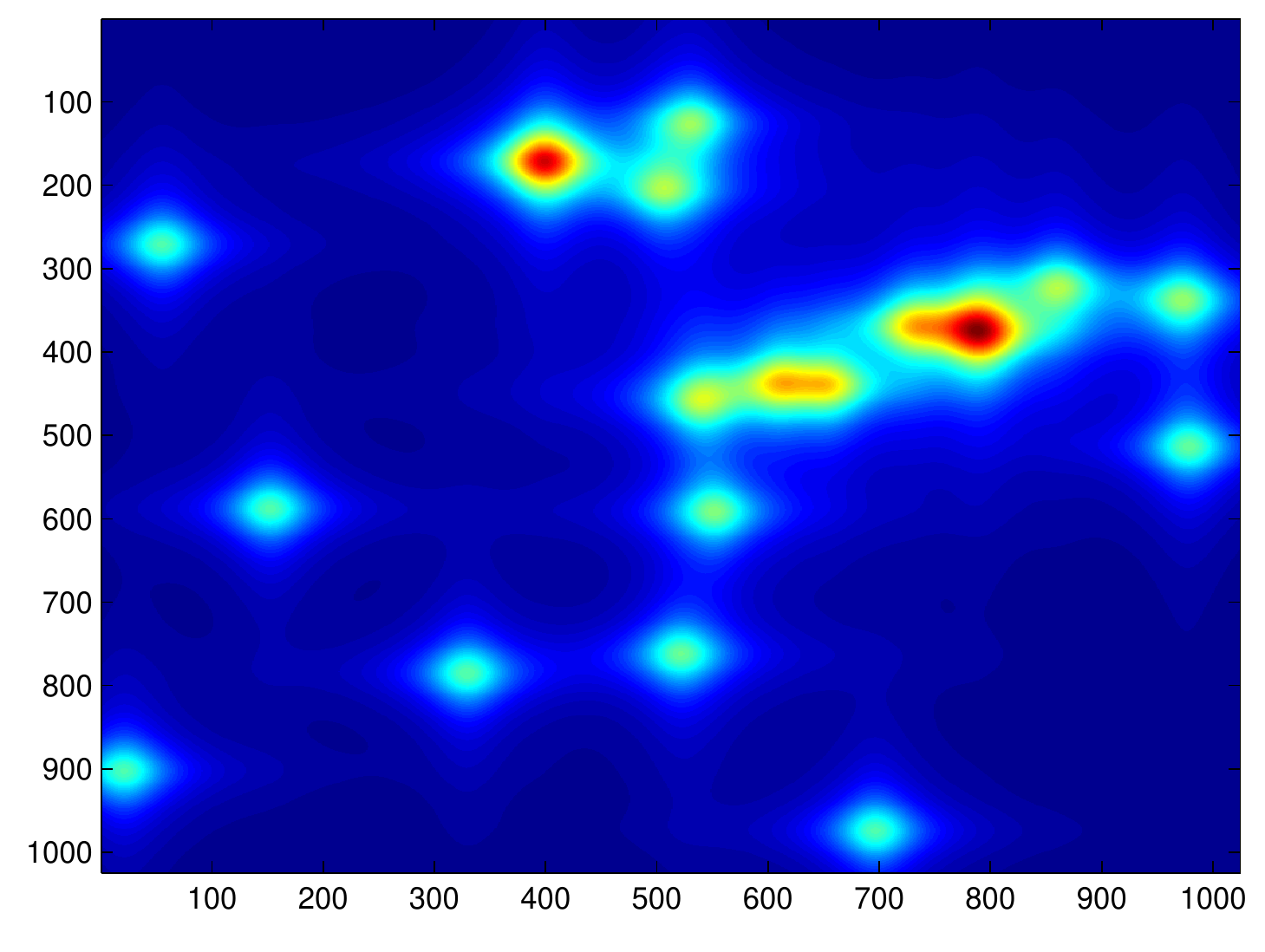}
 \caption{\footnotesize Rasterized Spectrum}
\end{center}
\end{subfigure}
\begin{subfigure}[t]{0.32\textwidth}
\begin{center}
 \includegraphics[width=\textwidth]{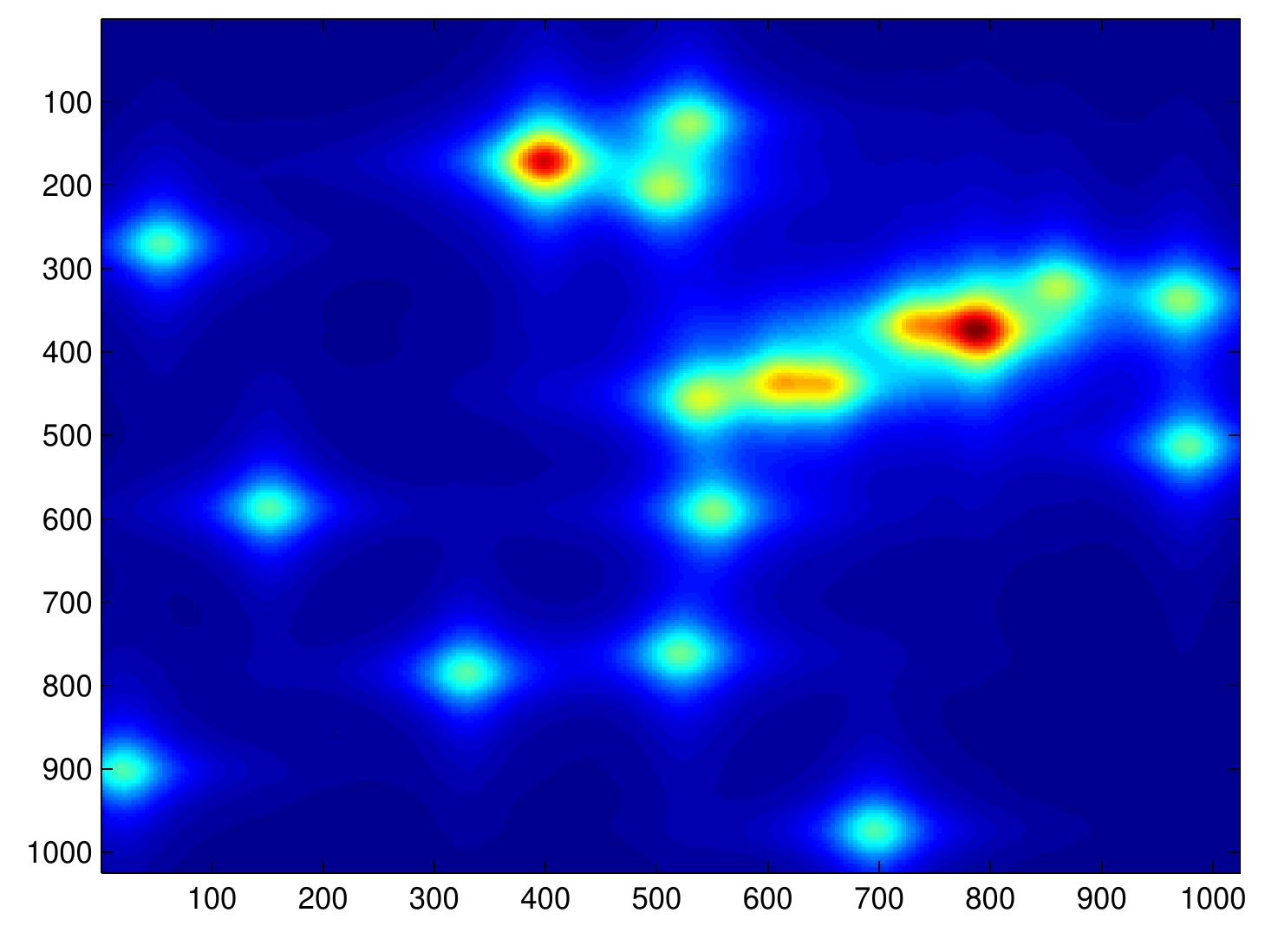} 
  \caption{\footnotesize Separable Reconstruction \\ $L^1$ error = 0.0157}
  \end{center}
\end{subfigure}
\begin{subfigure}[t]{0.32\textwidth}
\begin{center}
\includegraphics[width=\textwidth]{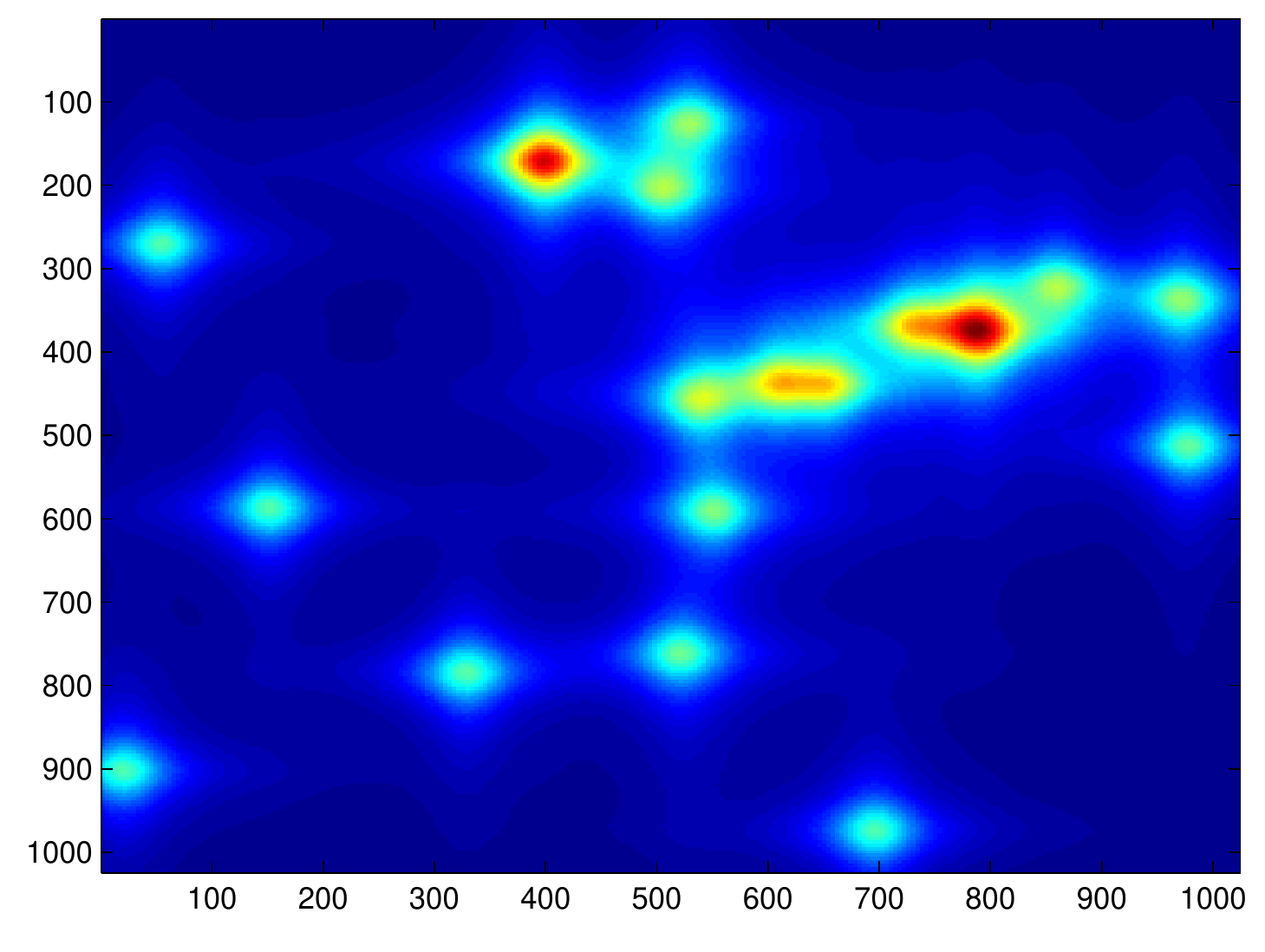} 
  \caption{\footnotesize Tensor Reconstruction \\ $L^1$ error = 0.0159}
  \end{center}
\end{subfigure}
\end{center}
\caption{Spectrum Model and `Full Sampling' Reconstructions. Reconstructions uses all samples from the subset $\{-200,-199,...,199,200\}^2 \subset \bbZ^2$. Images are at a resolution of $2^{10} \times 2^{10}$. Haar wavelets are used for tensor and separable cases. Observe that both reconstructions match the original very closely and have similar $L^1$ approximation errors.}
\label{spectrumsetup}
\end{figure}
Such an example is provided by NMR spectroscopy \cite[Eqn. (5.24)]{spindynamics}. A 2D spectrum is sometimes modelled as a product of 1D Lorentzian functions:
\be{ \label{spectrumdefine}
\begin{aligned}
f(x) & = \sum_{i=1}^r L_{2,p(i),s(i)}(x), \quad x,p(i),s(i) \in \bbR^2,
\\
L_{2,p,s}(x) & = L_{p_1,s_1}(x_1) \cdot L_{p_2,s_2}(x_2), \quad x,p,s \in \bbR^2
\\
L_{p,s} & = \frac{s}{s^2 + (x -p)^2}, \quad x,p,s \in \bbR.
\end{aligned}
}
We consider a specific spectrum $f$ of the above form. By looking at Figure \ref{spectrumsetup} we observe that, without any subsampling from the subset $\{-200,-199,...,199,200\}^2 \subset \bbZ^2$, the tensor and separable Haar wavelet reconstructions have almost identical $L^1$ errors, suggesting that this problem does not bias either reconstruction basis. We order the tensor and separable reconstruction bases using their corresponding level based orderings, which are defined in Lemma \ref{tensorwavelethyp} and Definition \ref{sepleveled} respectively. For separable wavelets we start at a base level of $J=0$ and stop at level 8 (so we truncate at the first $2^{10} \times 2^{10}$ wavelet coefficients) and for tensor wavelets we start at level
$J=0$ and stop at level 10 (when the problem was truncated at higher wavelet resolutions the improvement in reconstruction quality was negligible).

\begin{figure}[t]
\begin{center}
\begin{subfigure}[t]{0.3\textwidth}
\begin{center}
\includegraphics[width=\textwidth]{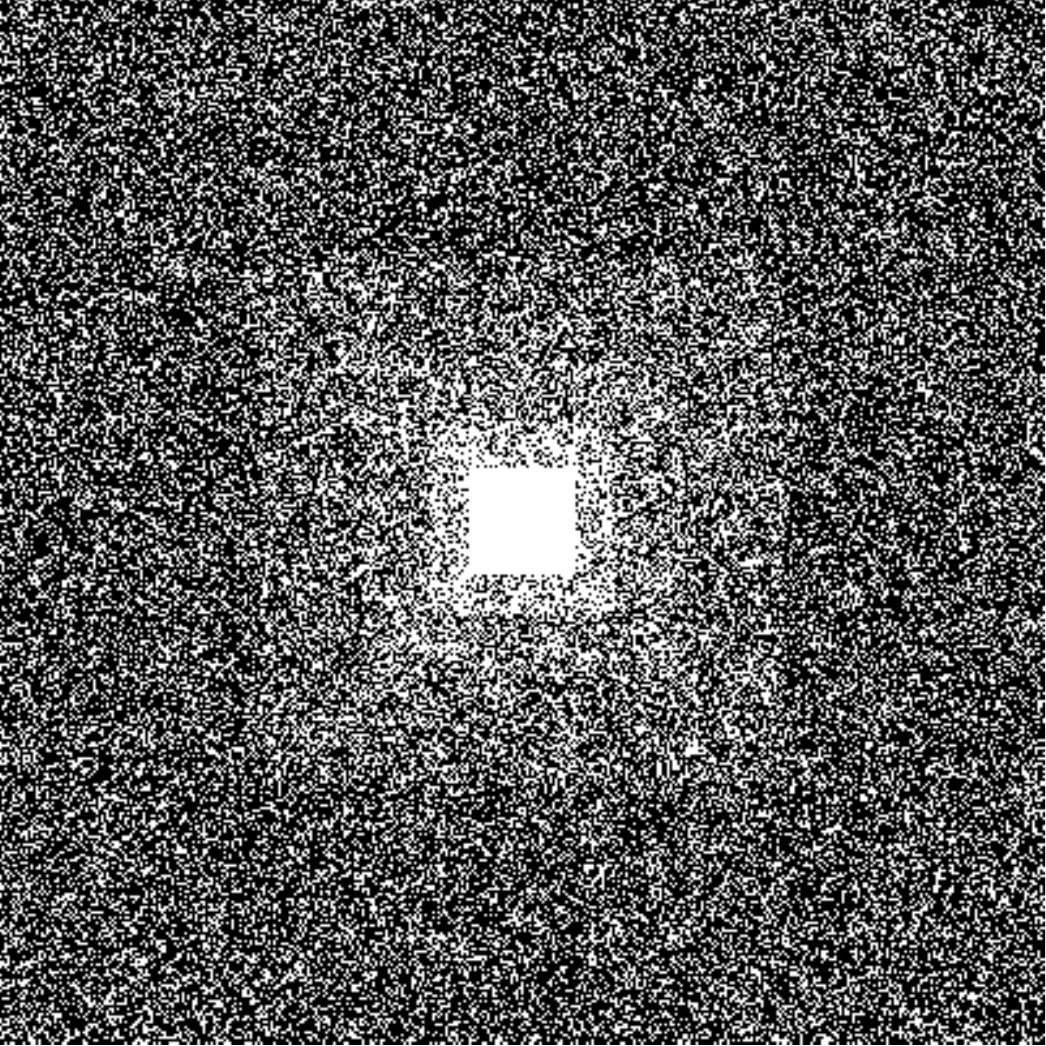}
 \caption{\footnotesize Linear Sampling Pattern}
\end{center}
\end{subfigure}
\begin{subfigure}[t]{0.3\textwidth}
\begin{center}
 \includegraphics[width=\textwidth]{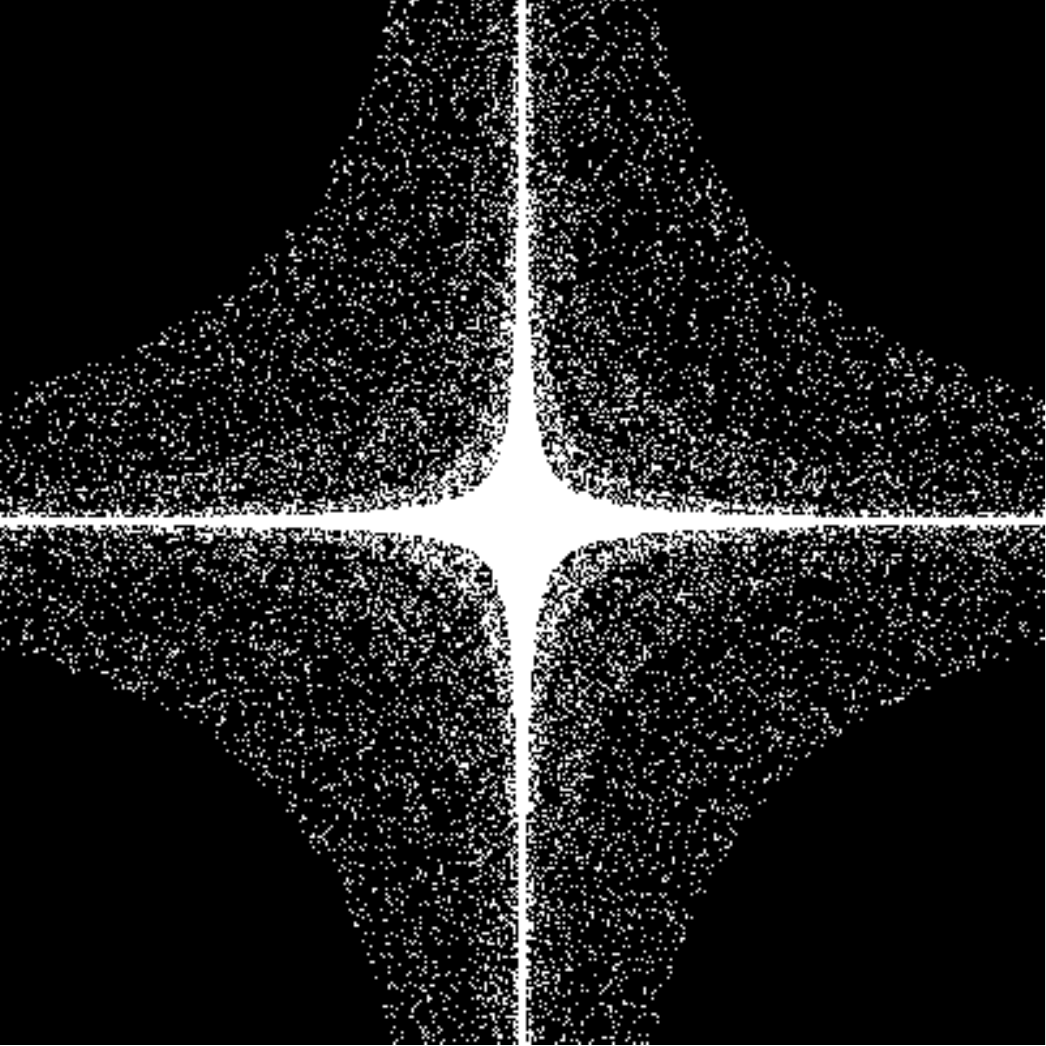} 
  \caption{Hyperbolic Sampling Pattern \\ (Boxed in)}
  \end{center}
\end{subfigure}
\end{center}
\caption{Sampling Patterns. Samples are from the subset $\{-200,-199,...,199,200\}^2 \subset \bbZ^2$. White indicates sample is taken.}
\label{spectrumsamples}
\end{figure}

\begin{figure}[!h]
\begin{center}
\begin{subfigure}[t]{0.4\textwidth}
\begin{center}
\includegraphics[width=\textwidth]{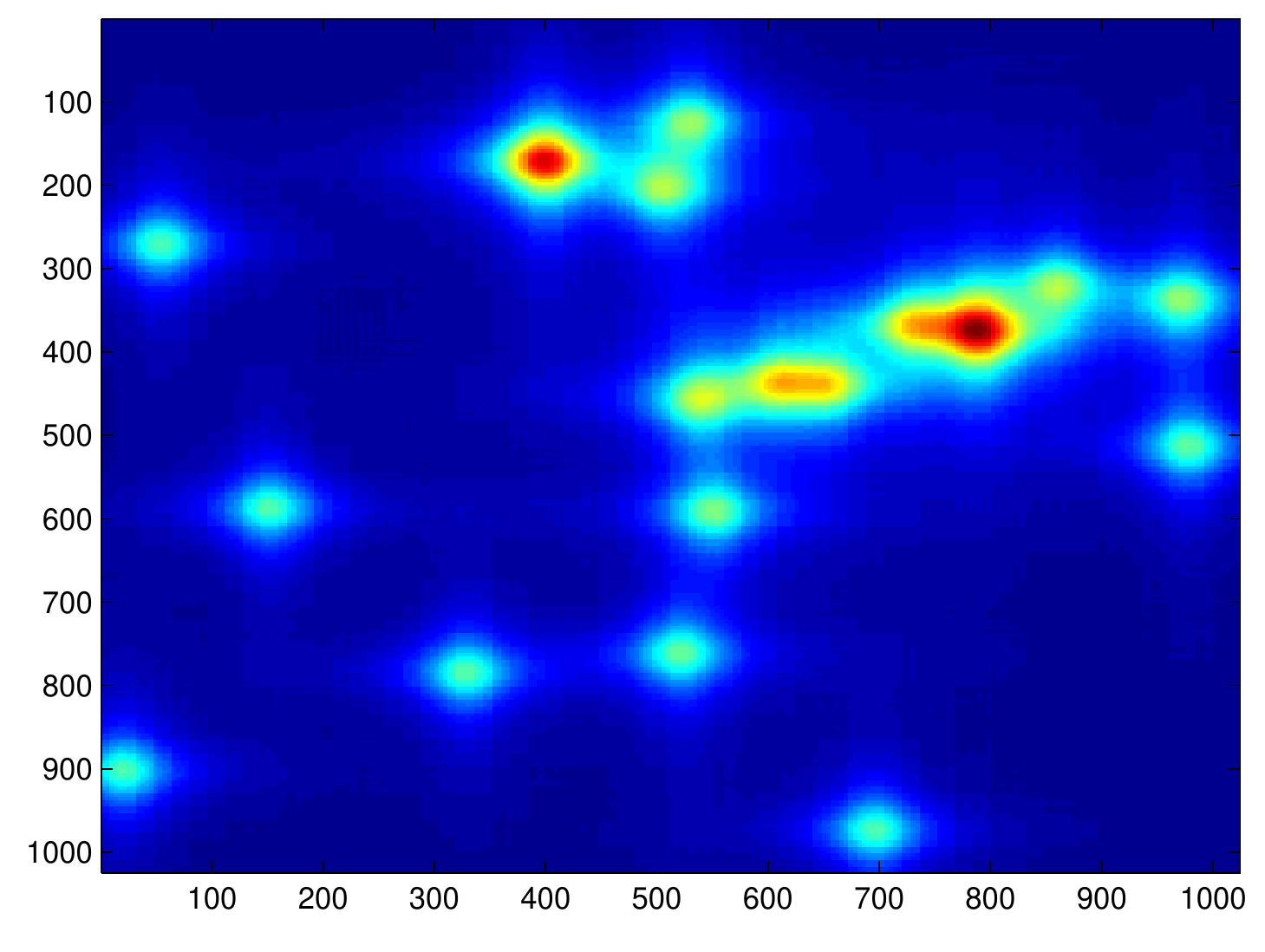}
 \caption{\footnotesize Separable Reconstruction \\ $L^1$ error = 0.0367}
\end{center}
\end{subfigure}
\begin{subfigure}[t]{0.4\textwidth}
\begin{center}
 \includegraphics[width=\textwidth]{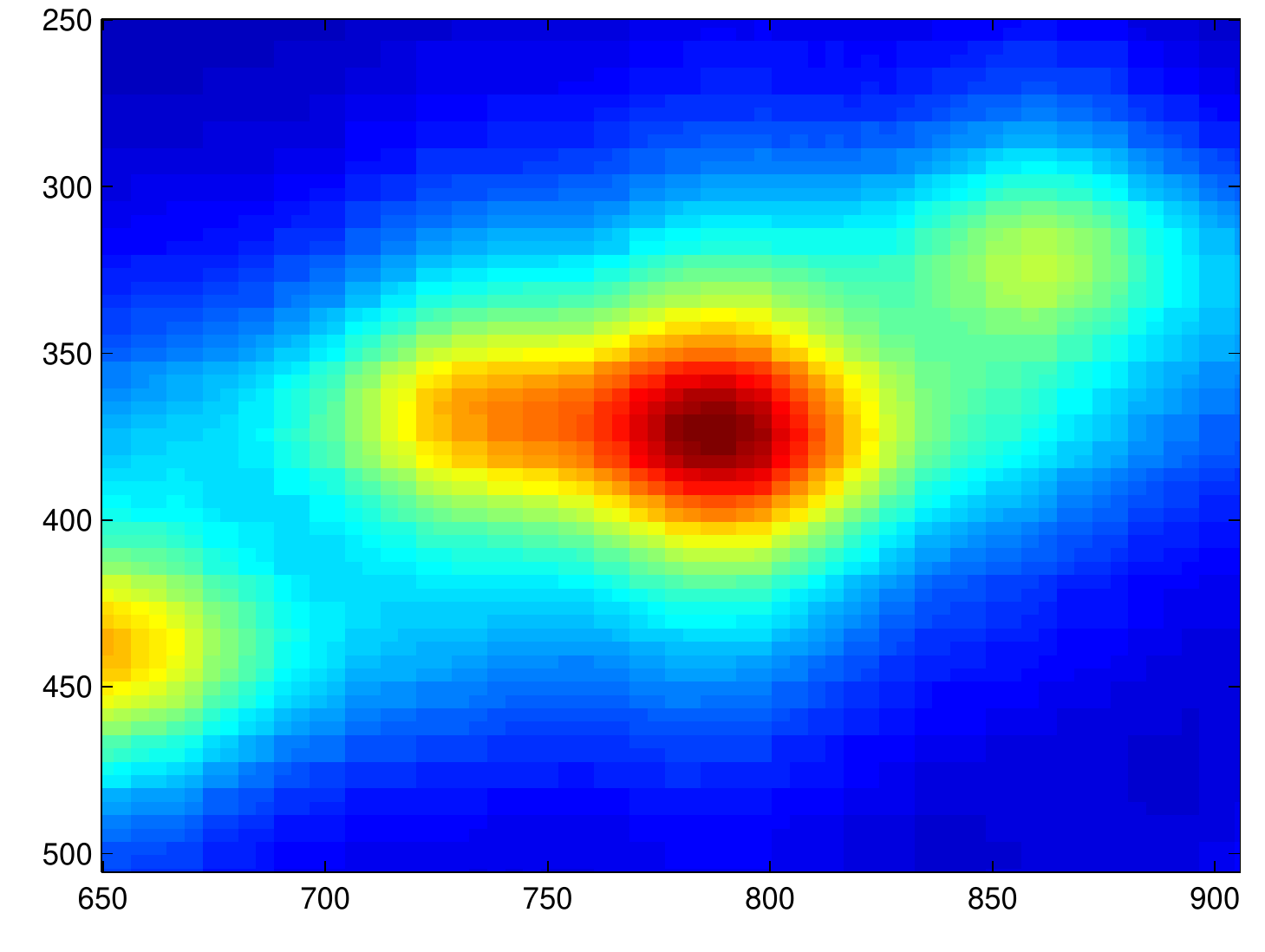} 
  \caption{\footnotesize Separable Closeup}
  \end{center}
\end{subfigure}
\begin{subfigure}[t]{0.4\textwidth}
\begin{center}
\includegraphics[width=\textwidth]{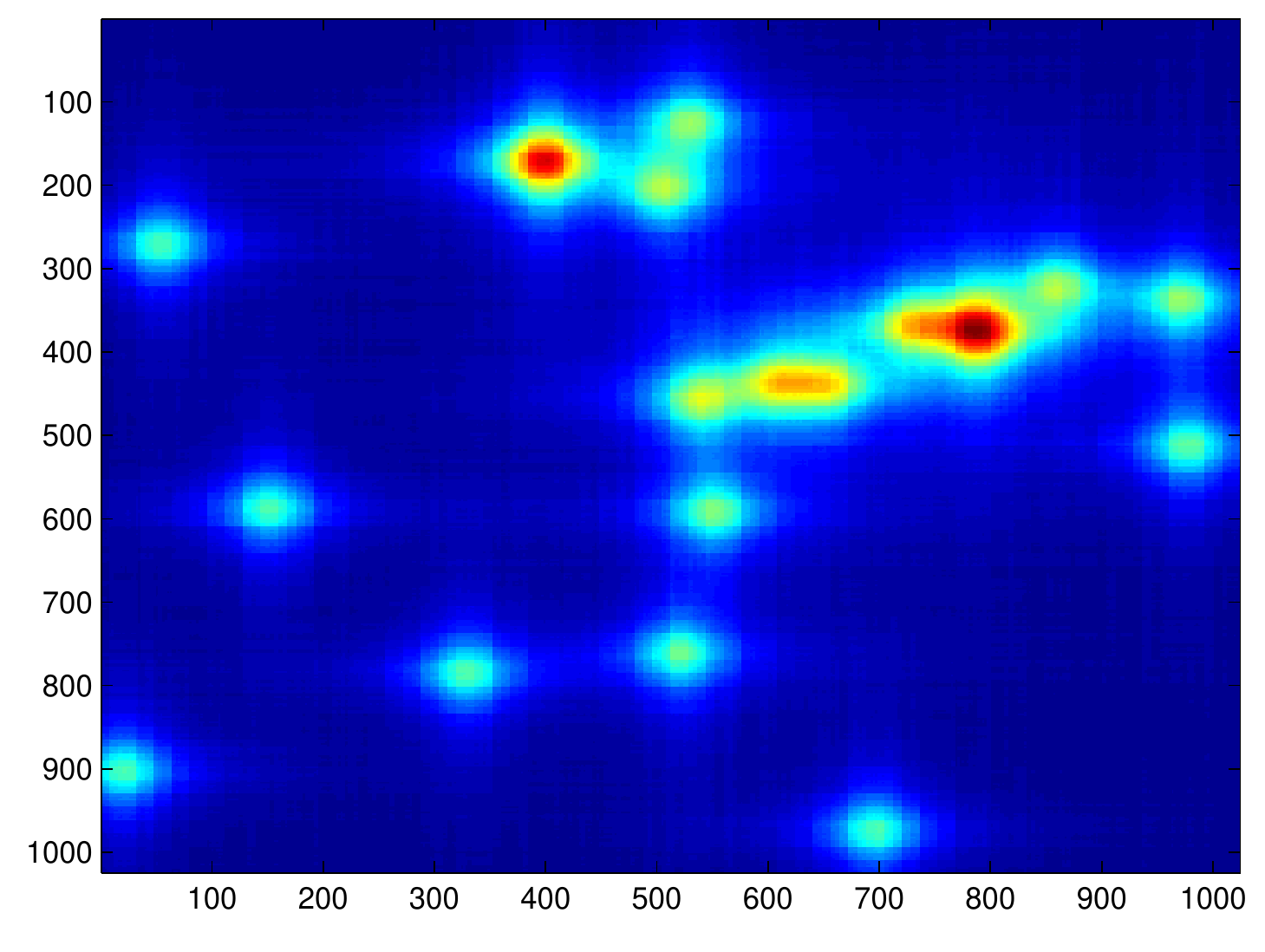}
 \caption{\footnotesize Tensor Reconstruction \\ $L^1$ error = 0.0592}
\end{center}
\end{subfigure}
\begin{subfigure}[t]{0.4\textwidth}
\begin{center}
 \includegraphics[width=\textwidth]{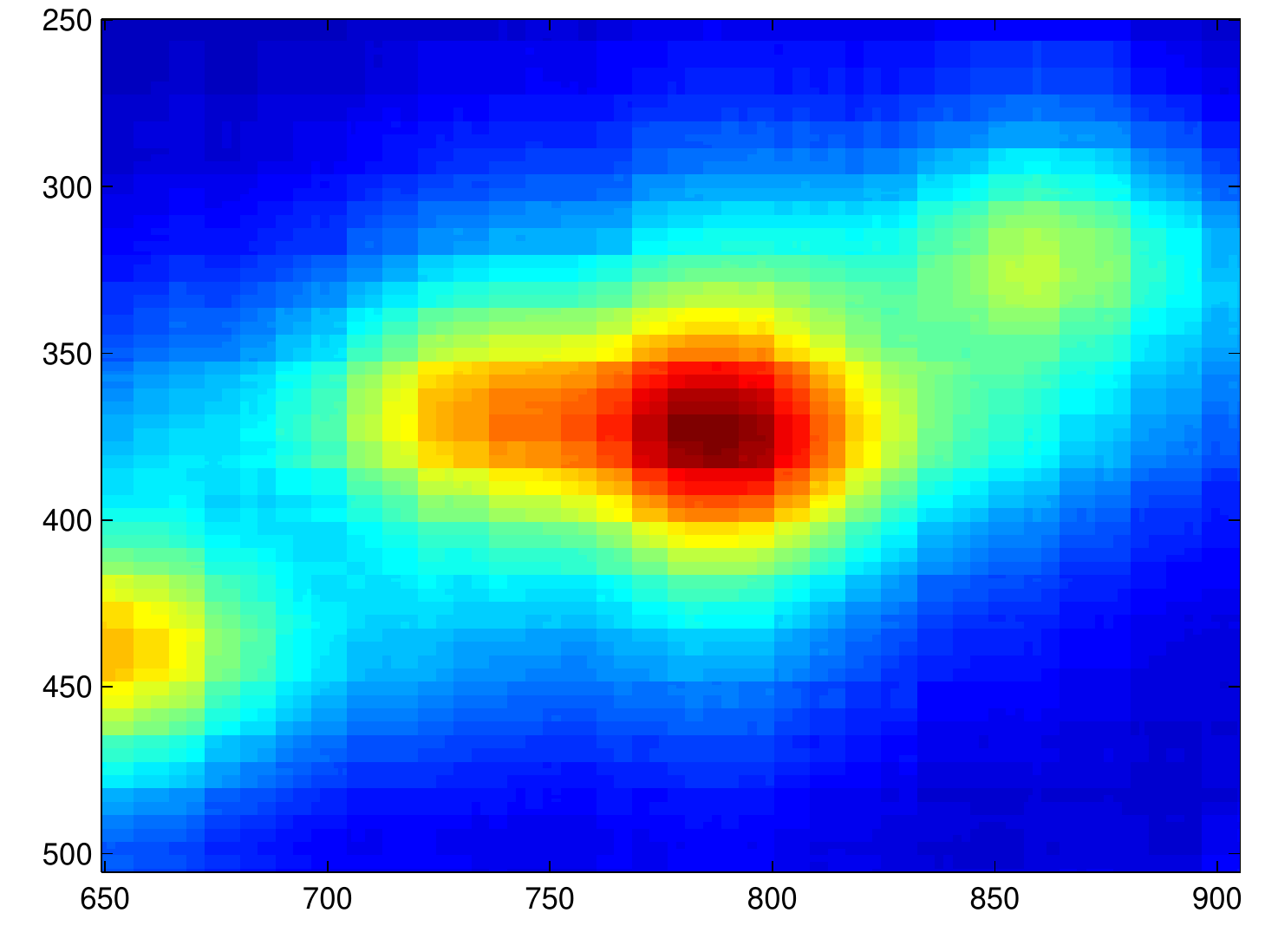} 
  \caption{\footnotesize Tensor Closeup}
  \end{center}
\end{subfigure}
\end{center}
\caption{Reconstructions from Linear Sampling Pattern} 
\label{linearrecons2d}

\end{figure}

We are now going to test how well these two bases perform under subsampling with different orderings of $B^2_\rf$. Two subampling patterns, one based on a linear ordering and another on a hyperbolic ordering, are presented in Figure \ref{spectrumsamples}. Ideally the hyperbolic subsampling pattern would not be restricted by the $\{-200,-199,...,199,200\}^2$ but this is numerically unfeasible.

Let us first consider what happens when using pattern (a) (see Figure \ref{linearrecons2d}). Notice that the separable reconstruction performs far better than the tensor reconstruction and therefore is more tolerant to subsampling with a linear ordering than the tensor case. This is unsurprising as the tensor problem suffers from noticeably large $1/\sqrt{N}$ incoherence when using a linear ordering when compared to the $1/N$ separable decay rate. 

Of course we should have fully considered the sparsity of these two problems which also factor into the ability to subsample, however $f$ was specifically chosen because it was sparse in the tensor basis and moreover we have seen that it provides a comparable reconstruction to the separable case when taking a full set of $\{-200,-199,...,199,200\}^2$ samples.
Next we observe what happens when using the pattern (b) (Figure \ref{hyprecons2d}). There is now a stark contrast to the linear case, in that both separable and tensor cases provide very similar reconstructions and furthermore the $L^1$ errors are very close. This suggests that both problems have similar susceptibility to subsampling when using hyperbolic sampling, which is reflected by their identical rates of incoherence decay with hyperbolic orderings.

\begin{figure}[h]
\begin{center}
\begin{subfigure}[t]{0.4\textwidth}
\begin{center}
\includegraphics[width=\textwidth]{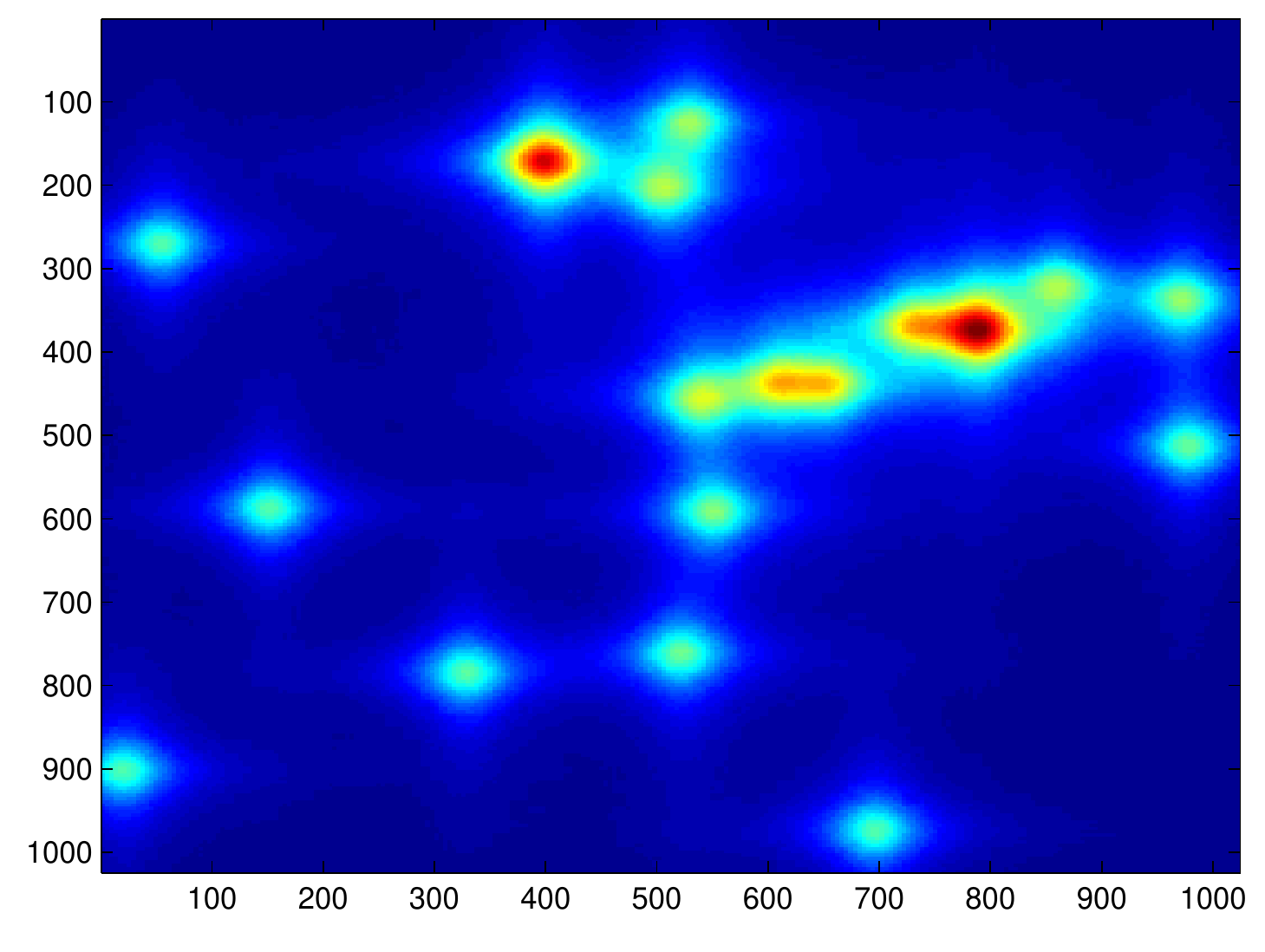}
 \caption{\footnotesize Separable Reconstruction \\ $L^1$ error = 0.0263}
\end{center}
\end{subfigure}
\begin{subfigure}[t]{0.4\textwidth}
\begin{center}
 \includegraphics[width=\textwidth]{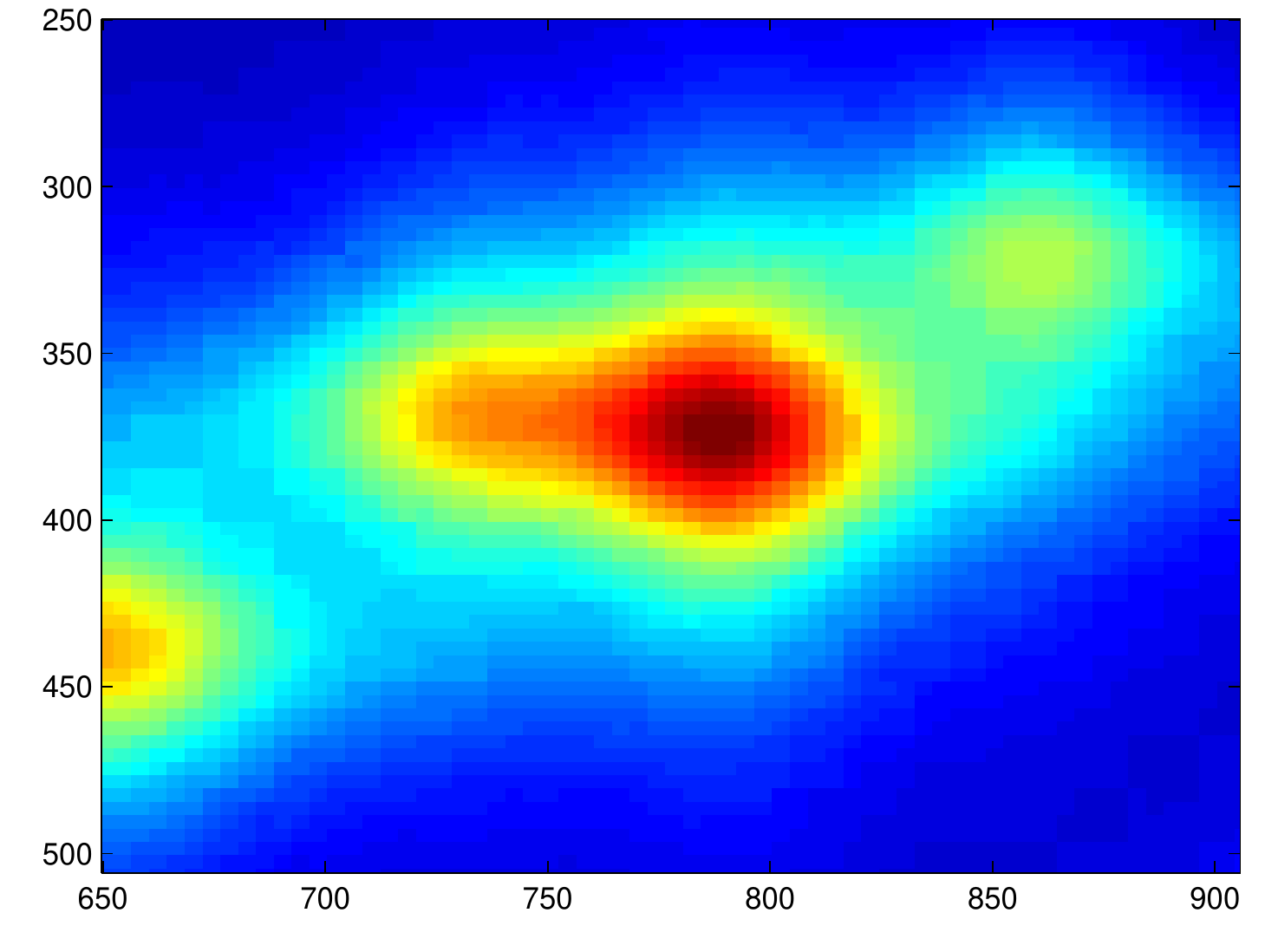} 
  \caption{\footnotesize Separable Closeup}
  \end{center}
\end{subfigure}
\begin{subfigure}[t]{0.4\textwidth}
\begin{center}
\includegraphics[width=\textwidth]{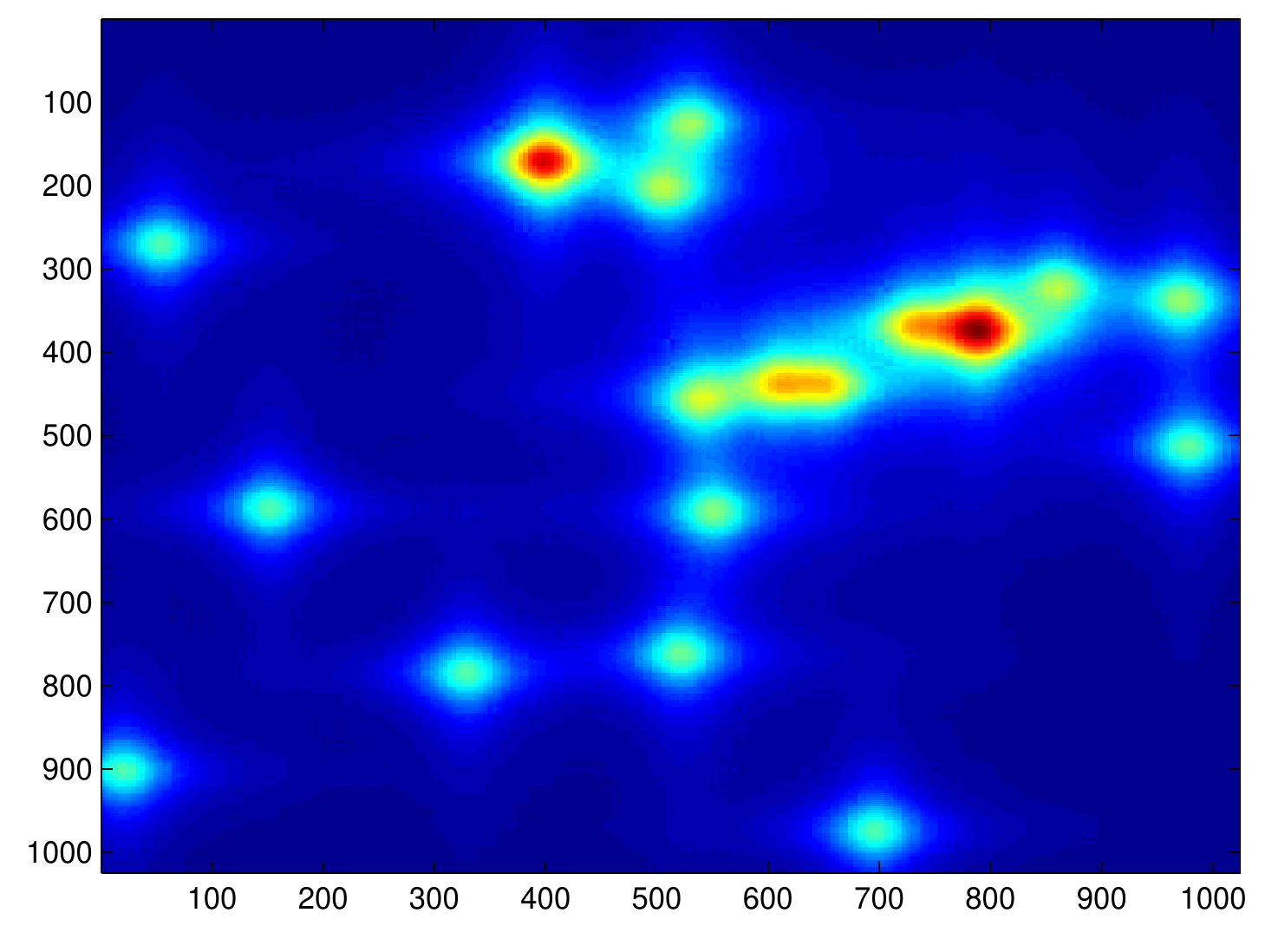}
 \caption{\footnotesize Tensor Reconstruction \\ $L^1$ error = 0.0277}
\end{center}
\end{subfigure}
\begin{subfigure}[t]{0.4\textwidth}
\begin{center}
 \includegraphics[width=\textwidth]{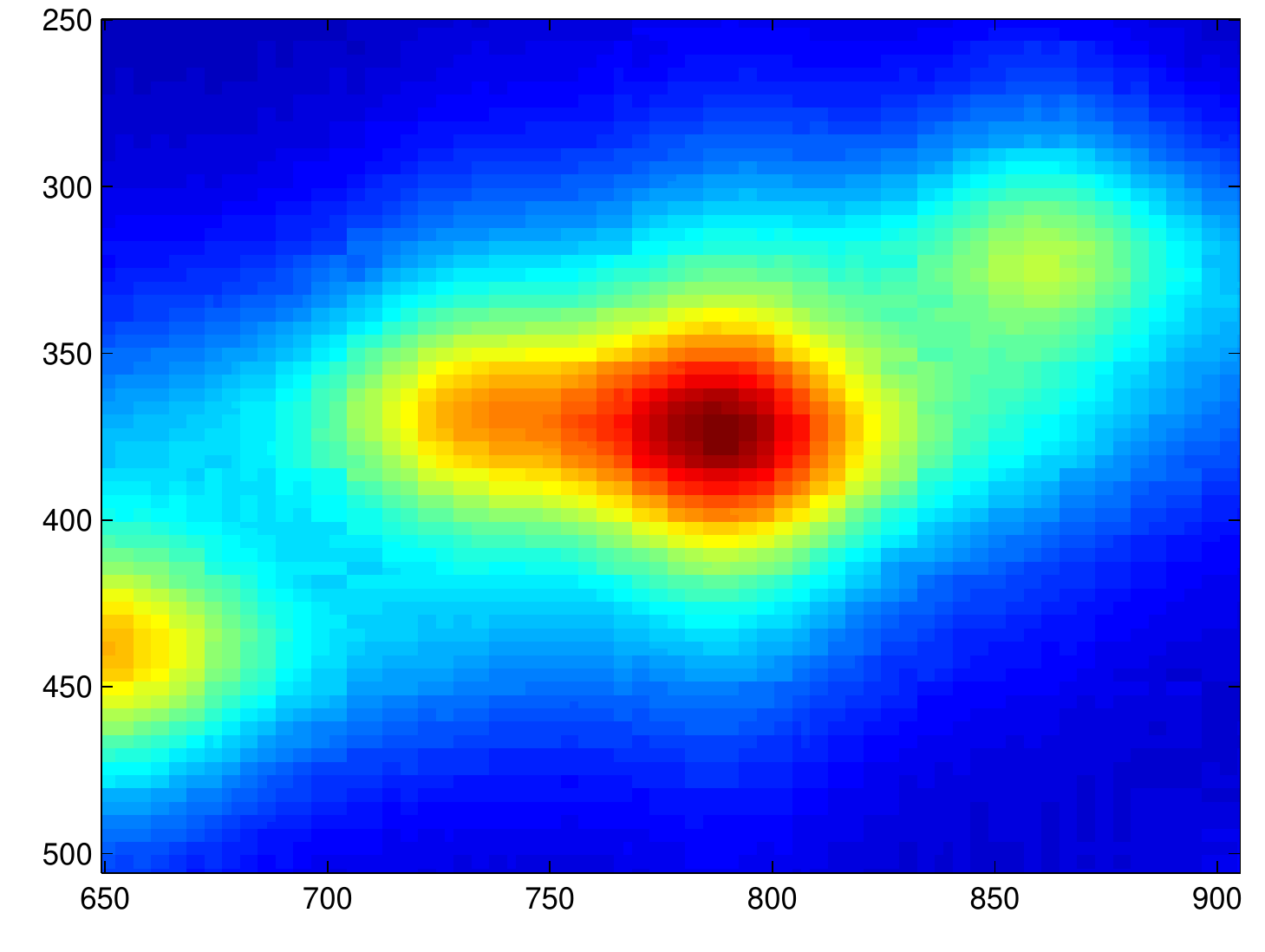} 
  \caption{\footnotesize Tensor Closeup}
  \end{center}
\end{subfigure}
\end{center}
\caption{Reconstructions from Hyperbolic Sampling Pattern} 
\label{hyprecons2d}

\end{figure}

\section{Appendix}

\textit{Proof of Proposition \ref{Hyperbolic4Separable}:} (\ref{hyperbolicbound2}) applied to part 1.) of Lemma \ref{characterisationlemma} shows that the decay of $\mu(\pi_N U)$ is bounded above by\footnote{for the definitions of $H_d, h_d$ see (\ref{hyperbolicdecay}) and (\ref{hddef}).} $F_\text{hyp}(\sigma(N))=1/H_d(\sigma(N)) \approx 1/h_d(N)$, which gives us the upper bound for $\mu(Q_N U)$ since $1/h_d(N)$ is decreasing.

For the lower bound, we focus on terms of the form $\lambda_d \circ \rho(m)=(t,...,t)$ for some $t \in \bbN$ and we set, for a fixed $q \in \bbN$
\[s=(1,...,1), \quad j:= \lceil \epsilon \log_2 t \rceil + q , \]
where we assume for now that $j \ge J$ is satisfied. This gives us
\be{ \label{specificlower}
\begin{aligned}
| \langle \Psi^s_{j,0}, \rho(m) \rangle |^2 & = \epsilon^d 2^{-dj} \prod_{i=1}^d | \mathcal{F} \psi(\epsilon 2^{-j} t) |^2
\\ & \ge  \frac{1}{2^{d(1+q)} t^d} \cdot | \mathcal{F} \psi(\epsilon 2^{-(\lceil \epsilon \log_2 t \rceil + q)} t) |^2
\\ & \ge \frac{1}{2^{d(1+q)} t^d} \cdot L_q^{2d} \quad ( \text{using (\ref{wavelower})}) .
\end{aligned}
}
Let $m$ now be arbitrary with $\prod_{i=1}^d \max( | \lambda_d \circ \rho(m)_i|,1)=M \ge 1$ and let $t = \lceil M^{1/d} \rceil +1$. Because $\rho$ corresponds to the hyperbolic cross there exists an $m'>m$ such that $\prod_{i=1}^d \max( | \lambda_d \circ \rho(m')_i|,1)=t^d$ where $\lambda_d \circ \rho(m')=(t,...,t)$. Notice that $t^d \le E(d)M$ for some constant dependent on the dimension $d$. Furthermore, (\ref{specificlower}) holds for $m=m'$ if we have that $j \ge J$, which is satisfied if $m$ is sufficiently large. Therefore we deduce by (\ref{specificlower}) that
\[
\begin{aligned}
\mu( Q_m U) \ge | \langle \Psi^s_{j,0}, \rho(m') \rangle |^2 & \ge \frac{1}{2^{d(1+q)} t^d} \cdot L_q^{2d} 
\\ & \ge \frac{1}{E 2^{d(1+q)+1} M} \cdot L_q^{2d}
\\ & = \frac{1}{E 2^{d(1+q)} \prod_{i=1}^d \max( | \lambda_d \circ \rho(m)_i|,1) } \cdot L_q^{2d}
\\ & \ge \frac{C}{E 2^{d(1+q)} h_d(m)} \cdot L^{2d}_q \quad (\text{using (\ref{hyperboliccrossZdecay}), $C>0$ some constant}). 
\end{aligned}
\]
This proves the lower bound.

\bibliographystyle{abbrv}
\bibliography{bibfirst}

\end{document}